%% file: main.tex
\documentclass[12pt,draft]{article}
\usepackage[final]{graphicx} 
 
\input{premable}

\newcommand{\st}{$s$-$t$ }
\newcommand{\SG}{\mathcal{G}}
\newcommand{\SV}{\mathcal{V}} 
\newcommand{\SE}{\mathcal{E}}

\title{Exponential Quantum Advantage for Pathfinding \\ in Regular Sunflower Graphs}
\author[1,5]{Jianqiang Li}
\author[2,3,4]{Yu Tong}
\affil[1]{Department of Computer Science, Pennsylvania State University, State College, PA, USA}
\affil[5]{Department of Computer Science, Rice University, Huston, TX, USA}
\affil[2]{Institute for Quantum Information and Matter, California Institute of Technology, Pasadena, CA, USA}
 \affil[3]{Department of Mathematics, Duke University, Durham, NC, USA}
 \affil[4]{Department of Electrical and Computer Engineering, Duke University, Durham, NC, USA}

\begin{document}

\maketitle
\begin{abstract}
    Finding problems that allow for superpolynomial quantum speedup is one of the most important tasks in quantum computation. A key challenge is identifying problem structures that can only be exploited by quantum mechanics. In this paper, we find a class of graphs that allows for exponential quantum-classical separation for the pathfinding problem with the adjacency list oracle and name it the regular sunflower graph.  We prove that, with high probability, a regular sunflower graph of degree at least $7$ is a mild expander graph, that is, the spectral gap of the graph Laplacian is at least inverse polylogarithmic in the graph size.
    
    
    We provide an efficient quantum algorithm to find an \st path in the regular sunflower graph while any classical algorithm takes exponential time. This quantum advantage is achieved by efficiently preparing a $0$-eigenstate of the adjacency matrix of the regular sunflower graph as a quantum superposition state over the vertices and this quantum state contains enough information to efficiently find an \st path in the regular sunflower graph.  

    Because the security of an isogeny-based cryptosystem depends on the hardness of finding an \st path on an expander graph \cite{Charles2009}, a quantum speedup of the pathfinding problem on an expander graph is of significance. Our result represents a step towards this goal as the first provable exponential speedup for pathfinding in a mild expander graph.

\end{abstract}

\section{Introduction} \label{sec:intro}

Quantum algorithms running on a fault-tolerant quantum computer have demonstrated an exponential advantage over classical algorithms in solving certain problems, such as simulating quantum physics~\cite{feynman1982SimQPhysWithComputers}, integer factoring, and Pell’s equation~\cite{shor1994Factoring,hallgren2007polynomial}. Since the invention of Shor's algorithm, there have been significant advances in quantum algorithms and tools. These include the adiabatic quantum computation algorithm \cite{farhi2000QCompAdiabatic}, the quantum approximate optimization algorithm \cite{farhi2014QAOA}, the quantum linear system algorithm \cite{harrow2009QLinSysSolver}, and the more recent quantum singular value transformation framework \cite{gilyen2018QSingValTransfArXiv}. Despite these advancements, only a handful of problems have been discovered that admit superpolynomial advantages \cite{aaronson2022much}.



Finding problems that enable exponential quantum speedup remains one of the biggest challenges in the area of quantum computation. The key challenge is to identify specific problem structures that can be uniquely exploited by quantum mechanics. Graph theory is one of the most promising areas for investigating these structures, as graphs of exponential size exhibit a wide range of complex structures. 


The most well-known graph structure that allows for exponential quantum-classical separations is the welded tree graph \cite{childs2003ExpSpeedupQW}. It consists of two binary trees of height $n$ connected by a cycle that alternates between the leaves of the two trees, making its size exponentially large as ($2^{n+2}-2$). With access to the graph through an adjacency list oracle, the welded tree problem starts at a root vertex $s$ to find the other root 
$t$, both of which are distinguished from the other vertices by having degree 2. The inherent structure of the welded tree graph creates conditions in which quantum algorithms can provide an exponential speedup over classical algorithms, making it an ideal candidate for exploring quantum computational advantages.


The welded tree graph structure is crucial in demonstrating exponential quantum speedups in various related problems. For example, the exponential advantage in adiabatic quantum computation with the no sign problem has been achieved by efficiently finding a marked vertex in a graph that is similar to the welded tree graph \cite{gilyen2020ExpAdvAdiabStoqQC}. The exponential quantum speedup of the graph property testing problem is obtained in the welded tree candy graph \cite{benDavid2020SymmetriesGraphPropertiesQSpeedups}.
More recently, the exponential quantum advantage of finding a marked vertex in the welded tree graph has been generalized to more general types of graphs, such as families of random hierarchy graphs \cite{babbush2023exponential}.

In this work, we study the problem of computing an \st path in a graph given the vertex $s$, the indicator function of $t$, and the adjacency list oracle (\cref{defn:classical_adjacency_list_oracle}) with access to the graph. Although quantum algorithms \cite{childs2003ExpSpeedupQW,jeffery2023multidimensional,babbush2023exponential,balasubramanian2023exponential} achieve exponential speedup in the welded tree graph for finding a marked vertex, finding an \st path in the welded tree graph is still one of the top open problems in the field of quantum query complexity \cite{aaronson2021open}. 
Recent results indicate that a certain class of quantum algorithms cannot efficiently find an \st path in the welded tree graph \cite{childs2022quantum}. However, embedding the welded tree graphs into other graph structures has achieved exponential quantum speedups for the pathfinding problem \cite{li2023exponential,li2023multidimensional}. Notably, the quantum algorithms that realize these exponential speedups \cite{li2023exponential,li2023multidimensional} are not included in the class of algorithms considered in \cite{childs2022quantum}. The adjacency list oracle of an exponential-size graph for the pathfinding problem can be instantiated. For example, in an exponentially large supersingular isogeny graph with supersingular elliptic curves as vertices and isogenies between them as edges, the adjacency list oracle can be instantiated by computing isogenies between supersingular elliptic curves.

Pathfinding problem in exponentially large expander graphs is important because the security of isogeny-based cryptography is based on the hardness of finding an \st path supersingular isogeny graphs \cite{Charles2009,eisentrager2018supersingular,wesolowski2022supersingular}, which is a class of expander graphs. 
Although a specific isogeny based scheme, namely the supersingular isogeny Diffie–Hellman key exchange (SIDH) \cite{jao2011towards}, has been broken recently \cite{castryck2023efficient}, the general problem of finding an \st path in the supersingular isogeny graph (which is an expander graph) remains unaffected. Thus, many other isogeny-based cryptosystems such as CSIDH \cite{castryck2018csidh} and SQISign \cite{deFeo2020sqisign} are immune to this attack \cite{Galbraith2022}. 
In addition, the welded tree path graph \cite{li2023exponential} and the welded tree circuit graph \cite{li2023multidimensional} that allow for an exponential speedup for the pathfinding problem are far from being an expander graph.

In this paper, we find a graph that is close to the expander graph and allows for exponential quantum speedup for the pathfinding problem. This graph deviates from the welded tree structure and we name it the \emph{regular sunflower graph}. A \emph{regular sunflower graph} consists of $n$ trees of height $m$ with $m=\Theta(n)$. The roots $s_1,s_2,\ldots, s_n$ of the $n$ trees $\mathcal{T}_1,\mathcal{T}_2,\ldots, \mathcal{T}_n$ are connected by the edges $\{s_1,s_2\} ,\{s_2,s_3\}, \ldots,\{s_{n-1},s_{n}\}, \{s_{n},s_{1}\}$ as a cycle. The leaves of trees $\mathcal{T}_i,\mathcal{T}_{i+1}$ for $i=1,2\ldots,n-1$ and $\mathcal{T}_1,\mathcal{T}_{n}$ are connected by random perfect matching such that the graph is regular.
An example of the $3$-regular graph is shown in \cref{fig:sunflowergraphd=3}.
The formal definition of a regular sunflower graph is in \cref{defn:randomG_build}.

\begin{figure}
    \centering
    \includegraphics[width=0.6\textwidth]{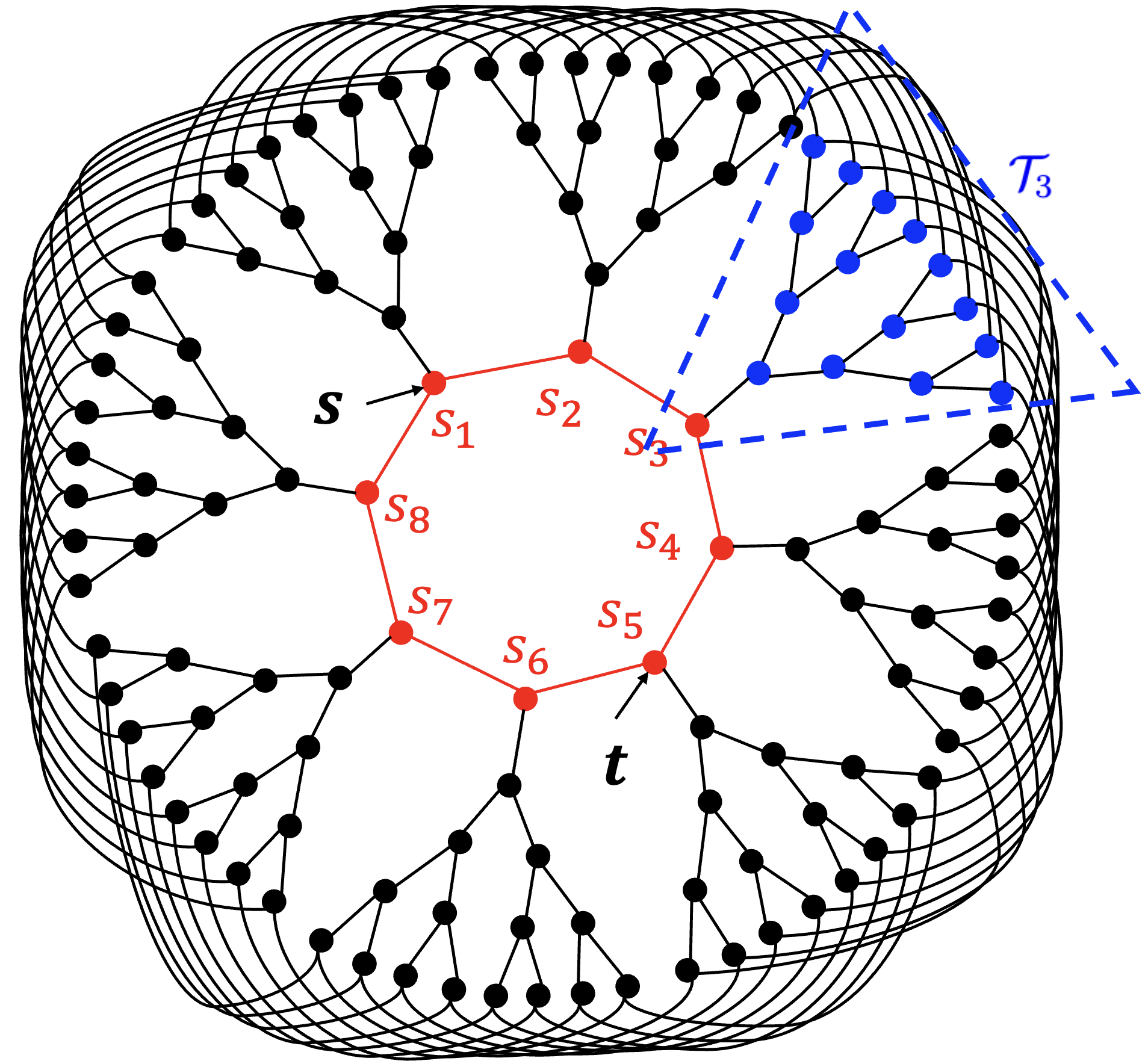}
    \caption{An example of the regular sunflower graph $\SG$ with $d=3, m=5, n=8$. The  $s$ and $t$ vertices are marked out. The tree within the dashed rectangle is the subtree $\mathcal{T}_i$ (in this instance $i=3$). The leaves of the trees $\mathcal{T}_i$ are connected via $(d-1)/2$ random perfect matchings. We note that we only prove the expansion property for $d\geq 7$, and $d=3$ is chosen here for visual clarity. }
   \label{fig:sunflowergraphd=3}
\end{figure}



Specifically, we show that, with high probability, the regular sunflower graph of degree at least $7$ is a mild expander graph. A mild expander graph is close to an expander graph in the sense that a random walk on a mild expander graph converges to a uniform distribution in $\polylog{|\SV|}$ time while the time needed for an expander graph is $O(\log |\SV|)$. The spectral gap of the adjacency matrix of the mild expander graph is $1/\poly(\log |\SV|)$ while the spectral gap of an expander graph is constant. Here, the spectral gap is the difference between the largest eigenvalue and the second largest eigenvalue of the adjacency matrix, which is equal to the spectral gap around the $0$ eigenvalue of the graph Laplacian. For comparison, the spectral gap of the adjacency matrix of the welded tree path graph \cite{li2023exponential} and the welded tree circuit graph \cite{li2023multidimensional} is $1/\poly(|V|)$, as removing a constant number of edges would disconnect both graphs.

The pathfinding problem in this paper is defined as follows.

\begin{problem}[Pathfinding Problem] \label{pro:pathfinding}
    Given the adjacency list oracle with access to the regular sunflower graph, the starting vertex $s=s_1$, and an indicator function $f_t(v)$ such that:
    \begin{equation}
    f_t(v)=
    \begin{cases}
      1, & \text{if}\ v = t \\
      0, & \text{otherwise}
    \end{cases}
  \end{equation}
the pathfinding problem is to compute an \st path. We fix $t$ as $t=s_{n/2+1}$ in this paper.
\end{problem}


Although many quantum algorithms achieve both polynomial and exponential speedups for the pathfinding problem in various graph structures, it remains unclear how to adapt these algorithms to demonstrate an exponential quantum advantage for the pathfinding problem in expander graphs or mild expander graphs. For example, the amplitude amplification technique achieves polynomial speedup for various graph problems, including the pathfinding problem in a general graph with oracle access \cite{durr2006quantum}. Quantum walk-based algorithms demonstrate polynomial quantum advantages in special graphs such as regular tree graphs, chains of star graphs \cite{reitzner2017finding,hillery2021finding,koch2018finding}, and supersingular isogeny graphs \cite{jaques2019quantum,tani2009claw}. The speedups resulting from these works are polynomial and therefore fall short of the exponential speedup that we aim for.

Recently, two distinct quantum algorithms leverage the quantum electrical flow state to show quantum advantages for finding an \st path in graphs with a unique \st path \cite{jeffery2023quantum} and graphs composed of welded Bethe trees \cite{Seaneletr}. \cite{jeffery2023quantum} also presents a quantum algorithm that uses a quantum subroutine for path detection to find an \st path. This quantum algorithm is faster than the quantum algorithm in \cite{durr2006quantum} for graphs with all \st paths being short. 
Again, both works mentioned above achieved polynomial, rather than exponential, speedups.

The first quantum algorithm that demonstrated an exponential quantum speedup for the pathfinding problem on the welded tree path graph \cite{li2023exponential} relies heavily on the graph being non-regular because vertices with varying degrees are used as markers to help find the path. Therefore, it cannot be extended to find an \st path in regular expander graphs. The second such quantum algorithm uses a multidimensional electrical network framework to generate a quantum superposition state on the edges of the welded tree circuit graph \cite{li2023multidimensional}, which has a large overlap with an \st path. Although this multidimensional electrical network framework has the potential to produce an exponential quantum speedup for the pathfinding problem, analyzing the multidimensional electrical network in the regular sunflower graph is challenging to demonstrate an exponential quantum advantage.

In this paper, we develop a quantum algorithm for the pathfinding problem that achieves an exponential speedup over classical algorithms in the regular sunflower graph.
Our quantum algorithm is complementary to the existing two quantum algorithms that achieve exponential speedups for the pathfinding problem. We compare the strategies used in these three algorithms below:

\begin{itemize}
    \item Problem Reduction \cite{li2023exponential}: Reduce the \st pathfinding problem in the welded tree path graph to vertex finding problem in the welded tree graph.

    \item Edge Superposition \cite{li2023multidimensional}:  Use the multidimensional electrical network framework to generate a quantum superposition state over edges. This quantum state has significant overlap with edges along an \st path in the welded tree circuit graph.
    
    \item Vertex Superposition {\textbf{(This work)}}: Use quantum eigenstate filtering technique to prepare a $0$-eigenstate of the adjacency matrix as a quantum superposition state over vertices. This quantum state has a large overlap with the vertices of an \st path in regular sunflower graphs. 
    
\end{itemize}

Specifically, we combine the quantum singular value transformation (QSVT) framework \cite{gilyen2018QSingValTransfArXiv} with the minimax filtering polynomial \cite{lin2019OptimalQEigenstateFiltering} to prepare a 0-eigenstate of the adjacency matrix $A$ as a quantum state, which
is a superposition over the vertices of the regular sunflower graph and has a large overlap with the vertices of an \st path. 
The existence of a 0-eigenvector of the adjacency matrix for certain special tree graphs has previously been used to demonstrate quantum advantages in formula evaluation \cite{ambainis2010any,farhi2007quantum}. More recently, the 0-eigenvector of the adjacency matrix for some random hierarchical graphs has been employed to find a marked vertex \cite{balasubramanian2023exponential}. To the best of our knowledge, this paper is the first to use the information contained in the $0$-eigenspace of an adjacency matrix to show an exponential quantum advantage for the pathfinding problem. 

While our results build upon prior work such as \cite{childs2002quantum,balasubramanian2023exponential}, which use Hamiltonian simulation for a random time followed by measurement to locate a special vertex, directly applying their techniques to our setting is nontrivial for several reasons.
First, the method in \cite{childs2003ExpSpeedupQW} requires full knowledge of the eigenspaces and eigenvalues of the underlying Hamiltonian (as stated in \cite[Lemma 1]{childs2003ExpSpeedupQW}), along with the ability to compute overlaps between these eigenstates and individual basis states. In our setting, this information is not readily accessible, making the approach cumbersome. Second, the approach in \cite{balasubramanian2023exponential} requires that the zero eigenspace is one-dimensional (see the end of Section 2 of \cite{balasubramanian2023exponential}), which simplifies the analysis of post-measurement behavior. In contrast, our setting involves a two-dimensional zero eigenspace within the relevant invariant subspace (\cref{cor:spectral_properties}), which complicates the dynamics under random-time evolution and introduces ambiguity in the resulting state. This makes it difficult to argue that the measurement yields sufficient support on the desired s-t path.



\subsection{Summary of the Main Results}

 We construct a family of $d$-regular ($d\geq 3$ is an odd integer) regular sunflower graphs $\SG=(\SV,\SE)$ (See Figure~\ref{fig:sunflowergraphd=3} for an example), with $|\SV|=\exp(\Or(n))$, and prove the following:
\begin{enumerate}
    \item With probability at least $1-\exp(-\Omega(n))$, the graph $\SG$ is a mild expander graph with spectral gap at least $1/\poly(n)$ when $d\geq 7$ (Theorem~\ref{thm:graph_expansion}).
    \item A quantum algorithm can find the \st path for a specific pair of vertices $s$ and $t$ (given $s$ but not $t$ at the beginning) in time $\poly(n)$ (Theorem~\ref{thm:the_algorithm}).
    \item Any classical algorithm requires $\exp(\Omega(n))$ time to find $t$ starting from $s$ with large probability (Theorem~\ref{thm:classical_lower_bound}).
\end{enumerate}

\subsection{Overview of the algorithm}


Our quantum algorithm, described in detail in \cref{alg:mildexpanderfinding}, Section~\ref{sec:Algorithm}, first prepares a $0$-eigenstate of the adjacency matrix $A$ of the regular sunflower graph as a quantum state. Then we show that having access to this quantum state enables us to efficiently find the \st path in the regular sunflower graph.

There are two standard quantum algorithmic tools that can be used to obtain a $0$-eigenstate of the adjacency matrix $A$. One is quantum phase estimation; the other is using quantum singular value transformation (QSVT) \cite{gilyen2018QSingValTransfArXiv} together with a filtering polynomial \cite{lin2019OptimalQEigenstateFiltering}. 
The input of the quantum phase estimation algorithm is $U=e^{-iAt}$ and $\ket{s}$. By post-selecting the estimated phase, one can approximately prepare a 0-eigenstate as a quantum state \cite[Proposition 3]{ge2017FasterGroundStatePrep}. In this work, we will take the second approach and prepare the eigenstate using QSVT, which has the advantage of achieving a better dependence on the approximation error \cite[Theorem 6]{tong2022quantum}.

In using QSVT, we will first need to construct a \emph{block encoding} of the adjacency matrix $A$, which is a unitary circuit encoding $A$ as part of the unitary matrix \cite{gilyen2018QSingValTransfArXiv}. A precise definition of the block encoding is given in \cref{defn:block_encoding}, and in Section~\ref{sec:the_oracle_models} we show that the adjacency list oracle can be used to construct a block encoding with constant overhead. 

With the block encoding of the adjacency matrix $A$, which is a Hermitian matrix, QSVT allows us to implement a matrix polynomial in the following sense: starting from a state $\ket{\phi}$, it yields a quantum circuit $\mathcal{V}_{\mathrm{circ}}$ such that
\[
\mathcal{V}_{\mathrm{circ}} \ket{0}_{\gamma}\ket{\phi}_{\beta} =\ket{0}_{\gamma}  f(A/\alpha)  \ket{\phi}_{\beta} + \ket{\perp}_{\gamma\beta},
\]
where $(\bra{0}_{\gamma}\otimes I_{\beta})\ket{\perp}_{\gamma\beta}=0$.
Here $f(x)$ is a degree $2\ell$ even polynomial satisfying $|f(x)|\leq 1$ for all $-1\leq x\leq 1$. $\alpha$ is a normalization factor from block encoding that ensures $\|A/\alpha\|\leq 1$. In our setting $\alpha=d^2$, as will be shown in \cref{lem:block_encoding_adjacency_matrix_new}. After applying $\mathcal{V}_{\mathrm{circ}}$ we can measure the ancilla qubits in register $\gamma$, and from the above equation we can see that, upon obtaining the all-0 state, we will have the state $f(A/\alpha)\ket{\phi}$ in the register $\beta$. Therefore the matrix polynomial $f(A/\alpha)$ can be implemented with QSVT, which succeeds with a certain probability.
In our algorithm, we will let $\ket{\phi}=\ket{s}$ where $s$ is the starting vertex given in \cref{pro:pathfinding}. Therefore, the state we prepare through QSVT is $f(A/\alpha)\ket{s}$.

\paragraph{From adjacency matrix $A$ to effective Hamiltonian $H$}
A very important fact about the adjacency matrix of the regular sunflower graph, similar to the welded tree graphs in \cite{childs2003ExpSpeedupQW} and the random hierarchical graphs in \cite{balasubramanian2023exponential}, is that $\ket{s}$ lives in a low-dimensional invariant subspace of $A$. In our scenario, even though $A$ acts on a Hilbert space of dimension $\exp(\Theta(n))$ ($n$ is the number of trees we use to construct the regular sunflower graph), this invariant subspace is only $\Or(n^2)$-dimensional (assuming $m=\Theta(n)$ as mentioned above). This invariant subspace, which we call the \emph{symmetric subspace} and denote by $\mathcal{S}$, is defined in \cref{def:SubspaceS} (we prove that it is indeed invariant in \cref{lem:invariant_subspace}). $\mathcal{S}$ is spanned by a set of $\Or(n^2)$ supervertex states defined in \cref{eq:supervertex_state_defn}, which include $\ket{s}$ as one of them.

Because $\mathcal{S}$ is an invariant subspace of $A$, we can consider a restriction of $A$ to $\mathcal{S}$, and denote it by $H:\mathcal{S}\to\mathcal{S}$. $H$ is therefore a linear operator acting on an $\Or(n^2)$-dimensional Hilbert space, and thus much easier to study than $A$. We call $H$ the \emph{effective Hamiltonian} because if we consider the Hermitian matrix $A$ to be a Hamiltonian, then the dynamics it generates will be completely determined by $H$ within the subspace $\mathcal{S}$. A rigorous definition of the effective Hamiltonian of $H$ can be found in \cref{defn:effective_hamiltonian}.

Because $H$ is the restriction of $A$ to $\mathcal{S}$, one can readily prove that for any polynomial $p(x)$, $p(A)\ket{\phi}=p(H)\ket{\phi}$ for any $\ket{\phi}\in\mathcal{S}$ (\cref{lem:effective_ham_polynomial}). In particular, because $\ket{s}\in\mathcal{S}$, we have
\[
f(A/\alpha)\ket{s} = f(H/\alpha)\ket{s}.
\]
In other words, QSVT gives us the ability to implement matrix functions of the effective Hamiltonian $H$.

\paragraph{Prepare a $0$-eigenstate of the adjacency matrix $A$ as a quantum state}

Recall that our goal is to prepare a 0-eigenstate of $A$. We will in fact prepare a 0-eigenstate of $H$, which is guaranteed to also be a 0-eigenstate of $A$ because $H$ is the restriction of $A$ to an invariant subspace. This is done by choosing $f$ so that it maps all non-zero eigenvalues of $H/\alpha$ to a small number (in fact exponentially small in $\ell$), and thereby filtering out the unwanted eigenstates. At the same time, we want $f(0)=1$ so that the 0-eigenstate is preserved. Because the cost of QSVT depends on $2\ell$, the degree of $f$, we also want the degree to be as small as possible. In \cite{lin2019OptimalQEigenstateFiltering}, a minimax filtering polynomial was constructed with these considerations in mind, which we will use here. The specific form of the polynomial $f(x):=R_{\ell}(x;\Delta/\alpha)$ is given in \eqref{eq:minimax_filtering_polynomial_appendix}. This polynomial allows us to approximate the 0-eigenspace projection operator of $H$, which we denote by $\Pi_0$:
\[
\|f(H/\alpha)-\Pi_0\|\leq 2e^{-\sqrt{2}\ell \Delta/\alpha},
\]
where $\Delta$ is the spectral gap around the $0$ eigenvalue of $H$, which is the smallest absolute value of the nonzero eigenvalues of $H$. The above bound of the approximation error is proved in \cref{lem:subspace_eigenstate_filtering_exact}. Applying $f(A/\alpha)$ to the initial state $\ket{s}$ therefore approximately prepares a 0-eigenstate of $H$:
\begin{equation}
\label{eq:state_prepared_through_qsvt}
    f(A/\alpha)\ket{s} = f(H/\alpha)\ket{s} \approx \Pi_0\ket{s}.
\end{equation}
Here $\Pi_0\ket{s}$ is an (unnormalized) $0$-eigenstate of $H$, and also is an eigenstate of $A$.

The above eigenstate filtering technique is described in detail in \cref{thm:robust_subspace_eigenstate_filtering}, with proof provided in Appendix~\ref{sec:robust_subspace_eigenstate_filtering}. Besides what was discussed above, we also take into account the error coming from inexact block encoding in \cref{thm:robust_subspace_eigenstate_filtering}.


In order to prepare the $0$-eigenstate of the adjacency matrix $A$ of the regular sunflower graph as the quantum state $\frac{\Pi_0\ket{s}}{\|\Pi_0\ket{s}\|}$ in polynomial time, the following two conditions must be satisfied:

\begin{enumerate}
    \item The spectral gap $\Delta$ of $H$ around the $0$ eigenvalue should be at least $1/\poly(n)$, that is, $$\Delta = \Omega(1/n^3).$$
    \item The probability of getting a $0$-eigenstate as a quantum state $\frac{\Pi_0\ket{s}}{\|\Pi_0\ket{s}\|}$ is at least $1/\poly(n)$, that is, \[
\|(\bra{0}_{\gamma}\otimes I_{\beta})\mathcal{V}_{\mathrm{circ}} \ket{0}_{\gamma}\ket{s}_{\beta}\|^2 = \|\Pi_0\ket{s}\|^2 = \Omega(1/n^2).
\].
\end{enumerate}
These two lower bounds are provided through the spectral properties of $H$ stated in \cref{cor:spectral_properties}, and the success probability lower bound further uses \eqref{eq:success_amplitude_lower_bound} and \eqref{eq:overlap_for_state_prep}. 

The state we prepare through the above procedure approximates the quantum state $\frac{\Pi_0\ket{s}}{\|\Pi_0\ket{s}\|}$, which is a 0-eigenstate of $H$. However, the 0-eigenspace of $H$ is two-dimensional (\cref{cor:spectral_properties}), being a 0-eigenstate does not uniquely determine $\Pi_0\ket{s}$. We show that the normalized state $\frac{\Pi_0\ket{s}}{\|\Pi_0\ket{s}\|}=\ket{\eta^{\mathrm{odd}}}$ in \eqref{eq:state_prepared_after_amplitude_amplification}, where $\ket{\eta^{\mathrm{odd}}}$ is a specific 0-eigenstate of $H$ defined in \cref{cor:spectral_properties}.
We summarize the above into the following lemma. A formal statement and a rigorous proof can be found in \cref{lem:generate_eta_odd}.
\begin{lemma-non}[0-eigenstate preparation (informal)]
The 0-eigenstate $\ket{\eta^{\mathrm{odd}}}$ of the effective Hamiltonian $H$ (and also of the adjacency matrix $A$) defined in \cref{cor:spectral_properties} can be prepared to $1/\poly(n)$ precision with $\poly(n)$ queries to the adjacency list oracle.
\end{lemma-non}

\paragraph{Pathfinding in regular sunflower graph and $0$-eigenstate of its adjacency matrix}

The 0-eigenstate $\ket{\eta^{\mathrm{odd}}}$ prepared in the above procedure is extremely useful in helping us find an \st path. In fact, one \st path consists of the roots of each tree $\mathcal{T}_i$ used to construct the regular sunflower graph $\SG$. These roots are $s_i$ (see Figure~\ref{fig:sunflowergraphd=3}) for $i=1,2,\cdots,n$.. According to \cref{cor:spectral_properties} (iii) the 0-eigenstate $\ket{\eta^{\mathrm{odd}}}$ has at least $\Omega(1/n)$ overlap with each $s_i$ for odd $i$ (in \cref{cor:spectral_properties} we write $s_i=S_{i,1}$ to be consistent with \cref{defn:supervertex}). Therefore measuring $\ket{\eta^{\mathrm{odd}}}$ in the computational basis yields a bit-string that represents each $s_i$ for odd $i$ with probability at least $\Omega(1/n^2)$. With $\Or(n^2)$ repetitions we can with large probability obtain all $s_i$ for odd $i$.

Collecting all these sample vertices and their neighbors, which can be obtained through the adjacency list oracle, we then have a $\Or(n^2)$-sized subgraph that contains the \st path, and therefore also the end vertex $t$, with large probability. $t$ can be identified by the indicator function mentioned in \cref{pro:pathfinding}. The path can then be found using a Breadth First Search which runs on a classical computer in time $\poly(n)$. This provides us with the main algorithmic result:
\begin{theorem-non}[Finding the \st path (informal)]
    Using $\poly(n)$ queries to the adjacency list oracle, we can find an \st path in the regular sunflower graph for the $s$ and $t$ vertices described in \cref{pro:pathfinding} with probability at least $2/3$.
\end{theorem-non}
A formal statement with rigorous proof can be found in \cref{thm:the_algorithm}.
The above algorithm and analysis rely heavily on understanding the properties of the eigenvalues and eigenvectors of the effective Hamiltonian $H$, which are described in \cref{cor:spectral_properties}. We obtain these results in Section~\ref{sec:spectral_properties} utilizing a special structure of $H$: it can be represented by a block tridiagonal matrix in \eqref{eq:Htilde_explicit}. Moreover, all eigenvectors of this block tridiagonal matrix can be factorized into tensor products of eigenvectors of even simpler matrices, as discussed in \cref{lem:spectrum_of_H_general_structure}. These facts help us find accurate estimates for the spectral gap and overlaps between the eigenstates and vertices, thus providing the guarantee that the proposed algorithm can find an \st path in $\poly(n)$ time.

\subsection{Future outlook} 

In addition to introducing a new graph structure that deviates from the traditional welded tree graph that enables exponential quantum speedups for pathfinding problems, our result also presents two additional categories of contributions.

First, our results shed some light on investigating the isogeny pathfinding problem, which is the underlying hard problem of isogeny-based cryptography. The regular sunflower graph is close to being an expander graph and shares structural features with the isogeny volcano graph, which is closely related to supersingular isogeny graphs \cite{arpin2024orientations}. Specifically, removing the outermost edges of the regular sunflower graph transforms it into a volcano graph, and our quantum algorithm is also able to demonstrate an exponential separation for finding the cycle ${s_1, s_2, \ldots, s_n}$ within the regular sunflower graph \cref{defn:randomG_build}. 
This connection provides the hope that this work may be useful for developing new algorithms for non-oracle problems. A starting point along this direction could be to investigate the properties of the $0$-eigenspace of the adjacency matrix of supersingular isogeny graphs. In particular, a modification of the current algorithm might help us find a non-oracular quantum superpolynomial speedup for the pathfinding problem in supersingular isogeny graphs, where the adjacency list oracle can be instantiated by computing isogenies between two elliptic curves.  
We summarize this in the following \cref{fig:pathfinding}

\begin{figure} [h!]
    \centering
    \includegraphics[width=0.6\textwidth]{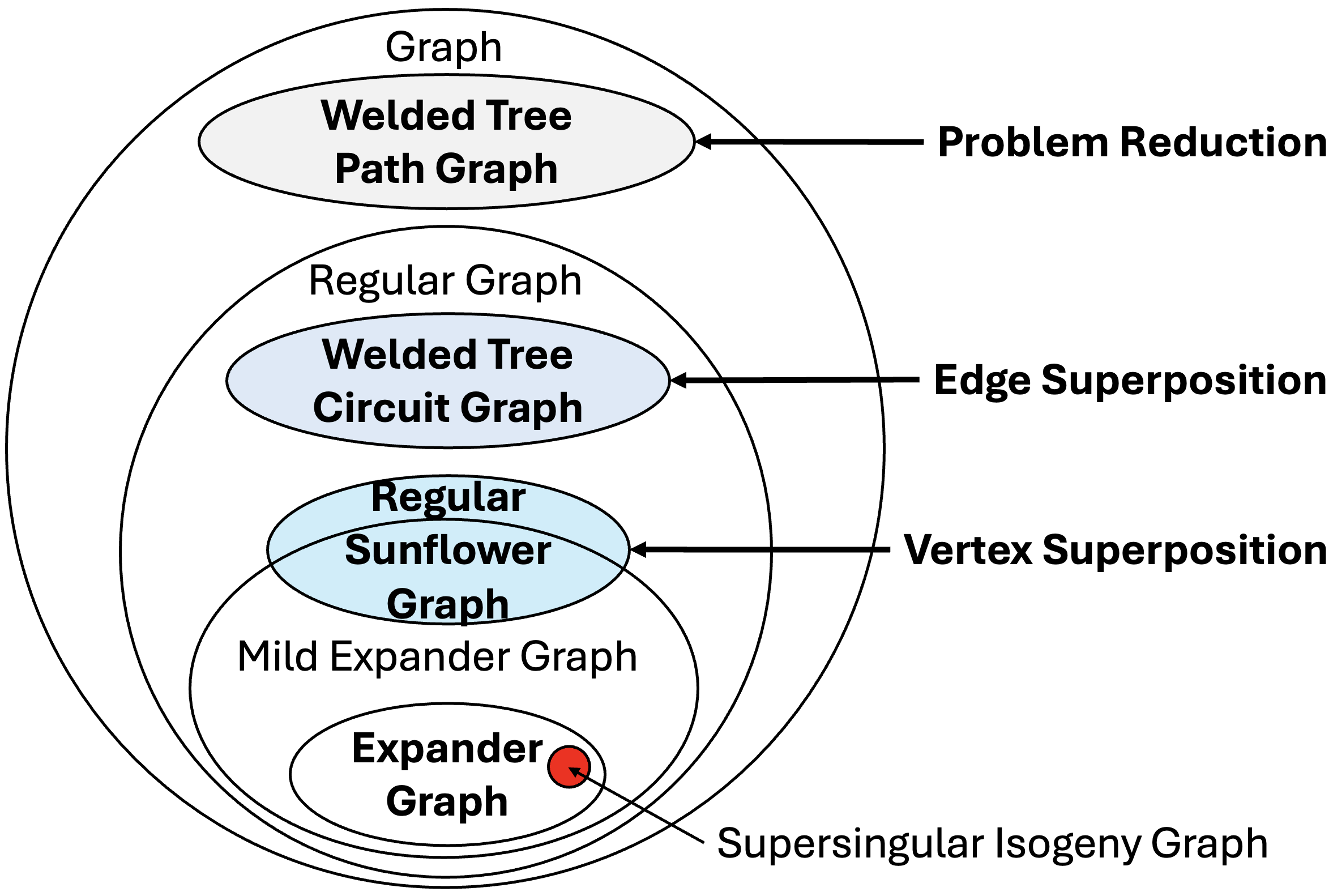}
\caption{Exponential Speedup Landscape of Pathfinding Problem}
  \label{fig:pathfinding}
\end{figure}

Second, our new quantum algorithm, which leverages the $0$-eigenspace of the adjacency matrix, provides insights to attack a wider range of problems, including the pathfinding problem in the welded tree graph. 
In \cite{childs2022quantum} the authors studied a class of quantum algorithms that are ``genuine'' and ``rooted'', and showed that this class of quantum algorithms cannot efficiently find the \st path in the welded tree graph.
Our new quantum algorithm is not ``rooted'' as defined in \cite{childs2022quantum}. This means that our algorithm also has the surprising property that it returns to us the \st path even though it does not remember the paths that it explores. 
This indicates that such algorithms may be more useful for the \st pathfinding problem than previously thought.
Whether our algorithm can be modified to to attack the welded tree pathfinding problem or not, it leads to two intriguing possibilities.
On the positive side, existing results on spectral graph theory, including those for graphs related to the welded tree graph \cite[Figure 2]{ganguly2021non}, may be combined with our algorithm for wider applications, including providing a potential approach to attack the welded tree pathfinding problem. On the other hand, if our algorithm cannot be modified to tackle the welded tree pathfinding problem, what is the essential feature that makes it work for the regular sunflower graph but not for the welded tree graph? Can we use this comparison to definitively rule out the class of algorithms excluded in \cite{childs2022quantum}? These are all questions that should be explored in future works.

\paragraph{Organization:} The rest of the paper is organized as follows: in Section~\ref{sec:background} we will introduce the notion of a (mild) expander graph, present the adjacency list oracle model, and provide other background information that will be useful later. In Section~\ref{sec:randomgraph} we will construct the regular sunflower graph and prove that with large probability it is a mild expander graph. In Section~\ref{sec:spectral_properties} we will analyze the spectral properties of the effective Hamiltonian of this graph. In Section~\ref{sec:Algorithm} we will provide the quantum algorithm to efficiently find an \st path in the graph. In Section~\ref{sec:the_classical_lower_bound} we will provide a query complexity lower bound result showing that any classical algorithm requires exponential time to find the \st path, thus establishing an exponential quantum-classical separation.

\section{Background}
\label{sec:background}
In this section, we will first introduce the basic notation for the graph problem that we are going to study. In particular, we will introduce the definition of mild expander graphs. 
Because we are studying the path-finding problem in an oracular setting, it is important to specify the classical and quantum oracles that provide information about the graph. We will also discuss the way to use the quantum oracle to build a unitary circuit encoding the adjacency matrix of the graph, thus enabling the quantum algorithm to process this matrix.
Certain spectral properties will feature prominently in our analysis of the algorithm, and we will give examples in this section to help readers gain an intuitive understanding.

\subsection{Graph expansion}
\label{sec:graph_expansion}

 Following \cite[Chapter 4]{vadhan2012pseudorandomness}, we define the \emph{neighborhood} and \emph{edge boundary} of a set of vertices and \emph{$(K,\epsilon)$-vertex expander graph} as follows:
\begin{definition}[Neighborhood in graph]
\label{defn:neighborhood}
    For a graph $G=(V,E)$, we define the neighborhood of a vertex $u\in V$ to be $\N (u):=\{v\in V:(u,v)\in E\}$.
    Similarly we define the neighborhood of a set $S\subset V$ to be $\N(S) = \bigcup_{u\in S}\N(u)$.
\end{definition}
The \emph{edge boundary} of a set of vertices is defined in the following way:
\begin{definition}[Edge boundary]
\label{defn:edge_boundary}
    For a graph $G=(V,E)$, we define the edge boundary of a subset $S\subset V$ to be $\partial (S) = \{(u,v):u\in S, v\notin S\}$.
\end{definition}


With the above, we are ready to define what an expander graph is:
\begin{definition} [$(K,\epsilon)$-vertex expander graph] 
\label{defn:graph_expansion}
  A $d$-regular, unweighted multigraph $G$ is a \emph{$(K,\epsilon)$-vertex expander} if for every set $S \subset V$ such that $|S|\leq K \leq |V|/2$, 
  we have
 \[
  |\N(S)\backslash S| \geq \epsilon d  |S|.
 \]
\end{definition}
There are usually three equivalent definitions of expanders based on vertex expansion, edge expansion, and spectral expansion. A $(K,\epsilon)$-vertex expander graph according to \cref{defn:graph_expansion} also satisfies
\[
|\partial(S)|\geq |\N(S)\backslash S| \geq \epsilon d  |S|,
\]
where the first inequality is because each element in $\N(S)\backslash S$ must be connected to $S$ by at least one edge. Through Cheeger's inequalities, if $K=|V|/2$, this also implies spectral expansion, i.e., the spectral gap of the graph Laplacian (for a definition see \cite[Chapter~1]{chung1997spectral}) is lower bounded by at least $\epsilon^2 d/2$ \cite[Theorem 4.9]{vadhan2012pseudorandomness}.


Based on the definition of $(K,\epsilon)$-vertex expander in \cref{defn:graph_expansion}, we define \emph{mild expander graph} as following.
\begin{definition} [Mild expander graph] 
\label{defn:mildexpandergraph}
  A $d$-regular, unweighted multigraph $G$ is a \emph{$(K,\epsilon)$-vertex mild expander graph} if for every set $S \subset V$ such that $|S|\leq K \leq |V|/2$ and $\epsilon d= 1/\polylog {|V|}$, 
  we have
 \[
  |\N(S)\backslash S| \geq \epsilon d  |S|.
 \]
\end{definition}




\subsection{The oracle models}
\label{sec:the_oracle_models}

We consider a multigraph $G=(V,E)$ containing $\mathsf{N}$ vertices, with degree $d=\Or(1)$. Here $E$ is a multiset containing all the edges and it allows multiplicity. Each vertex is labelled by a bit-string of length $\mathsf{n}=\lceil\log_2(\mathsf{N})\rceil$. Below we will not distinguish between a vertex $v$ and its bit-string label or the integer the bit-string represents.

In the classical setting, we will access this multigraph through a classical oracle 
\begin{definition}[The adjacency list oracle]
\label{defn:classical_adjacency_list_oracle}
    Let $G=(V,E)$ be an undirected multigraph of degree $d=\Or(1)$, $V=\{0,1,2\cdots,\mathsf{N}-1\}$, $\mathsf{n}=\lceil\log_2(\mathsf{N})\rceil$.
    For any $v\in V$ and $k\in\{1,2,\cdots,d\}$, if the $k$th neighbor (using the natural ordering of bit-strings) of the vertex $v\in\mathcal{V}$ (not counting multiplicity) exists, then
    $O_{G,1}(v,k)$ returns the this neighbor. If the $k$th neighbor does not exist, then $O_{G,1}(v,k)=k+2^{\mathsf{n}}$. For any $v,v'\in V$, $O_{G,2}(v,v')$ returns the multiplicity of the edge $(v,v')$.
\end{definition}
There are several other oracle models that can be considered. For example, we can let $O_{G}(v,k)$ return the $k$th neighbor counting multiplicity, or let $O_{G}$ return a list of all neighbors at the same time. These alternative models are equivalent to the one we are considering up to $\poly(d)=\Or(1)$ overhead. 

In the quantum setting, in order to obtain a speedup, we need to query the above adjacency list oracle in superposition. This necessitates having access to black-box unitaries implementing the adjacency list oracle. We define these unitaries as follows:
\begin{definition}
\label{defn:quantum_adjacency_list_oracle}
    Let $G=(V,E)$ be an undirected multigraph of degree $d=\Or(1)$, $V=\{0,1,\cdots,\mathsf{N}-1\}$, $\mathsf{n}=\lceil\log_2(\mathsf{N})\rceil$.
    We define unitaries $O^{Q}_{G,1}$ and $O^{Q}_{G,2}$ to satisfy, for any $v\in V$ and $k\in\{0,1,\cdots,d-1\}$,
    \[
    O^{Q}_{G,1}\ket{v,k,c} = \ket{v,k,c\oplus O_{G,1}(v,k)},\quad O^{Q}_{G,2}\ket{v,v',c}=\ket{v,v',c\oplus O_{G,2}(v,v')}.
    \]
    where $O_{G,1}$, $O_{G,2}$ are defined in \cref{defn:classical_adjacency_list_oracle}, $c$ is any bit-string, and $\oplus$ denotes bit-wise addition modulo 2.
\end{definition}
There are other ways of defining the oracles. Most notably, in \cite[Footnote 4]{gilyen2020ExpAdvAdiabStoqQC}, where the oracle maps $\ket{v,k,c}$ to the $(k+1)$th neighbor of $v$ counting multiplicity, and an additional index function that makes the mapping invertible. We remark that the definition adopted here is equivalent to the one in \cite{gilyen2020ExpAdvAdiabStoqQC} up to a $\poly(d)=\Or(1)$ overhead, because both of them can be queried $d$ times to obtain all the neighbors, and then can be converted to each other with $\poly(d)$ logical operations on the list of neighbors. Redundant information can then be uncomputed using $d$ queries to the inverses of these oracles. 

All information of the multigraph is represented in its adjacency matrix, as defined below:
\begin{definition}
    \label{defn:adjacency_matrix_of_multigraph}
    Let $G=(V,E)$ be an undirected multigraph of degree $d=\Or(1)$, $\mathsf{N}=|V|$. 
    Then its adjacency matrix $A=(A_{v,v'})_{v,v'\in V}$ is an $\mathsf{N}\times\mathsf{N}$ matrix in which $A_{v,v'}$ is equal to the multiplicity of the edge $(v,v')$ for any $v,v'\in V$.
\end{definition}


The quantum algorithm that we are going to describe in Section~\ref{sec:Algorithm} will be centered on the adjacency matrix, and we therefore need some way to encode the adjacency matrix into a quantum circuit. We will do so through a technique called block encoding as defined below:
 
\begin{definition}[Block encoding {\cite[Definition 43]{gilyen2018QSingValTransf}}]
\label{defn:block_encoding}
Suppose that $A$ is an $s$-qubit operator and let $\alpha, \epsilon\in \mathbb{R}_{+}$ and $a\in \mathbb{N}$. We say that the $(s+a)$-qubit unitary $U$ is an $(\alpha,a,\epsilon)$ block encoding of A on registers $\beta_1$ and $\beta_2$, if 
\[
\|A-\alpha \left(\bra{0}_{\beta_1}^{\otimes a}\otimes I_{\beta_2}\right)U \left(\ket{0}_{\beta_1}^{\otimes a}\otimes I_{\beta_2}\right) \| \leq \epsilon.
\]    
\end{definition}

Essentially, the above definition means that we can construct a unitary circuit in which a certain sub-matrix corresponds to the matrix $A$ that we want to process in our algorithm. 
We will construct such a block encoding by treating the adjacency matrix as a sparse matrix and thereby applying \cite[Lemma~48]{gilyen2018QSingValTransfArXiv}. The sparse matrix block encoding in this lemma requires access to the \emph{sparse access oracle}, which we define below:
\begin{definition}[Sparse access oracle {\cite[Lemma~48]{gilyen2018QSingValTransfArXiv}}]
\label{defn:sparse_access_oracle}
Let $A\in C^{2^n\times 2^n}$ be a Hermitian matrix that is $d$ sparse, i.e., each row and each column have at most $d$ nonzero entries.

Let $A_{ij}$ be a $b$-bit binary description of the $ij$-matrix element of $A$. The sparse-access oracle $O^Q_{\mathrm{val}}$ is defined as follows.
\[
 O^Q_{\mathrm{val}}: \ket{v}\ket{v'}\ket{0}^{\otimes b} \longrightarrow \ket{v}\ket{v'}\ket{A_{ij}} \quad\quad \forall
 v, v'\in [2^n]-1.
\] 

Let $r_{vk}$ be the index for the $k$-th non-zero entry of the $v$-th row of $A$. If there are fewer than $k$ nonzero entries, then $r_{vk}=k+2^n$. The sparse access oracle $O^Q_{\mathrm{loc}}$ is defined as follows.
\[ O^Q_{\mathrm{loc}} : \ket{v}\ket{k} \longrightarrow \ket{v}\ket{r_{vk}} \quad\quad \forall v \in [2^n]-1, k \in[s_r].\]
\end{definition}

In the above we assume only one oracle to access the location of non-zero entries in each row, rather than two separate oracles for rows and columns respectively as in \cite[Lemma~48]{gilyen2018QSingValTransfArXiv}. This is because we focus on a Hermitian matrix, and therefore the row and column oracles are identical.
We note that the sparse access oracle defined above can be simulated by the the adjacency list oracle, as stated in the lemma below:
\begin{lemma} 
\label{lem:adj-simu-sparse}
Let $A$ be the adjacency matrix 
The sparse access oracle $O^Q_{\mathrm{val}}$, $O^Q_{\mathrm{loc}}$ (\cref{defn:sparse_access_oracle}) of the adjacency matrix $A$ of a $d$-regular multigraph can be simulated by $O(d)$ queries of the adjacency list oracle $O^Q_{G,1}$, $O^Q_{G,2}$ (\cref{defn:quantum_adjacency_list_oracle}). 
\end{lemma} 

A proof of this lemma is provided in Appendix~\ref{sec:block_encoding_construction_from_adjacency_list}. One can easily see that the adjacency list oracle contains all the information we need for building the sparse matrix $A$, so such a result is in no way surprising. However, we need to make sure that no redundant information remains in the simulation of the sparse access oracle, and this is done through uncomputation as described in Appendix~\ref{sec:block_encoding_construction_from_adjacency_list}.

We then restate \cite[Lemma~48]{gilyen2018QSingValTransfArXiv} here, while specializing to the case of a Hermitian matrix to simplify the notation:
\begin{lemma} [Block-encoding of sparse-access matrices {\cite[Lemma~48]{gilyen2018QSingValTransfArXiv}}] 
\label{lem:block_encoding_sparse_matrix_restatement}
 \label{lem:block-encoding-sparse}Let $A\in C^{2^n\times 2^n}$ be a Hermitian matrix that is $d$ sparse, and each element $A_{v,v'}$ of $A$ has absolute value at most 1. We also assume that the $b$-bit binary description of $A_{v,v'}$ is exact. We can implement a $(d, n+3,\epsilon)$-block-encoding of $A$ with a single use of $O^Q_{\mathrm{loc}}$, two uses of $O^Q_{\mathrm{val}}$ and additionally using $O(n+\log^{2.5}(\frac{d}{\epsilon}) )$ one and two qubit gates while using $O(b,\log^{2.5}(\frac{d}{\epsilon}))$ ancilla qubits.
\end{lemma}

Combining \cite[Lemma~48]{gilyen2018QSingValTransfArXiv} with \cref{lem:adj-simu-sparse}, we arrive at the following result:
\begin{lemma}
    \label{lem:block_encoding_adjacency_matrix_new}
     Let $G=(V,E)$ be an undirected multigraph of degree $d=\Or(1)$, $V=\{0,1,\cdots,\mathsf{N}-1\}$, $\mathsf{n}=\lceil\log_2(\mathsf{N})\rceil$. Then a $(d^2,\mathsf{n}+3,\epsilon_A)$-block encoding of the adjacency matrix $A$ as defined in \cref{defn:adjacency_matrix_of_multigraph} can be constructed using $\Or(1)$ queries to the adjacency list oracle in \cref{defn:quantum_adjacency_list_oracle}, $\Or(\mathsf{n}+\log^{2.5}(d/\epsilon_A))$ additional elementary gates, and $\Or(\log^{2.5}(d/\epsilon_A))$ ancilla qubits.
\end{lemma}

\begin{proof}
    We simulate the sparse access oracle $O^Q_{\mathrm{loc}}$ and $O^Q_{\mathrm{val}}$ using the adjacency list oracle $O^Q_{G,1}$, $O^Q_{G,2}$ through \cref{lem:adj-simu-sparse}. With the sparse access oracle, we can then construct the block encoding through \cite[Lemma~48]{gilyen2018QSingValTransfArXiv} as restated in \cref{lem:block_encoding_sparse_matrix_restatement}. 

    Note that because \cite[Lemma~48]{gilyen2018QSingValTransfArXiv} requires all matrix entries to have absolute value at most $1$, it therefore only applies to $A/d$ rather than $A$, and yields a $(d,\mathsf{n}+3,\epsilon_A/2d)$-block encoding of $A/d$ (we choose to let the error be $\epsilon_A/2d$). Obviously, a block encoding of $A/d$ is also a block encoding of $A$, but with a different set of parameters, and this is how we arrive at a $(d^2,\mathsf{n}+3,\epsilon_A/2)$-block encoding of $A$.

    Because we can only represent the entries of $A/d$ with $b$ bits, this introduces an error which is at most $\Or(2^{-b})$. Because the error only occurs on the non-zero entries, and there are at most $d$ such entries in each row, the resulting error in the matrix when measured in the spectral norm is at most $\Or(d2^{-b})$ by Lemma~\ref{lem:bound_spectral_norm_using_inf_norm}. Consequently, to make this error at most $\epsilon_A/2$, it suffices to choose $b=\Or(\log(d/\epsilon_A))$. Therefore the effect of $b$ on the number of ancilla qubits is subsumed by the $\log^{2.5}(d/\epsilon_A)$ term.
\end{proof}

\subsection{The 0-eigenvector of the adjacency matrix}
\label{sec:the_zero_eigenvector_of_A}

Our quantum algorithm is built on the observation that the $0$-eigenvector of the adjacency matrix of the graph sometimes yields useful information.
This is also the underlying idea of the quantum algorithm for exit-finding in \cite{balasubramanian2023exponential}. 
Below we will look at some example graphs and the corresponding $0$-eigenvectors of the adjacency matrices. These examples will be useful for our analysis of the regular sunflower graph we construct. 
Unless otherwise specified, we assume that all graphs are unweighted, which means that all edges have weight $1$.  

\begin{definition}[Adjacency matrix of 
 a cycle graph $C_n$] A \emph{cycle graph} consists of $n$ vertices $\{v_1,v_2,\ldots,v_n\}$ with edges $\{v_1,v_{2}\},\{v_2,v_{3}\}, \ldots, \{v_{n-1},v_{n}\}, \{v_{n},v_{1}\}$. 
 The adjacency matrix $D_0$ of $C_n$ is defined as    
\begin{equation}
    \label{eq:defn_D0}
    D_0 = \begin{pmatrix}
        0     & 1 &        &            &  1   \\
        1 & 0     & 1  &            &     \\
              & 1 & 0      & \ddots     &     \\
              &       & \ddots & \ddots     & 1    \\
          1    &       &        & 1 & 0
    \end{pmatrix}.
\end{equation}
    
\end{definition}
\begin{definition}[Adjacency matrix of weighted path graph $P_m$] A \emph{path graph} consists of $n$ vertices $\{v_1,v_2,\ldots,v_m\}$ with edges $\{v_i,v_{i+1}\}$ for $i=1,2,\ldots,n-1$. The \emph{path graph} $P_m$ is weighted if each edge $\{v_i,v_{i+1}\}$ has weight $t_i \in \R$. The adjacency matrix $D_1$ of  $P_m$ is defined as    
 \begin{equation}
    D_1 = 
    \begin{pmatrix}
         0     & t_1  &        &            &     \\
        t_1  & 0     & t_2   &            &     \\
              & t_2  & 0      & \ddots     &     \\
              &       & \ddots & \ddots     & t_{m-1}     \\
              &       &        & t_{m-1}  & 0
    \end{pmatrix}_{ m\times m}.
\end{equation}
\end{definition}

\begin{example}[$0$-eigenvector of adjacency matrix $D_0$ of a cycle graph $C_{n}$ {\cite[Lemma 6.5.1]{spielman2019spectral}}] 
 \label{example:zero_eigvec_cycle}
 Let $n$ be an integer that is a multiple of $4$ and $D_0$ be the adjacency matrix of a cycle graph $C_n$, then eigenvalues of $D_0$ are $2\cos(\frac{2\pi l}{n})$ for $l=1,\ldots, n$ and  the two orthogonal $0$-eigenvectors $x,y$ of $D_0$ are the following:
 \begin{equation}
      x_l= \cos (l\pi/2) \text{ where } l = 1,2,\ldots, n. 
 \end{equation}
\begin{equation}
      y_l=  \sin (l\pi/2) \text{ where } l = 1,2,\ldots, n.
 \end{equation}
\end{example}

When it comes to the $0$-eigenvector of the adjacency matrix of a weighted path graph of odd length, one can also compute the unique $0$-eigenvector implicitly as in the proof of \cite[Lemma 3.2]{balasubramanian2023exponential}. We formulate this result as following:

\begin{example}[$0$-eigenvector of adjacency matrix $D_1$ of an weigthed path graph $P_{m}$ ] 
 \label{example:zero_eigvec_path}
 Let $m$ be an odd integer and $d$ be a positive integer.
 Let $D_1$ be the adjacency matrix of the path graph $P_{m}$ with edges weights $t_1,t_2,\ldots,t_{m-1}$,
 then the unique $0$-eigenvector of $D_1$ is $\ket{x}=(x_1,x_2,\cdots,x_m)^\top$, where
 \begin{equation}
      x_l=  \left\{\begin{array}{ll}
  0 & \mbox{for even } l= 2,4,\ldots,m-1.  \\
  \prod_{k=1}^{(l-1)/2}\left( - \frac{t_{2k-1}}{t_{2k}}\right) x_1 & \mbox{for odd } l= 1,3,\ldots,m.\\
\end{array}\right. 
 \end{equation}
\end{example}

In particular, we will use the case where
$t_1=\sqrt{d-2}$, $t_2=\ldots=t_{m-1}=\sqrt{d-1}$. In this case, for even $l\geq 2$, we have 
\begin{equation}
\label{eq:relation_between_x0_and_xl}
    |x_l| = \sqrt{\frac{d-2}{d-1}} |x_1|.
\end{equation}

 






\section{The regular sunflower graphs $\SG$} \label{sec:randomgraph}

In this section, we will introduce the family of multigraphs on which we can obtain an exponential quantum advantage. We call this family of multigraphs \emph{sunflower graphs} and denote an instance by $\SG$.
When the degree of the graph is greater than 7, we will show that the sunflower graph $\SG$ is a mild expander graph with high probability in Section~\ref{sec:spectral_properties}. 
In order to rigorously establish a classical query complexity lower bound, we will need to add isolated vertices to the sunflower graph, as was done in \cite{childs2003ExpSpeedupQW,balasubramanian2023exponential}. This will be discussed in detail in Section~\ref{sec:the_enlarged_sunflower_graph}.
In Section~\ref{sec:invariant_subspace_and_effective_ham} we will discuss the structure in the adjacency matrix of the sunflower graph, which is also present in its enlarged version.

\subsection{Graph definition}
\label{sec:graph_definition}


We first provide the definition of the regular sunflower graph $\SG$. 

\begin{definition}[The regular sunflower graph $\SG=(\SV,\SE)$]
\label{defn:randomG_build}
Let $\mathcal{T}_1,\mathcal{T}_2,\ldots,\mathcal{T}_n$
be rooted trees of height $m \geq 2$. For each $\mathcal{T}_i$, the root has degree $d-2$, and all other internal vertices have degree $d-1$. All leaves must be distance $m$ from the root.  We then link the roots $s_1,s_2,\ldots,s_n$ of trees $\mathcal{T}_{1}, \mathcal{T}_{2}, \ldots, \mathcal{T}_{n}$ with the edges $\{s_1,s_2\},\{s_2,s_3\}, \ldots, \{s_{n-1},s_n\},\{s_n,s_1\}$. Next, we link the
$(d-2)(d-1)^{m-2}$ leaves of $\mathcal{T}_{i}$ with the leaves of $\mathcal{T}_{i+1}$ for $1 \leq i\leq n-1$ and the leaves of $\mathcal{T}_{n}$ with the leaves of $\mathcal{T}_{1}$.  For each pair of trees, select $(d-1)/2$ random perfect matchings between their leaves and add all their edges to the graph.  The resulting multigraph is a \emph{$d$-regular sunflower graph $\SG=(\SV,\SE)$}.
\end{definition}


   
An illustration of the graph is given in Figure~\ref{fig:sunflowergraphd=3} for $d=3, m=4, n=8$. 
For simplicity, throughout the paper we will assume $m=\Theta(n)$.

In this work we focus on finding the path between an entrance vertex $s$ and an exit vertex $t$. In the graph $\SG$ defined above, $s$ is the root vertex of the tree $\mathcal{T}_1$, and $t$ is the root vertex of the tree $\mathcal{T}_{n/2+1}$. 
Here $s$ is given to us directly at the beginning in the form of a bit-string representing it. We also give an indicator function for identifying $t$.

\begin{definition}
    \label{defn:exit_oracle}
    The indicator function $f_t$ satisfies:
    \[
    f_t(v)=\begin{cases}
        1 & \text{ if } v=t, \\
        0 & \text{ otherwise.}
    \end{cases}
    \]
\end{definition}


It is illustrative to group some of the vertices in this multigraph together to reveal some additional structure in it.

\begin{definition}
    \label{defn:supervertex}
    A supervertex $S_{i,j}$ is the set of vertices in the $j$th layer of tree $\mathcal{T}_i$ in the sunflower graph defined in \cref{defn:randomG_build}.
\end{definition}
The cardinality of this set can be computed from \cref{defn:randomG_build} to be
\begin{equation}
\label{eq:cardinality_supervertex}
    s_{i,j} = |S_{i,j}| =
    \begin{cases}
        1, &\text{ if }j=1, \\
        (d-1)^{j-2}(d-2),&\text{ otherwise.}
    \end{cases}
\end{equation}
We will call $S_{i,j}$ a \emph{supervertex}. From this we can compute the total number of vertices in the sunflower graph to be $\mathsf{N}_{\SG} = n(d-1)^{m-1}$.

Let $e_{ij,kl}$ be the number of edges, counting multiplicity, between two sets of vertices $S_{i,j}, S_{k,l}$ in $\SG=(\SV,\SE)$, where $1 \leq i,k \leq n, 1\leq j,l\leq m$.  In other words,
\begin{equation}
    \label{eq:num_edges_from_adjacency_matrix}
    e_{ij,kl} = \sum_{u\in S_{i,j}}\sum_{v\in S_{k,l}} A_{u,v}.
\end{equation}
When $j=l=m$, $k= i + 1$ for $1 \leq i\leq n-1$ or $k=1,i=n$,  $e_{ij,kl}=\frac{d-1}{2} (d-1)^{m-2}(d-2)$, which consists of $\frac{d-1}{2}$ unions of random perfect matching between the vertices in $S_{i,m}$ and the vertices in $S_{i+1,m}$ for $1\leq i\leq n-1 $, also between $S_{1,m}$ and $S_{n,m}$. When $j=l=1$, $k= i + 1$ for $1 \leq i\leq n-1$ or $k=1,i=n$,  $e_{ij,kl}=1$, which is the edge between the vertex in $S_{i,1}$ and the vertex in $S_{i+1,1}$ for $1 \leq i\leq n-1 $. When $i=k$ and $j=l+1$ or $l-1$, we have $e_{ij,kl}=\max\{s_{i,j},s_{k,l}\}$ due to the tree structure of $\mathcal{T}_{i}$. For other cases of $i,j,k,l$, we have $e_{ij,kl}=0$. Because the graph is undirected, $e_{ij,kl}=e_{kl,ij}$. To summarize:
\begin{equation}
\label{eq:num_edges_bw_supervertices}
e_{ij,kl} =  \left\{\begin{array}{ll}
        1       &       \mbox{if } j=l=1 \mbox{ and } k=i + 1 \mbox{ for } 1 \leq i\leq n-1;       \\
         1       &       \mbox{if } j=l=1 \mbox{ and }  k=1,i=n;\\
         \max\{s_{i,j},s_{k,l}\}  &       \mbox{if } i=k \mbox{ and } j= l+ 1  \mbox{ or } l-1 ;      \\
        \frac{d-1}{2}(d-2)(d-1)^{m-2} & \mbox{if } j=l=m \mbox{ and } k= i+ 1 \mbox{ for } 1 \leq i\leq n-1;  \\
        \frac{d-1}{2}(d-2)(d-1)^{m-2} & \mbox{if } j=l=m \mbox{ and } k=1, i=n;\\
          0       &        \mbox{Otherwise}.      \\
        
    \end{array}\right.    
\end{equation}

By the construction of the graph is \cref{defn:randomG_build}, one can readily verify that the number of edges (counting multiplicity) connecting $u\in S_{i,j}$ to another fixed supervertex $S_{k,l}$ is the same as that of all $u\in S_{i,j}$. In other words, $\sum_{v\in S_{k,l}}A_{u,v}$ is the same for all $u\in S_{i,j}$. By \eqref{eq:num_edges_from_adjacency_matrix}, we therefore have
\begin{equation}
    \label{eq:num_edges_connecting_u_to_Skl}
    \sum_{v\in S_{k,l}} A_{u,v} = \frac{e_{ij,kl}}{s_{i,j}}, \quad \forall u\in S_{i,j}. 
\end{equation}

From the above discussion, we can see that the supervertices can be arranged into a graph: we link an edge between any pair of supervertices $S_{i,j}$ and $S_{k,l}$ for which $e_{ij,kl}\neq 0$. The resulting graph is shown in Figure~\ref{fig:supergraph}, and is referred to as the \emph{supergraph}. We can see from that the supergraph has a certain 2-dimensional structure that is useful for the analysis of it properties.

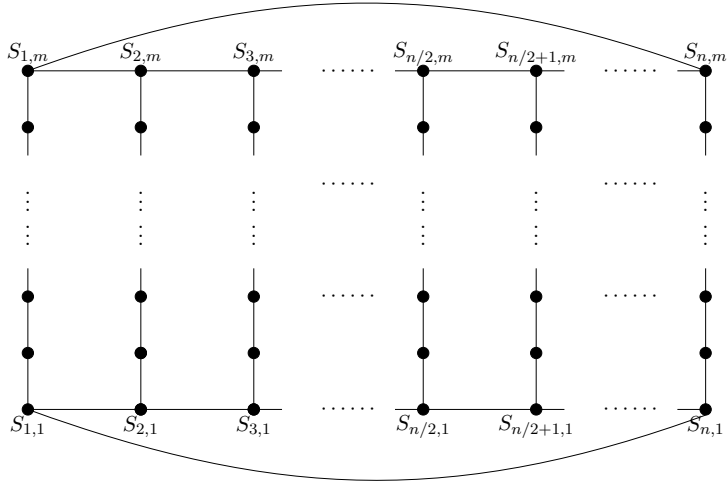
\begin{figure}
  \centering
  \resizebox{4in}{!}{
\begin{tikzpicture}
\node at (0,0) {\begin{tikzpicture}
  \draw[-] (-2,0)--(-0.1,0);
  \draw[-] (-2,6)--(-0.1,6);
  \draw[-] (-2,0)--(-2,0.9);
  \draw[-] (-2,2.5)--(-2,0.9);
   \draw[-] (-2,4.5)--(-2,5.9);

  \draw[-] (0,0)--(-0.1,0);
  \draw[-] (0,0)--(0,0.9);
  \draw[-] (0,2.5)--(0,0.9);
   \draw[-] (0,4.5)--(0,5.9);

  \draw[-] (0,0)--(2.5,0);
   \draw[-] (0,6)--(2.5,6);
 
  \draw[-] (2,0)--(2,0.9);
  \draw[-] (2,2.5)--(2,0.9);
   \draw[-] (2,4.5)--(2,5.9);

   \draw[-] (4.5,0)--(7.5,0);
     \draw[-] (4.5,6)--(7.5,6);

  \draw[-] (5,0)--(5,0.9);
  \draw[-] (5,2.5)--(5,0.9);
   \draw[-] (5,4.5)--(5,5.9);

  \draw[-] (7,0)--(7,0.9);
  \draw[-] (7,2.5)--(7,0.9);
   \draw[-] (7,4.5)--(7,5.9);

\draw[-] (10,0)--(10,0.9);
  \draw[-] (10,2.5)--(10,0.9);
   \draw[-] (10,4.5)--(10,5.9);

   \draw[-] (10,0)--(9.5,0);
   \draw[-] (10,6)--(9.5,6);

  
  \filldraw (-2,0) circle (.1);
    \filldraw (-2,1) circle (.1);
     \filldraw (-2,2) circle (.1);
      \filldraw (-2,5) circle (.1);
       \filldraw (-2,6) circle (.1);

  \filldraw (0,0) circle (.1);
    \filldraw (0,1) circle (.1);
     \filldraw (0,2) circle (.1);
      \filldraw (0,5) circle (.1);
       \filldraw (0,6) circle (.1);

    \filldraw (2,0) circle (.1);
    \filldraw (2,1) circle (.1);
     \filldraw (2,2) circle (.1);
      \filldraw (2,5) circle (.1);
       \filldraw (2,6) circle (.1);

  \filldraw (5,0) circle (.1);
    \filldraw (5,1) circle (.1);
     \filldraw (5,2) circle (.1);
      \filldraw (5,5) circle (.1);
       \filldraw (5,6) circle (.1);
    
   \filldraw (7,0) circle (.1);
    \filldraw (7,1) circle (.1);
     \filldraw (7,2) circle (.1);
      \filldraw (7,5) circle (.1);
       \filldraw (7,6) circle (.1);

  \filldraw (10,0) circle (.1);
    \filldraw (10,1) circle (.1);
     \filldraw (10,2) circle (.1);
      \filldraw (10,5) circle (.1);
       \filldraw (10,6) circle (.1);
    
  \filldraw (0,0) circle (.1);
  \filldraw (2,0) circle (.1);
  \filldraw (0,1) circle (.1);

  \node at (-2,-0.3) {$S_{1,1}$};
  \node at (0,-0.3) {$S_{2,1}$};
  \node at (2,-0.3) {$S_{3,1}$};
  \node at (5,-0.3) {$S_{n/2,1}$};
  \node at (7,-0.3) {$S_{n/2+1,1}$};
  \node at (10,-0.3) {$S_{n,1}$};

    \node at (-2,6.3) {$S_{1,m}$};
  \node at (0,6.3) {$S_{2,m}$};
  \node at (2,6.3) {$S_{3,m}$};
  \node at (5,6.3) {$S_{n/2,m}$};
  \node at (7,6.3) {$S_{n/2+1,m}$};
  \node at (10,6.3) {$S_{n,m}$};
  
  \node at (3.7,0) {$\cdots \cdots$};
   \node at (3.7,2) {$\cdots \cdots$};
   \node at (3.7,4) {$\cdots \cdots$};
   \node at (3.7,6) {$\cdots \cdots$};

   \node at (8.7,0) {$\cdots \cdots$};
   \node at (8.7,2) {$\cdots \cdots$};
   \node at (8.7,4) {$\cdots \cdots$};
   \node at (8.7,6) {$\cdots \cdots$};

   \node at (0,3.8) {$\vdots$};
   \node at (0,3.2) {$\vdots$};
   \node at (-2,3.8) {$\vdots$};
   \node at (-2,3.2) {$\vdots$};
 \node at (2,3.8) {$\vdots$};
   \node at (2,3.2) {$\vdots$};
   \node at (5,3.8) {$\vdots$};
   \node at (5,3.2) {$\vdots$};
      \node at (7,3.8) {$\vdots$};
   \node at (7,3.2) {$\vdots$};
      \node at (10,3.8) {$\vdots$};
   \node at (10,3.2) {$\vdots$};

  \draw[-] (-2,0) to[out=-20,in=-160] (10,-0.1); 
  \draw[-] (-2,6) to[out=20,in=160] (10,6); 
  

  \end{tikzpicture}};
\end{tikzpicture}
}
\caption{The supergraph consisting of the supervertices defined in \cref{defn:supervertex}. Each pair of supervertices are linked by an edge if there exists an edge in the sunflower graph $\SG$ between two vertices contained in these two supervertices respectively.}
\label{fig:supergraph}
\end{figure}

\subsection{Expansion properties}
\label{sec:expansion_properties}

In this section we will investigate the expansion properties of the regular sunflower graph $\SG$. First, we note that each pair of adjacent sets $S_{i,m}$ and $S_{i+1,m}$ for $1\leq i\leq n-1$, also between $S_{1,m}$ and $S_{n,m}$ form a random bipartite graph, with the connectivity given by $(d-1)/2$ random perfect matchings. The following lemma gives us the expansion property of a bipartite graph with $N$ vertices on each side whose connectivity is given through $D$ random perfect matchings:




\begin{lemma}
    \label{lem:biparmildexp}
    Let $L$ and $R$ be two sets of vertices with $|L|=|R|=N$. Link $L$ and $R$ through $D\geq 3$ random perfect matchings. Denote the resulting graph by $G_B = (V_B, E_B)$, and let $\chi=2/3$, $\delta =1/(2\log N)$. Then $G_B$ has the following expansion properties with probability $1-\Theta(1/N^{2\log (3/2)(1-\delta)})$:
    \begin{itemize}
        \item[(i)] For any subset $L'\subseteq L$ and $|L'|\leq \chi N$, we have $|\N(L')|=|\N(L')\backslash L'| \geq (1+\delta)|L'|$, where $\N(L')$ denotes the neighborhood of $L'$ as defined in Definition~\ref{defn:neighborhood}.
        \item[(ii)] For any subset $T\subseteq L\cup R$ and $|T| \leq  N$, we have  $|\N(T) \backslash T| \geq \frac{\delta}{2} |T|$. 
    \end{itemize}
    In other words, with probability $1-\Theta(1/N^{2\log (3/2)(1-\delta)})$, this bipartite regular graph is a mild expander graph.
\end{lemma}

The proof of Lemma \ref{lem:biparmildexp} follows the proof of \cite[Theorem 4.1.1]{kowalski2019introduction}, so we defer the proof to Appendix \ref{sec:apppendmildbipar}.
In the context of subgraph formed by two adjacent sets of vertices between $S_{i,m}$ and $S_{i+1,m}$ for $1\leq i\leq n-1$, also between $S_{1,m}$ and $S_{n,m}$ of the regular sunflower graph $\SG$, we have $D=(d-1)/2$, and $N=(d-2)(d-1)^{m-2}$. We will next use this lemma to study the expansion property of $\SG$.



\begin{theorem} 
\label{thm:graph_expansion}
Let $m=\Theta(n)$ and $d\geq 7$ be a odd integer constant. With at least $1-\exp(-\Omega(n))$ probability, the $d$-regular sunflower graph $\SG=(\SV,\SE)$ defined in Definition~\ref{defn:randomG_build} is a mild expander graph as defined in Definition~\ref{defn:mildexpandergraph}. More precisely, for any subset $T \subseteq \SV$, and $|T|\leq \frac{|\SV|}{2}$, we have $|\N(T)\backslash T| \geq \frac{1}{\polylog {|\SV|}} \cdot |T|$.
\end{theorem}

\begin{proof}
    First we will introduce some parameters to be used later.
    Let $\chi=\frac{2}{3}$, $\delta= 1/\Theta(n)$ and $N=(d-2)(d-1)^{m-2}$. 
    We note that
    \[
    |\mathcal{V}| = n(d-1)^{m-1}=\Theta(nN).
    \]
    To prove the theorem, it suffices to show that with  at least $1-\exp(-\Omega(n))$ probability, for any subset $T \subseteq \SV$ and $|T|\leq |\SV|/2$, we have $|\N(T) \backslash T |\geq \Omega(|T|/n^3)$.

  Let $S_i=\bigcup_{j=1}^{m} S_{i,j}$ be the vertices of $i$th constituent tree $\mathcal{T}_i$ as defined in Definition~\ref{defn:randomG_build} and $T_i=T\cap S_i$. 
  Let $A_i=T_i\cap S_{i,m}$ and $B_i=T_i\backslash A_i$. In other words, $A_i$ contains the elements of $T_i$ (and therefore $T$) that are also the leaf vertices of the constituent tree $\mathcal{T}_i$, while $B_i$ contains the elements of $T_i$ that are internal vertices of $\mathcal{T}_i$.
  We want to show that $T$ has at least $\Omega(|T|/n^3)$ adjacent vertices that are not in itself. Therefore, we assume towards contradiction that with probability at least $\exp(-o(n))$, we can find a subset $T$ such that 
  \begin{equation}
      \label{eq:assumption_contradict}
      |\N(T)\setminus T|=o(|T|/n^3),\quad |T|\leq |\mathcal{V}/2|.
  \end{equation}

  Observe that the number of edges out of the set $T$ is at least $\big||A_i|-|A_{i+1}|\big|$ in the bipartite part between trees $i$ and $i+1$, because the leaves of $\mathcal{T}_i$ and $\mathcal{T}_{i+1}$ are linked through $(d-1)/2$ random perfect matchings, each of which is a bijection. Therefore we have 
  \[
  \big||A_i|-|A_{i+1}|\big|=o(|T|/n^3).
  \] 
  Hence for each $A_i$ and $A_j$, we have $$\big||A_i|-|A_{j}|\big|=o(|T|/n^{2}).$$
  
  Without loss of generality, assume that $|T_{1}|\geq  |T_i|$ for $i\in [n]$, then we have $|T_1|\geq \frac{1}{n}|T|$. 
  \begin{itemize}
      \item We first consider $|A_1|\leq \frac{1}{3}|T_1|$. Note that because $B_1$ is a subset of internal vertices of $\mathcal{T}_1$, whose vertices have at least $d-2$ children each. With this fact we can prove that $|\Gamma(B_1)\setminus B_1|\geq (d-2)|B_1|\geq |B_1|\geq \frac{2}{3}|T_1|$. 

        This can be proven by considering each connected component of $B_1$. By induction on the size of a connected component $S$ we can show that it has at least $(d-2)|S|$ child vertices that are not contained in $S$. Because different connected components cannot share child vertices due to the tree structure of $\mathcal{T}_1$, we can show that $|\Gamma(B_1)\setminus B_1|\geq (d-2)|B_1|$.
          
      Again because the vertices in $B_1$ are all internal vertices of $\mathcal{T}_1$, we have $(\N(B_1)\setminus B_1)\cap T = (\N(B_1)\setminus B_1)\cap A_1$. Thus there are at least $\frac{1}{3}|T_1|$ vertices inside $\N(B_1)$ that are also outside the set $T$, that is, $|\N(T)\setminus T|\geq \frac{1}{3}|T_1| \geq \frac{1}{3n}|T|$. This contradicts the assumption in \eqref{eq:assumption_contradict}.
      
      \item If $|A_1|\geq \frac{1}{3}|T_1|\geq \frac{1}{3n}|T|$, we have $|A_1| \leq \frac{10}{9}\frac{1}{n}|T| $. This is true by the following argument: Note that $|A_1|+|A_2|+\ldots +|A_n| \leq |T|$ and $||A_i|-|A_{j}||=o(|T|/n^{2})$, so $n|A_1|\leq |T|+o(|T|/n)$.
  Hence $|A_1|\leq \frac{1}{n}|T|+ o(|T|/n^{2}) \leq \frac{10}{9}\cdot \frac{1}{n}|T| $.

     Since $|T|\leq |\SV|/2$, $|\SV|/n=\frac{d-1}{d-2}N$ and $\frac{d-1}{d-2} \leq \frac{6}{5}$, we have $|A_1|\leq \frac{10}{9}\cdot \frac{1}{n}|T|\leq \frac{10}{9}\cdot \frac{1}{n}\frac{|\SV|}{2} = \frac{5}{9} \frac{d-1}{d-2}N \leq \chi N$. That is $|A_1|\leq \chi N$. 
     This tells us that $|A_1|$ is small enough for the expansion property in Lemma~\ref{lem:biparmildexp} (i) to hold.
     By Lemma~\ref{lem:biparmildexp} (i), and because $A_1$ is contain in one side of the bipartite regular graph between $S_{1,m}$ and $S_{2,m}$, we know that 
     \begin{equation}
     \label{eq:expansion_property_hold_ineq}
         |(S_{2,m}\cap\N(A_1))\setminus A_1|\geq (1+\delta)|A_1| \geq \frac{1}{3n}(1+\delta)|T|,
     \end{equation}
     with probability at least $1-\exp(-\Omega(n))$.
     Note that in the above equation it is redundant to exclude the set $A_1$, because $S_{2,m}\cap\N(A_1)$ does not intersect with $A_1\subset S_{1,m}$. We keep it there anyway to keep the notation consistent with Lemma~\ref{lem:biparmildexp}.
     The number of vertices in $S_{2,m}$ that are adjacent to $A_1$ and at the same time not in $T$ is 
     \[
     |(S_{2,m}\cap\N(A_1))\backslash T|
     \geq |(S_{2,m}\cap\N(A_1))\backslash A_1|-|A_2|\geq (1+\delta)|A_1|-|A_2|\geq \Omega(|T|/n^2),
     \]
     where in the second inequality we have used $||A_1|-|A_{2}||=o(|T|/n^3)$.
     Because $S_{2,m}\cap\N(A_1)\subset \N(T)$, the above inequalities imply that $|\N(T)\setminus T|\geq \Omega(|T|/n^2)$.
     Therefore for the assumption in \eqref{eq:assumption_contradict} to holds with probability at least $\exp(-o(n))$, we will need \eqref{eq:expansion_property_hold_ineq} to fail with probability at least $\exp(-o(n))$. But \eqref{eq:expansion_property_hold_ineq} holds with probability at least $1-\exp(-\Omega(n))$ by Lemma~\ref{lem:biparmildexp} (i), and therefore we have reached a contradiction.
  \end{itemize}
  
   Therefore, for any set $T$ with $|T|\leq |\SV|/2$, we have $|\N(T)\setminus T|=\Omega(|T|/n^3)$.
\end{proof}

\subsection{The enlarged sunflower graph}
\label{sec:the_enlarged_sunflower_graph}

On top of the sunflower graph constructed in the previous sections, we also add isolated vertices to the graph, as was done in \cite{childs2003ExpSpeedupQW,balasubramanian2023exponential}. 
The purpose is to ensure that when one randomly picks a vertex, with overwhelming probability, one will not be able to find a vertex that is in any way useful for finding a path. 


More precisely, we introduce a new set of vertices $\mathcal{V}_{\mathrm{aux}}$, and set $|\mathcal{V}_{\mathrm{aux}}|=\mathsf{N}_{\mathrm{aux}}$ to be any non-negative integer of our choosing. We then construct an \emph{enlarged sunflower graph}:
\begin{definition}[The enlarged sunflower graph]
    \label{defn:enlarged_sunflower_graph}
    Let the sunflower graph $\SG=(\SV,\SE)$ be the sunflower graph as defined in \cref{defn:randomG_build}. Let $\mathcal{V}_{\mathrm{aux}}$ be a set of $\mathsf{N}_{\mathrm{aux}}$ vertices. Then we call the graph $\tilde{\SG}=(\SV\cup \mathcal{V}_{\mathrm{aux}},\SE)$ the enlarged sunflower graph. 
\end{definition}
Note that we do not add any edge in this process, but only add on additional isolated vertices. When $\mathsf{N}_{\mathrm{aux}}=0$ we recover the original sunflower graph. When randomly picking a vertex from the enlarged sunflower graph $\tilde{\SG}$, the probability of obtaining a vertex in $\mathcal{V}$, i.e., a potentially useful vertex, is $\mathsf{N}_{\mathcal{G}}/(\mathsf{N}_{\mathcal{G}}+\mathsf{N}_{\mathrm{aux}})$, where $\mathsf{N}_{\mathcal{G}}=|\mathcal{V}|=n(d-1)^{m-1}$. In our classical lower bound proof in Section~\ref{sec:the_classical_lower_bound}, we choose $\mathsf{N}_{\mathrm{aux}} = \Omega(\mathsf{N}_{\mathcal{G}}^2)$ to make sure that this probability is exponentially small in $m$. The quantum algorithm works for any choice of  $\mathsf{N}_{\mathrm{aux}}$, and the query complexity is independent of its value, as can be seen from \cref{thm:the_algorithm}.
 
\subsection{The invariant subspace and effective Hamiltonian}
\label{sec:invariant_subspace_and_effective_ham}

In this section we will discuss the adjacency matrix $A$ of the regular sunflower graph $\SG$ in \cref{defn:randomG_build}. In particular, certain symmetries of this graph makes it possible for us to identify an invariant subspace of the adjacency matrix $A$. Within this invariant subspace we can replace $A$ with its restriction to this subspace, which we will call the \emph{effective Hamiltonian} and denote it by $H$. Below we will provide a precise definition of $H$ as well its structure. 


Hereafter we will need to view vertices in the enlarged sunflower graph $\tilde{\SG}$ defined in \cref{defn:enlarged_sunflower_graph} as quantum states, which is how they are stored on a quantum computer. 
The multigraph $\tilde{\SG}$ contains $\mathsf{N}=n(d-1)^{m-1} + \mathsf{N}_{\mathrm{aux}}$, where $\mathsf{N}_{\mathrm{aux}}$ can be chosen to be any non-negative integer. From the oracular models introduced in \cref{defn:classical_adjacency_list_oracle} and \cref{defn:quantum_adjacency_list_oracle}, these vertices each correspond to a bit-string of length $\mathsf{n}=\lceil\log_2(\mathsf{N})\rceil$.
Each vertex $v\in\SV\cup\SV_{\mathrm{aux}}$ is then stored on a quantum computer as a computational basis state $\ket{v}$ on at least $\mathsf{n}$ qubits (one is of course free to add on more qubits). These $\mathsf{n}$ qubits give rise to a $2^{\mathsf{n}}$-dimensional Hilbert space which we denote by $\mathcal{H}$.




Each supervertex defined in \cref{defn:supervertex} also corresponds to a quantum state, which we call the \emph{supervertex state} is defined as 
\begin{equation}
\label{eq:supervertex_state_defn}
    \ket{S_{i,j}} = \frac{1}{\sqrt{s_{i,j}}} \sum_{ v \in S_{i,j}}\ket{v},
\end{equation}
where the value of $s_{i,j}=|S_{i,j}|$ is given in \eqref{eq:cardinality_supervertex}.
These supervertex states are all orthogonal to each other because supervertices do not overlap.


\begin{definition}[\emph{Symmetric Subspace} $\mathcal{S}$] \label{def:SubspaceS}

We define the $mn$-dimensional \emph{symmetric subspace } $\mathcal{S}$ as follows
    \begin{equation}
    \label{eq:the_symmetric_subspace}
    \mathcal{S} = \mathrm{span}\{\ket{S_{i,j}}:1\leq i\leq n,1\leq j\leq m\}.
\end{equation}

\end{definition}



Below, we will show that the symmetric subspace $\mathcal{S}$ is in fact an invariant subspace of the adjacency matrix $A$.

\begin{lemma}
    \label{lem:invariant_subspace}
    Let $\tilde{\SG}$ be the enlarged sunflower graph defined in \cref{defn:enlarged_sunflower_graph}.
    Let $A$ be the adjacency matrix of the this multigraph as defined in \cref{defn:adjacency_matrix_of_multigraph}. Then for the supervertex state $\ket{S_{k,l}}$ defined in \eqref{eq:supervertex_state_defn}, we have
    \[
    A\ket{S_{k,l}} = \sum_{i=1}^n\sum_{j=1}^m\frac{e_{ij,kl}}{\sqrt{s_{i,j}s_{k,l}}}\ket{S_{i,j}},
    \]
    where $e_{ij,kl}$ is the number of edges linking vertices in $S_{i,j}$ (defined in \cref{defn:supervertex}) to those in $S_{k,l}$, and $s_{i,j}=|S_{i,j}|$.
\end{lemma}

\begin{proof}
    We first observe that we only need to concern ourselves with the subgraph that is the sunflower graph $\SG$, since the isolated vertices are not contained, or connected to, any supervertex.
    \begin{align*}
    A\ket{S_{kl}} = & \frac{1}{\sqrt{s_{kl}}} \sum_{v \in S_{kl}} A \ket{v}
    = \frac{1}{\sqrt{s_{kl}}} \sum_{v \in S_{kl}} \sum_{u} A_{u,v}\ket{u} \\ 
    = &\sum_{ij} \frac{1}{\sqrt{s_{kl}}}  \sum_{u\in S_{ij}} \sum_{v \in S_{kl}} A_{u,v}\ket{u}  
    = \sum_{ij} \frac{1}{\sqrt{s_{kl}}}  \sum_{u\in S_{ij}} \frac{e_{ij,kl}}{s_{ij}}\ket{u} \\ 
    =&\sum_{ij} \frac{e_{ij,kl}}{\sqrt{s_{ij}s_{kl}}} \ket{S_{ij}}. 
 \end{align*}
 where the fourth equality comes from \eqref{eq:num_edges_connecting_u_to_Skl}, and the last equality comes from the definition of the supervertex state in \eqref{eq:supervertex_state_defn}.
\end{proof}

We define, with $e_{ij,kl}$ and $s_{i,j}$ as in Lemma~\ref{lem:invariant_subspace},
\begin{equation}
    \label{eq:effective_ham_matrix_entries}
    \tilde{H}_{ij,kl} = \frac{e_{ij,kl}}{\sqrt{s_{i,j}s_{k,l}}},
\end{equation}
and these numbers serve as matrix elements that completely determines how $A$ acts on the invariant subspace $\mathcal{S}$. Because they will play an important role in our analysis later on, we will compute them explicitly here.

\begin{lemma}
    \label{lem:effective_ham_matrix_entries_value}
    $\tilde{H}_{ij,kl}$ defined in \eqref{eq:effective_ham_matrix_entries} takes the following value: 
    
    \[
     \tilde{H}_{ij,kl} = \left\{\begin{array}{ll}
      1 & \mbox{if } j=l=1 \mbox{ and }  k=i + 1 \mbox{ for } 1 \leq i\leq n-1; \\
      1 & \mbox{if } j=l=1 \mbox{ and }  k= 1, i=n;\\
     \sqrt{d-1} & \mbox{if } 2 \leq j\leq m-1, 2 \leq l \leq m \mbox{ and } i=k;\\
     \sqrt{d-2} & \mbox{if } j=1, l=2 \mbox{ and } i=k; \\
       \frac{d-1}{2} & \mbox{if } j=l=m \mbox{ and }  k=i + 1 \mbox{ for } 1 \leq i\leq n-1;\\
       \frac{d-1}{2} & \mbox{if } j=l=m \mbox{ and } k= 1, i=n;\\
       0 & \text{otherwise;}
    \end{array}\right.  
   \]

   and $\tilde{H}_{ij,kl}=\tilde{H}_{kl,ij}$.
        
\end{lemma}
\begin{proof}
    Recall the values of $s_{i,j}$ given in \eqref{eq:cardinality_supervertex} and those of $e_{ij,kl}$ given in \eqref{eq:num_edges_bw_supervertices}, we can compute $\tilde{H}_{ij,kl} = \frac{e_{ij,kl}}{\sqrt{s_{i,j}s_{k,l}}}$ directly. 
\end{proof}

Any matrix can be regarded as a linear map, and we will look at the adjacency matrix $A$ from this viewpoint. Because the symmetric subspace $\mathcal{S}$ is an  invariant subspace of $A$ according to \cref{lem:invariant_subspace}, we can then \emph{restrict} the linear map $A$ to this subspace to obtain a well-defined linear map. 

\begin{definition}[The effective Hamiltonian]
    \label{defn:effective_hamiltonian}
    The effective Hamiltonian $H:\mathcal{S}\to\mathcal{S}$ is the restriction of the adjacency matrix $A:\mathcal{H}\to\mathcal{H}$ (\cref{defn:adjacency_matrix_of_multigraph}) of the enlarged sunflower graph (\cref{defn:enlarged_sunflower_graph}) to its invariant subspace $\mathcal{S}$ given in \eqref{eq:the_symmetric_subspace}. In other words, we define $H:\mathcal{S}\to\mathcal{S}$ to be
    \[
    H: \mathcal{S}\ni \ket{\phi}\mapsto A\ket{\phi}.
    \]

\end{definition}

We name $H$ the ``effective Hamiltonian'' because if we consider quantum dynamics described by the time evolution operator $e^{-iA t}$, then if the initial state is in $\mathcal{S}$ the dynamics will be completely captured by the restriction of $A$ to $\mathcal{S}$, and therefore $H$ serves as a Hamiltonian governing the time evolution. This was a key idea in the quantum walk algorithm in \cite{childs2003ExpSpeedupQW}.

We will use the following fact when proving the correctness and efficiency of our algorithm:
\begin{lemma} 
\label{lem:effective_ham_polynomial}
Let $A$ be the adjacency matrix of the regular sunflower graph and $H$ be the effective Hamiltonian restricted to the symmetric subspace $\mathcal{S}$ as defined \cref{def:SubspaceS}.  
Then for any vector $\ket{v} \in \mathcal{S}$, we have 
$f(A) \ket{v}= f(H) \ket{v}$ where $f(\cdot)$ is a polynomial function.
\end{lemma}
\begin{proof}
    By linearity, we only need to prove for monomials $f(x)=x^k$. We do so by induction on $k$. When $k=1$ we have $A\ket{v}=H\ket{v}$ by definition. If we have $A^{k-1}\ket{v}=H^{k-1}\ket{v}$, then 
    \[
    A^k\ket{v} = A(A^{k-1}\ket{v})=A(H^{k-1}\ket{v})=H(H^{k-1}\ket{v})=H^k\ket{v}.
    \]
\end{proof}

Compared to the adjacency matrix $A$, the effective Hamiltonian $H$ has the advantage that it only acts on a $mn$-dimensional subspace, and therefore can be represented by a matrix of size $mn\times mn$. We will next proceed to construct this representation and reveal further structure through it.

We first define an isometry $V_{\mathcal{S}}:\mathbb{C}^m\otimes \mathbb{C}^n \to \mathcal{S}$ through
\begin{equation}
    \label{eq:isometry_to_S_defn}
    V_{\mathcal{S}}: \sum_{i=1}^n\sum_{j=1}^m \Phi_{ij}\ket{\mathfrak{b}_j}\ket{\mathfrak{a}_i}\mapsto \sum_{i=1}^n\sum_{j=1}^m \Phi_{ij}\ket{S_{i,j}},
\end{equation}
for any $\Phi_{ij}\in\mathbb{C}$,
and where $\ket{\mathfrak{a}_i}\in\mathbb{C}^n$ is the $n$-dimensional vector with 1 on the $i$th entry and 0 everywhere else, and $\ket{\mathfrak{b}_j}\in\mathbb{C}^m$ is the $m$-dimensional vector with 1 on the $j$th entry and 0 everywhere else. 
Here we use $\ket{\mathfrak{b}_j}\ket{\mathfrak{a}_i}$ rather than $\ket{\mathfrak{a}_i}\ket{\mathfrak{b}_j}$ in order to reveal the block-tridiagonal structure of the matrix representation of the effective Hamiltonian that is to be introduced later.
It can be easily verified that this is a bijective isometry, and that its inverse $V_{\mathcal{S}}: \mathcal{S}\to \mathbb{C}^m\otimes \mathbb{C}^n$ is
\begin{equation}
    \label{eq:inverse_isometry_VS}
    V_{\mathcal{S}}^{-1}: \sum_{i=1}^n\sum_{j=1}^m \Phi_{ij}\ket{S_{i,j}} \mapsto  \sum_{i=1}^n\sum_{j=1}^m \Phi_{ij}\ket{\mathfrak{b}_j}\ket{\mathfrak{a}_i}.
\end{equation}

With this isometry, we now consider the linear map $V_{\mathcal{S}}^{-1} HV_{\mathcal{S}}$. This linear map maps $\mathbb{C}^m\otimes \mathbb{C}^n$ to itself, and therefore can be written as a $mn\times mn$ matrix. We therefore define
\begin{definition}
    \label{defn:matrix_representation_of_effective_ham}
    We call $\tilde{H} = V_{\mathcal{S}}^{-1} HV_{\mathcal{S}}$ the \emph{matrix representation} of the effective Hamiltonian $H$, for the effective Hamiltonian $H$ defined in \cref{defn:effective_hamiltonian} and $V_{\mathcal{S}}$ given in \eqref{eq:isometry_to_S_defn}.
\end{definition}

Because $V_{\mathcal{S}}$ is a bijective isometry, $H$ and $\tilde{H}$ are unitarily equivalent to each other, and therefore have the same spectrum. Their eigenvectors also have one-to-one correspondence: for any eigenvector $\ket{\Psi}$ of $H$, $V_{\mathcal{S}}^{-1}\ket{\Psi}$ is an eigenvector of $\tilde{H}$ corresponding to the same eigenvalue, and vice versa.

We will be able to obtain a more explicit characterization of $\tilde{H}$ compared to $H$, which helps us to use $\tilde{H}$ to analyze $H$. 

\begin{lemma}
    \label{lem:matrix_entries_of_H_tilde}
    The matrix entry of $\tilde{H}$ on the $((j-1)n+i)$th row and $((l-1)n+k)$th column is $\tilde{H}_{ij,kl}$ in \eqref{eq:effective_ham_matrix_entries}.
\end{lemma}

\begin{proof}
    We only need to observe that
    \[
    \tilde{H}\ket{\mathfrak{b}_l}\ket{\mathfrak{a}_k} = V_{\mathcal{S}}^{-1} H V_{\mathcal{S}} \ket{\mathfrak{b}_l}\ket{\mathfrak{a}_k}  = V_{\mathcal{S}}^{-1} H\ket{S_{k,l}} = V_{\mathcal{S}}^{-1} \sum_{ij}\tilde{H}_{ij,kl} \ket{S_{k,l}} = \sum_{ij}\tilde{H}_{ij,kl}\ket{\mathfrak{b}_j}\ket{\mathfrak{a}_i},
    \]
    where we have used \cref{lem:invariant_subspace} and the definition of $\tilde{H}_{ij,kl}$ in \eqref{eq:effective_ham_matrix_entries}. Also note that $\ket{\mathfrak{b}_j}\ket{\mathfrak{a}_i}$, when written as a vector in $\mathbb{C}^{nm}$, has 1 on the $((j-1)n+i)$th entry and 0 everywhere else.
\end{proof}

With the matrix entries available, we can now write down the matrix $\tilde{H}$ explicitly:
\begin{equation}
\label{eq:Htilde_explicit}
    \tilde{H} = \begin{pmatrix}
        D_0     & t_1 I &        &            &     \\
        t_1 I & 0     & t_2 I  &            &     \\
              & t_2 I & 0      & \dots     &     \\
              &       & \ddots & \ddots     & t_{m-1} I    \\
              &       &        & t_{m-1} I & \gamma D_0
    \end{pmatrix},
\end{equation}
where $D_0$ is the $n\times n$ adjacency matrix associated with the cycle graph $C_n$ given in \eqref{eq:defn_D0}, and $t_1=\sqrt{d-2}$, $t_2=t_3=\dots=t_{m-1}=\sqrt{d-1}$, $\gamma=\frac{d-1}{2}$. One may notice that $\tilde{H}$ has the same sparsity pattern, i.e., the position of the non-zero entries, as the adjacency matrix of the supergraph in Figure~\ref{fig:supergraph}. This is not a coincidence since $\tilde{H}_{ij,kl}\neq 0$ {if and }only if $S_{i,j}$ and $S_{k,l}$ are linked in the supergraph, as can be seen from \eqref{eq:effective_ham_matrix_entries}. We observe that $\tilde{H}$ is a block tridiagonal matrix, and this is useful for analyzing its spectral properties.




\section{Spectral properties of the effective Hamiltonian}\label{sec:spectral_properties}

In this section, we analyze the spectral properties of the effective Hamiltonian $H$ (defined in \cref{defn:effective_hamiltonian}) associated with the enlarged sunflower graph $\tilde{\SG}$ in \cref{defn:enlarged_sunflower_graph}. Specifically, we will show that there is a unique $0$-eigenvector of $H$ that overlaps the starting state $\ket{s}=\ket{S_{1,1}}$, and the spectral gap of $H$ around $0$ is bounded from below by $1/\poly(m,n)$. 



To achieve the above objectives it suffices to study the matrix $\tilde{H}$. Because of the definition $\tilde{H}=V_{\mathcal{S}}^{-1}HV_{\mathcal{S}}$ in \cref{defn:matrix_representation_of_effective_ham}, we know that $\tilde{H}$ and $H$ share the same spectrum, and there is a bijective correspondence between their respective eigenstates.
We have already written down $\tilde{H}$ as a block tridiagonal matrix. There is a more compact way to express $\tilde{H}$ which can help us understand the structure of its eigenvalues and eigenvectors.
We let
\begin{equation}
\label{eq:defn_D1}
    D_1 = 
    \begin{pmatrix}
         0     & t_1  &        &            &     \\
        t_1  & 0     & t_2   &            &     \\
              & t_2  & 0      & \ddots     &     \\
              &       & \ddots & \ddots     & t_{m-1}     \\
              &       &        & t_{m-1}  & 0
    \end{pmatrix}_{ m\times m}.
\end{equation}
Then we can also rewrite $\tilde{H}$ as follows:
\begin{equation}
\label{eq:Htilde_compact}
    \tilde{H} = (\ket{\mathfrak{b}_1}\bra{\mathfrak{b}_1} + \gamma\ket{\mathfrak{b}_m}\bra{\mathfrak{b}_m})\otimes D_0 + D_1\otimes I,
\end{equation}
where $\mathfrak{b}_j$ is the $m$-dimensional vector with 1 on the $j$th entry and 0 everywhere else.

We will first state the result we want to prove regarding the eigenvalues and eigenvectors of $\tilde{H}$:
\begin{theorem}
    \label{thm:spectral_properties_Htilde} 
    For the matrix $\tilde{H}$ in \eqref{eq:Htilde_explicit}, the following statements are true:
    \begin{itemize}
        \item[(i)] The non-zero eigenvalues of $\tilde{H}$ are bounded away from zero by $\Omega(1/(mn^2))$.
        \item[(ii)] Let $\ket{\Psi}$ be the normalized $0$-eigenvector of matrix $D_1$, and let $\ket{\Phi^{\mathrm{even}}}$ and $\ket{\Phi^{\mathrm{odd}}}$ be the two orthogonal $0$-eigenvectors of matrix $D_0$ defined as 
        $\ket{\Phi^{\mathrm{even}}}=(\Phi^{\mathrm{even}}_1,\Phi^{\mathrm{even}}_2,\ldots,\Phi^{\mathrm{even}}_{n})^\top$ and $\ket{\Phi^{\mathrm{odd}}}=(\Phi^{\mathrm{odd}}_1,\Phi^{\mathrm{odd}}_2,\ldots,\Phi^{\mathrm{odd}}_{n})^\top$, where
        \begin{equation}
        \label{eq:cos_sin_modes}
            \Phi^{\mathrm{even}}_l = \frac{1}{\sqrt{n/2}}\cos\left(l\pi/2\right),\quad \Phi^{\mathrm{odd}}_l = \frac{1}{\sqrt{n/2}}\sin\left(l\pi/2\right).
        \end{equation}
        
        The $0$-eigenspace of $\tilde{H}$ is $2$-dimensional and is spanned by the orthonormal basis 
        $$
        \ket{\chi^{\mathrm{even}}}=\ket{\Psi}\ket{\Phi^{\mathrm{odd}}},\quad \ket{\chi^{\mathrm{odd}}}=\ket{\Psi}\ket{\Phi^{\mathrm{even}}}.
        $$
       \item[(iii)] The two quantum states $\ket{\chi^{\mathrm{even}}}$ and $\ket{\chi^{\mathrm{odd}}}$ satisfy the following conditions: for all even $1\leq i\leq n$, we have $|\braket{\chi^{\mathrm{even}}|\mathfrak{b}_1,\mathfrak{a}_i}|=\Omega(1/\sqrt{mn})$,  $|\braket{\chi^{\mathrm{odd}}|\mathfrak{b}_1,\mathfrak{a}_i}|=0$.\footnote{Here we use the notation that $\ket{\mathfrak{b}_j,\mathfrak{a}_i}=\ket{\mathfrak{b}_j}\ket{\mathfrak{a}_i}$.}
      For all odd $1\leq i\leq n$, we have $|\braket{\chi^{\mathrm{odd}}|\mathfrak{b}_1,\mathfrak{a}_i}|=\Omega(1/\sqrt{mn})$, $|\braket{\chi^{\mathrm{even}}|\mathfrak{b}_1,\mathfrak{a}_i}|=0$.
    \end{itemize}
\end{theorem}
We will postpone the proof of this theorem to later and first discuss its implications.
This theorem immediate implies the following about the eigenvalues and eigenvectors of the effective Hamiltonian $H$: 
\begin{corollary}
    \label{cor:spectral_properties}
    For the effective Hamiltonian defined in \cref{defn:effective_hamiltonian}, the following statements are true
    \begin{itemize}
        \item[(i)] The non-zero eigenvalues of $H$ are bounded away from zero by $\Omega(1/(mn^2))$.
        
        \item[(ii)] The 0-eigenspace of $H$ is 2-dimensional, and is spanned by an orthonormal basis
        \[
        \ket{\eta^{\mathrm{even}}} = V_{\mathcal{S}}\ket{\chi^{\mathrm{even}}},\quad \ket{\eta^{\mathrm{odd}}} = V_{\mathcal{S}}\ket{\chi^{\mathrm{odd}}},
        \]
        where $\ket{\chi^{\mathrm{even}}}$ and $\ket{\chi^{\mathrm{odd}}}$ are from \cref{thm:spectral_properties_Htilde} (ii), and $V_{\mathcal{S}}$ is the isometry defined in \eqref{eq:isometry_to_S_defn}.
        
        \item[(iii)] The two quantum states $\ket{\chi^{\mathrm{even}}}$ and $\ket{\eta^{\mathrm{odd}}}$ satisfy the following conditions: for all even $1\leq i\leq n$, we have $|\braket{\eta^{\mathrm{even}}|S_{i,1}}|=\Omega(1/\sqrt{mn})$,  $|\braket{\eta^{\mathrm{odd}}|S_{i,1}}|=0$.
      For all odd $1\leq i\leq n$, we have $|\braket{\eta^{\mathrm{odd}}|S_{i,1}}|=\Omega(1/\sqrt{mn})$, $|\braket{\eta^{\mathrm{even}}|S_{i,1}}|=0$.
    \end{itemize}
\end{corollary}

\begin{proof}
    Because by \cref{defn:matrix_representation_of_effective_ham}, $\tilde{H}=V_{\mathcal{S}}^{-1}H V_{\mathcal{S}}$, $H$ and $\tilde{H}$ share the same spectrum. Therefore (i) is a direct consequence of \cref{thm:spectral_properties_Htilde} (i), and this fact also implies that 0 is an eigenvalue of $H$ with two-fold degeneracy. Since $\ket{\chi^{\mathrm{even}}}$ and $\ket{\chi^{\mathrm{odd}}}$ are eigenvectors of $\tilde{H}$, $\ket{\eta^{\mathrm{even}}}$ and $\ket{\eta^{\mathrm{odd}}}$ must be eigenvectors corresponding to the same eigenvalue, i.e., 0. Because $\ket{\chi^{\mathrm{even}}}$ and $\ket{\chi^{\mathrm{odd}}}$ form an orthonormal basis, $\ket{\eta^{\mathrm{even}}}$ and $\ket{\eta^{\mathrm{odd}}}$ must also form an orthonormal basis since $V_{\mathcal{S}}$ is an isometry and thus preserves the inner product. We therefore have (ii). For (iii), we only need to note the fact that because $\ket{S_{i,1}}=V_{\mathcal{S}}\ket{\mathfrak{b}_1,\mathfrak{a}_i}$, we have $\braket{\eta^{\mathrm{even}}|S_{i,1}}=\braket{\chi^{\mathrm{even}}|\mathfrak{b}_1,\mathfrak{a}_i}$ and $\braket{\eta^{\mathrm{odd}}|S_{i,1}}=\braket{\chi^{\mathrm{odd}}|\mathfrak{b}_1,\mathfrak{a}_i}$ for all $i$.
\end{proof}

From the above we can see that, to understand the properties of $H$, we only need to study $\tilde{H}$, which we will do next.
Using the compact form representation of $\tilde{H}$ in \eqref{eq:Htilde_compact}, the following \cref{lem:spectrum_of_H_general_structure} shows that the eigenvalues and eigenvectors of $\tilde{H}$ have a particular structure.
\begin{lemma}
\label{lem:spectrum_of_H_general_structure}
Let $\mu_l=2\cos(2l\pi/n)$ and let $\ket{\phi_l}$ be the corresponding eigenvector of $D_0$ (defined in \eqref{eq:defn_D0}), with $l=1,2,\ldots,n$.
Let $\lambda_j(a,b)$ be the $j$th smallest eigenvalue of $H_1(a,b)$ with $j=1,2,\ldots,m$, where 
    \begin{equation}
    \label{eq:H1ab}
        H_1(a,b) =   \begin{pmatrix}
         a     & t_1  &        &            &     \\
        t_1  & 0     & t_2   &            &     \\
              & t_2  & 0      & \ddots     &     \\
              &       & \ddots & \ddots     & t_{m-1}     \\
              &       &        & t_{m-1}  & b
    \end{pmatrix}_{ m\times m}.           
    \end{equation}

Let $\gamma=\frac{d-1}{2}$ and $\ket{\psi^l_j}$ be the $\lambda_j(\mu_l,\gamma\mu_l)$-eigenvector of $H_1(\mu_l,\gamma\mu_l)$ with $1 \leq l \leq n, 1\leq j\leq m$. The eigenvalues and eigenvectors of $\tilde{H}$ in \eqref{eq:Htilde_explicit}
 are $$\lambda_j(\mu_l,\gamma\mu_l)
 \text{ and } \ket{\psi^l_j}\ket{\phi_l} \text{ for } l=1,2,\ldots,n, j=1,2,\ldots,m.$$ 
 
 
\end{lemma}
\begin{proof}
We prove this lemma by examining each eigenvalue and eigenvector of $\tilde{H}$. That is, for each $l=1,2,\ldots,n, j=1,2,\ldots,m$,
    \begin{equation}
    \label{eq:verify_eigenvals_eigenvecs1}
    \begin{aligned}
         \tilde{H}\ket{\psi_j^l}\ket{\phi_l} &= \big[\big(\ket{\mathfrak{b}_1}\bra{\mathfrak{b}_1}+\gamma\ket{\mathfrak{b}_m}\bra{\mathfrak{b}_m}\big)\otimes D_0  + D_1\otimes I\big]\ket{\psi_j^l}\ket{\phi_l} \\ 
         &= \big[\big(\mu_l\ket{\mathfrak{b}_1}\bra{\mathfrak{b}_1} + \gamma\mu_l\ket{\mathfrak{b}_m}\bra{\mathfrak{b}_m} + D_1\big)\ket{\psi_j^l}\big]\otimes \ket{\phi_l} \\
         &= \lambda_j(\mu_l,\gamma\mu_l)\ket{\psi_j^l} \ket{\phi_l},
    \end{aligned}
    \end{equation}
    where in the second equation we have used $D_0\ket{\phi_l}=\mu_l\ket{\phi_l}$, and in the third equation we have used \begin{equation}
       H_1(\mu_l,\gamma\mu_l) \ket{\psi_j^l} = \big(\mu_l\ket{\mathfrak{b}_1}\bra{\mathfrak{b}_1} + \gamma\mu_l\ket{\mathfrak{b}_m}\bra{\mathfrak{b}_m} + D_1\big)\ket{\psi_j^l} =\lambda_j(\mu_l,\gamma\mu_l)\ket{\psi_j^l}.
    \end{equation}
    \cref{eq:verify_eigenvals_eigenvecs1} then provides us with $n\times m$ eigenpairs, which are all eigenpairs of $\tilde{H}$ because the dimension of the vector space is also $n\times m$.
\end{proof}

Using this lemma, we will identify the 0-eigenspace of $\tilde{H}$, and lower bound spectral gap around the eigenvalue 0.
Since the eigenvalues of $D_0$ are readily available as $\mu_l = 2\cos(2\pi l/n)$,
we can compute the spectrum of $\tilde{H}$ by computing the eigenvalues $\lambda_j(\mu_l,\gamma\mu_l)$ of the tridiagonal matrices $H_1(\mu_l,\gamma\mu_l)$ for $j=1,2,\ldots,m$, $l=1,2,\ldots,n$. 
Using this fact, we next show that $\lambda_j(\mu_l,\gamma\mu_l) =0$ if and only if $\mu_l=0$ and $\lambda_j(\mu_l,\gamma\mu_l) =\Omega(1/(mn))$ for $\mu_l\neq 0$. Specifically, we prove the first result by computing the determinant of $H$ and the latter by computing the inverse of $H_1(\mu_l,\gamma\mu_l)$ for the case of $\mu_l\neq 0$.

We note that the determinant of a general tridiagonal matrix can be computed through \cite[Theorem 2.1]{el2006inversion}. We apply this result to the special class of tridiagonal matrices $H_1(a,b)$ as defined in \cref{eq:H1ab}.
\begin{lemma}[Determinant of tridiagonal matrix $H_1(a,b)$ {\cite[Theorem 2.1]{el2006inversion}}]  \label{lem:detH_1}
Let $\beta =(\beta_0,\ldots,\beta_{m})$ be a $m+1$ dimensional vector defined as follows, 
  \begin{equation}
    \beta_i = \left\{\begin{array}{ll}
  1 & \mbox{if }  i=0 \\
 a & \mbox{if } i=1 \\
 - t_{i-1}^2 \beta_{i-2} & \mbox{if } i=2,\ldots,m-1; \\
  b- t_{m-1}^2 \beta_{m-2} & \mbox{if } i=m. \\
\end{array}\right.   \end{equation} 
The determinant of $H_1(a,b)$ is equal to $\beta_m$.
\end{lemma}

\begin{corollary} \label{cor:0detH_1}
Let $t_1=\sqrt{d-2}$, $t_2=t_3=\dots=t_{m-1}=\sqrt{d-1}$, $\gamma=\frac{d-1}{2}$ and let $m$ be an odd integer. For every $\mu_l = 2\cos(2\pi l/n)$ with $l=1,2\ldots, n$, the eigenvalue $\lambda_j(\mu_l,\gamma\mu_l)$ of $H_1(\mu_l,\gamma\mu_l)$ is equal to $0$ if and only if $\mu_l=0$. 
\end{corollary}
\begin{proof}
   By \cref{lem:detH_1} and $m$ being an odd integer, the determinant of $ H_1(a,b) $ is equal to $\beta_m =b+\prod_{k=1}^{\frac{m-1}{2}} t_{2k}^2 a$ when $m =1 \bmod 4$ or $\beta_m =b-\prod_{k=1}^{\frac{m-1}{2}} t_{2k}^2 a$ when $m =3 \bmod 4$. For the matrix $H_1(\mu_l,\gamma\mu_l)$, we have $$\beta_m =\gamma\mu_l\pm\prod_{k=1}^{\frac{m-1}{2}} t_{2k}^2 \mu_l = \left(\frac{d-1}{2} \pm (d-1)^{m-1}\right)\mu_l.$$
   
Since $\frac{d-1}{2} \pm (d-1)^{m-1}\neq 0$, we have $\beta_m=0$ if and only if $\mu_l=0$.
\end{proof}

When $\mu_l \neq 0$, we next show that all the eigenvalues of $H_1(\mu_l,\gamma\mu_l)$ are far away from $0$, that is, $|\lambda_j(\mu_l,\gamma\mu_l)| =\Omega(1/(mn^2))$. The idea is to compute the inverse of the tridiagonal matrix $H_1(\mu_l,\gamma\mu_l)$, for which we have the following lemma:
\begin{lemma}
    \label{lem:inverseH_spectral_norm}
    Let $H_1(a,b)$ be as defined in \eqref{eq:H1ab}, with odd $m$, $t_1=\sqrt{d-2}$, $t_2=t_3=\dots=t_{m-1}=\sqrt{d-1}$, $\gamma=\frac{d-1}{2}$, for integer $d\geq 3$. Let $\mu_l = 2\cos(2\pi l/n)$ with $l=0,1\ldots, n-1$ satisfying $\mu_l\neq 0$. Let $a=\mu_l$ and $b=\gamma\mu_l$. Then $\|H_1^{-1}(a,b)\|=\Or(mn^2)$.
\end{lemma}
The proof of this lemma can be found in Appendix~\ref{sec:spectrum_estimates}. Because the eigenvalues of $H_1^{-1}(a,b)$ are exactly $1/\lambda_j(\mu_l,\gamma\mu_l)$, for $j=1,2,\cdots,m$, the above lemma tells us that $1/|\lambda_j(\mu_l,\gamma\mu_l)|=\Or(mn^2)$, and hence $|\lambda_j(\mu_l,\gamma\mu_l)|=\Omega(1/(mn^2))$ for $l$ such that $\mu_l\neq 0$.

When $\mu_l=0$, the corresponding $\lambda_j(\mu_l,\gamma\mu_l)=\lambda_j(0,0)$ can still be non-zero, and the following lemma helps us bound these eigenvalues away from 0:

\begin{lemma}
    \label{lem:spectral_gap_D1}
    Let $D_1$ be as defined in \eqref{eq:defn_D1},
    in which $t_1=\sqrt{d-2}$, $t_2=\cdots =t_{m-1}=\sqrt{d-1}$.
    Then $D_1$ has a non-degenerate 0-eigenstate $\ket{\Psi}=(\Psi_1,\Psi_2,\cdots,\Psi_m)^\top$ where
    \[
    \Psi_j = \prod_{k=1}^{(j-1)/2}\left(-\frac{t_{2k-1}}{t_{2k}}\right)\Psi_1 = (-1)^{(j-1)/2}\sqrt{\frac{d-2}{d-1}}\Psi_1,\quad \text{for all odd }j\geq 2,
    \]
    and $\Psi_j=0$ for even $j$.
    Moreover, 0 is separated from the rest of the spectrum of $D_1$ by a gap of at least $2\sqrt{d-2}/(m-1)$.
\end{lemma}

The proof can be found in Appendix~\ref{sec:spectrum_estimates}. Because $H_1(0,0)=D_1$, the above lemma tells us that $|\lambda_j(0,0)|\geq 2\sqrt{d-2}/(m-1)$ for all $j$ such that $\lambda_j(0,0)\neq 0$. Moreover, this lemma also provides an explicit formula for $\ket{\Psi}$, i.e., the 0-eigenvector of $D_1$, which is needed for constructing the 0-eigenvector of $\tilde{H}$.

We will then summarize the above discussion into the proof of \cref{thm:spectral_properties_Htilde}.

\begin{proof}[Proof of \cref{thm:spectral_properties_Htilde}]
    From \cref{lem:spectrum_of_H_general_structure} we know that all eigenvalues of $\tilde{H}$ are of the form $\lambda_j(\mu_l,\gamma\mu_l)$, for $\mu_l=2\cos(2l\pi/n)$.
    We divide these eigenvalues into two categories: those with $\mu_l\neq 0$ or those with $\mu_l=0$. For those with $\mu_l\neq 0$, they are all bounded away from 0 by at least $\Omega(1/(mn^2))$ as a result of \cref{lem:inverseH_spectral_norm}. For those with $\mu_l=0$, they are either 0 or bounded away from 0 by at least $\Omega(1/m)$ as a result of \cref{lem:spectral_gap_D1}. Combining these two cases we have shown that all non-zero eigenvalues are bounded away from 0 by at least $\Omega(1/(mn^2))$. This proves (i).

    For $\lambda_j(\mu_l,\gamma\mu_l)=0$, we need $\mu_l=0$ by \cref{cor:0detH_1}. There are two values for $l$ that achieve this: $l=n/4$ and $l=3n/4$. From \cref{example:zero_eigvec_cycle}, the corresponding 0-eigenvectors of $D_0$ are $\ket{\Phi^{\mathrm{even}}}$ and $\ket{\Phi^{\mathrm{odd}}}$ given in \eqref{eq:cos_sin_modes}.
    It is easy to check that $\braket{\Phi^{\mathrm{even}}|\Phi^{\mathrm{odd}}}=0$ and these two eigenvectors are both normalized. For $\lambda_j(0,0)=\lambda_j(\mu_l,\gamma\mu_l)$ with these two values of $l$, we know from \cref{lem:spectral_gap_D1} that (again observe $H_1(0,0)=D_1$) there is only one $j$ that satisfies $\lambda_j(0,0)=0$, corresponding to the eigenvector $\ket{\Psi}$ given in that lemma. Therefore \cref{lem:spectrum_of_H_general_structure} tells us that the 0-eigenspace of $\tilde{H}$ is 2-dimensional and spanned by $\ket{\Psi}\ket{\Phi^{\mathrm{even}}}$ and $\ket{\Psi}\ket{\Phi^{\mathrm{odd}}}$, and these two vectors are normalized and orthogonal to each other. 
    Therefore we have (ii).

    For (iii), note that 
    \[
    \begin{aligned}
        \braket{\chi^{\mathrm{even}}|\mathfrak{b}_j,\mathfrak{a}_i} &= \braket{\Psi|\mathfrak{b}_j}\braket{\Phi^{\mathrm{even}}|\mathfrak{a}_i} = \frac{1}{\sqrt{n/2}}\Psi_j^* \cos(i\pi/2), \\
    \braket{\chi^{\mathrm{odd}}|\mathfrak{b}_j,\mathfrak{a}_i} &= \braket{\Psi|\mathfrak{b}_j}\braket{\Phi^{\mathrm{odd}}|\mathfrak{a}_i} = \frac{1}{\sqrt{n/2}}\Psi_j^* \sin(i\pi/2).
    \end{aligned}
    \]
    When $i$ is even, $|\cos(i\pi/2)|=1$ and $\sin(i\pi/2)=0$. Therefore $|\braket{\chi^{\mathrm{even}}|\mathfrak{b}_j,\mathfrak{a}_i}|=|\Psi_j|/\sqrt{n/2}$ while $\braket{\chi^{\mathrm{odd}}|\mathfrak{b}_j,\mathfrak{a}_i}=0$. Because $\ket{\Psi}$ is normalized, we have
    \[
    1=\sum_{j=1}^m |\Psi_j|^2 = |\Psi_1|^2\left(1+\frac{d-2}{d-1}\frac{m-1}{2}\right).
    \]
    Therefore $|\Psi_1|=\Omega(1/\sqrt{m})$. Consequently $|\braket{\chi^{\mathrm{even}}|\mathfrak{b}_j,\mathfrak{a}_i}|=\Omega(1/\sqrt{mn})$. When $i$ is odd, the corresponding statements can be proved in the same way.
\end{proof}
\section{The algorithm}\label{sec:Algorithm}


In this section, we provide a quantum algorithm to find an \st path in the enlarged sunflower graph $\tilde{\SG}$ defined in \cref{defn:enlarged_sunflower_graph} and show that this quantum algorithm requires only polynomial queries to the oracles in \cref{defn:quantum_adjacency_list_oracle} and \cref{defn:exit_oracle}. Because of the freedom to choose the $\mathsf{N}_{\mathrm{aux}}$, i.e., the number of isolated vertices in $\tilde{\SG}$, we can choose $\mathsf{N}_{\mathrm{aux}}=0$ and see that the algorithm works just as well for the sunflower graph defined in \cref{defn:randomG_build}.
The algorithm that we will present is reminiscent of the two-measurement algorithm in \cite{childs2002quantum}. 
Here we first present the high-level idea. 

The input is the black-box unitaries $O^Q_{\tilde{\SG},1}$ and $O^Q_{\tilde{\SG},2}$ (\cref{defn:quantum_adjacency_list_oracle}) implementing the adjacency list oracle in \cref{defn:classical_adjacency_list_oracle} of the graph $\SG$, a vertex $s \in \SV$ that is the root of $\mathcal{T}_1$ (defined in \cref{defn:randomG_build}), and a indicator function $f_t$ (defined in \cref{defn:exit_oracle}). 
The vertex $t$ could also be given as input but it is not necessary to know it in advance for the algorithm. 

The algorithm works as follows:
Starting from the initial state $\ket{s}$ we apply the a polynomial function of the adjacency matrix $A$ that has the effect of filtering out all eigenvectors corresponding to non-zero eigenvalues. 
Because this is a non-unitary operation, it only succeeds with probability approximately the overlap between $\ket{s}$ and the $0$-eigenspace of the effective $H$ (defined in \cref{defn:effective_hamiltonian}).
While not necessary for the exponential speedup, the success probability can be boosted from $\Omega(1/\poly(n,m))$ to close to $1$ using fixed-point amplitude amplification \cite{yoder2014FixedPointSearch}. 
Upon successfully preparing the projected state, we then measure it in the computational basis.
This returns the bit-string representing the vertex $S_{i,1}$, i.e., the root of the tree $\mathcal{T}_i$, for odd $i$, each with probability at least $\Omega(1/\poly(m,n))$.
Therefore, repeating this process gives samples that cover all odd vertices $S_{i,1}$ for odd $i$ along the target \st path.\footnote{Strictly speaking $S_{i,1}$ is a supervertex defined in \cref{defn:supervertex} rather than a vertex. However, given that $S_{i,1}$ contains only one vertex, we will use $S_{i,1}$ to refer to the vertex contained in it.}
To fill in the remaining even numbered vertices, we query all neighbors of the vertices in the previous step.
$t$ will then be included among the these vertices with large probability, which we identify through querying $f_t$ $\poly(m,n)$ times.
Then a Breadth First Search of this graph gives a path from $s$ to $t$.

\begin{algorithm}
\caption{Finding an \st path in the enlarged sunflower graph $\tilde{\SG}$}
\begin{algorithmic}[1]\label{alg:mildexpanderfinding}
\REQUIRE Oracles $O^Q_{\tilde{\SG},1}$ and $O^Q_{\tilde{\SG},2}$ (\cref{defn:quantum_adjacency_list_oracle}), a vertex $s \in \SV$ that is the root of $\mathcal{T}_1$ (\cref{defn:randomG_build}), and a classical oracle $f_t$ (\cref{defn:exit_oracle}).

\ENSURE  The set of vertices of an \st path.
\STATE Construct circuit unitary ${\mathcal{V}}_{\mathrm{circ}}$ acting on registers $\alpha,\beta_1,\beta_2$ using Theorem~\ref{thm:robust_subspace_eigenstate_filtering} that satisfies, as in \eqref{eq:state_prepared_through_filtering},
\[
{\mathcal{V}}_{\mathrm{circ}}\ket{0}_{\alpha\beta_1\beta_2} = \ket{0}_{\alpha\beta_1}\Pi_0 \ket{s}_{\beta_2} + \ket{\perp},
\]
where $(\bra{0}_{\alpha\beta_1}\otimes I_{\beta_2})\ket{\perp}=0$.
\STATE Apply fixed-point amplitude amplification to $\tilde{\mathcal{V}}_{\mathrm{circ}}$ to prepare the state, as in \eqref{eq:state_prepared_after_amplitude_amplification},
\[
\frac{\Pi_0 \ket{s}}{\|\Pi_0 \ket{s}\|} = \ket{\eta^{\mathrm{odd}}}
\]
by boosting the probability of getting the all-0 state upon measuring the registers $\alpha,\beta_1$ close to 1.
\STATE Let $\mathcal{M}=\emptyset$.
\FOR{$\chi=1,2\cdots,N_s$}
\STATE Generate the state $\ket{\eta^{\mathrm{odd}}}$ through Steps 1 and 2.
\STATE Measure the register $\beta_2$ to obtain a vertex of $\tilde{\SG}$, which is added to $\mathcal{M}$
\ENDFOR
\STATE Use all vertices in $\mathcal{M}$ and all their neighbors to generate $\mathcal{G}_{\mathrm{samp}}$ as a subgraph of $\tilde{\SG}$.
\STATE Find $t$ by querying $f_t$ for each vertex in $\mathcal{G}_{\mathrm{samp}}$. If $t$ cannot be found then abort.
\STATE Use Breadth First Search to find an $\st$ path in $\mathcal{G}_{\mathrm{samp}}$. If the path can be found then return the path. If not then abort.
\end{algorithmic}
\end{algorithm}

We first give the following result on how to construct the polynomial function of $A$ that filters out the unwanted eigenvectors. Given the block encoding of the adjacency matrix $A$, the construction largely follows \cite[Theorem 3]{lin2019OptimalQEigenstateFiltering}, which we modify to account for the fact that we only have an approximate block encoding of $A$ and that we only have guarantees about the spectral gap of $H$ rather than $A$.

\begin{theorem}[Robust subspace eigenstate filtering]
    \label{thm:robust_subspace_eigenstate_filtering}
    Let $A$ be a Hermitian matrix acting on the Hilbert space $\mathcal{H}$, with $\|A\|\leq \alpha$. Let $\mathcal{S}$ be an invariant subspace of $A$. Let $H$ be the restriction of $A$ to $\mathcal{S}$. We assume that $0$ is an eigenvalue of $H$ (can be degenerate), and is separated from the rest of the spectrum of $H$ by a gap at least $\Delta$.

    We also assume that $A$ can be accessed through its $(\alpha,\mathfrak{m},\epsilon_A)$-block encoding $U_A$ acting on $\beta_1,\beta_2$ as defined in \cref{defn:block_encoding}.
    Then there exists a unitary circuit $\mathcal{V}_{\mathrm{circ}}$ on registers $\alpha,\beta_1,\beta_2$ such that
     \begin{equation}
    \label{eq:robust_subspace_eigenstate_filtering}
        \|(\bra{0}_{\alpha\beta_1}\otimes I_{\beta_2})\mathcal{V}_{\mathrm{circ}}(\ket{0}_{\alpha\beta_1}\otimes \Pi_{\mathcal{S}})-\Pi_{0}\Pi_{\mathcal{S}}\|\leq \epsilon + \varsigma,
    \end{equation}
    where
    \begin{equation}
        \label{eq:robustness_err}
        \varsigma = \frac{16\ell^2\epsilon_A}{\alpha}\left[\log\left(\frac{2\alpha}{\epsilon_A}+1\right)+1\right]^2,
    \end{equation}
    and it uses $2\ell=\Or((\alpha/\Delta)\log(1/\epsilon))$ queries to (control-) $U_A$ and its inverse, as well as $\Or(\ell \mathfrak{m})$ other single- or two-qubit gates. In the above $\Pi_{0}:\mathcal{S}\to\mathcal{S}$ is the projection operator into the $0$-eigenspace of $H$, and $\Pi_{\mathcal{S}}:\mathcal{H}\to\mathcal{S}$ is the projection operator into the invariant subspace $\mathcal{S}$.
\end{theorem}


The proof of the above theorem is provided in Appendix~\ref{sec:robust_subspace_eigenstate_filtering}. Here we will explain how this result fit into the context of our path-finding algorithm. In our setting, $A$ is the adjacency matrix defined in \cref{defn:adjacency_matrix_of_multigraph}. $\mathcal{S}$ is the symmetric subspace defined in \eqref{eq:the_symmetric_subspace}. The restriction of $A$ to $\mathcal{S}$ is defined to be the effective Hamiltonian in \cref{defn:effective_hamiltonian}. The spectral gap of $H$ around the eigenvalue $0$ is guaranteed by \cref{cor:spectral_properties} (i) to be at least $\Omega(1/(mn^2))$. Therefore the first part of the conditions are all satisfied, with $\Delta=\Omega(1/(mn^2))$.

By \cref{lem:block_encoding_adjacency_matrix_new}, we can see that the required approximate block-encoding $U_A$ can be constructed with $\Or(1)$ queries to $O^Q_{\tilde{\SG},1}$ and $O^Q_{\tilde{\SG},2}$, $\Or(\mathsf{n}+\log^{2.5}(1/\epsilon_A))$ ancilla qubits in the register $\beta_1$, and $\Or(\mathsf{n}+\log^{1.5}(1/\epsilon_A))$ additional one- or two-qubit gates. Here because $|\mathcal{V}|=n(d-1)^{m-1}+\mathsf{N}_{\mathrm{aux}}$, we have $\mathsf{n} = \Or(\log_2(|\mathcal{V}|))=\Or(m\log(n)+\log(\mathsf{N}_{\mathrm{aux}}))$.

Since all the conditions in Theorem~\ref{thm:robust_subspace_eigenstate_filtering} are satisfied, it then ensures that we can perform the unitary $\mathcal{V}_{\mathrm{circ}}$ $2\ell=\Or((1/\Delta)\log(1/\epsilon))$ queries to $U_A$. 
Applying this circuit to the state $\ket{s}$, i.e., the bit-string representing the entrance $s$, using the fact that $\Pi_{\mathcal{S}}\ket{s}=\ket{s}$, we have
\begin{equation}
    \label{eq:state_prepared_through_filtering}
    \begin{aligned}
    &\mathcal{V}_{\mathrm{circ}}(\ket{0}_{\alpha\beta_1}\otimes \ket{s}) 
    =\mathcal{V}_{\mathrm{circ}}(\ket{0}_{\alpha\beta_1}\otimes \Pi_{\mathcal{S}})\ket{s} \\
    &=\ket{0}_{\alpha\beta_1}\otimes \Pi_0\Pi_{\mathcal{S}}\ket{s} + \ket{\perp}+\ket{\varepsilon} \\
    &=\ket{0}_{\alpha\beta_1}\otimes \Pi_0\ket{s} + \ket{\perp}+\ket{\varepsilon},
    \end{aligned}
\end{equation}
where $\ket{\perp}$ satisfies $(\bra{0}_{\alpha\beta_1}\otimes I_{\beta_2\gamma})\ket{\perp} = 0$, and $\|\ket{\varepsilon}\|\leq \epsilon + \varsigma$ where $\varsigma$ is defined in \eqref{eq:robustness_err_appendix}.

Our immediate goal is to approximately obtain the state $\Pi_0\ket{s}$, which is flagged by the all-0 state in registers $\alpha,\beta_1$. The most direct way to achieve this is to post-select the ancilla qubit measurement result after applying $\mathcal{V}_{\mathrm{circ}}$. This entails an overhead of 
$1/\|(\bra{0}_{\alpha\beta_1}\otimes I)\mathcal{V}_{\mathrm{circ}}(\ket{0}_{\alpha\beta_1}\otimes \ket{s})\|^2$.
A more efficient way is to apply fixed-point amplitude amplification. By \cite[Theorem~27]{gilyen2018QSingValTransfArXiv} the overhead is improved almost quadratically to
\begin{equation}
\label{eq:amplitude_amplification_overhead_first_expression}
    \Or(\log(1/\epsilon')/\|(\bra{0}_{\alpha\beta_1}\otimes I)\mathcal{V}_{\mathrm{circ}}(\ket{0}_{\alpha\beta_1}\otimes \ket{s})\|),
\end{equation}
where $\epsilon'$ is an error incurred through the amplitude amplification process.
Note that
\begin{equation}
\label{eq:success_amplitude_lower_bound}
    \|(\bra{0}_{\alpha\beta_1}\otimes I)\mathcal{V}_{\mathrm{circ}}(\ket{0}_{\alpha\beta_1}\otimes \ket{s})\| \geq \|\Pi_0\ket{s}\|-\|\ket{\varepsilon}\|=\Omega(\|\Pi_0\ket{s}\|),
\end{equation}
where we have chosen $\epsilon$ and $\epsilon_A$ so that 
\[
\|\ket{\varepsilon}\|\leq \epsilon + \varsigma\leq \|\Pi_0\ket{s}\|/2.
\]
Therefore the amplitude amplification overhead is
\begin{equation}
    \label{eq:overhead_amplitude_amplification}
 \Or(\|\Pi_0\ket{s}\|^{-1}\log(1/\epsilon')),
\end{equation}
We will then be able to obtain a state  $\ket{\psi_{\mathrm{target}}}$ on registers $\beta_2,\gamma$ satisfying
\begin{equation}
\label{eq:eigenstate_filtering_state_prep_err}
\left\|\ket{\psi_{\mathrm{target}}}-\frac{\Pi_0\ket{s}}{\|\Pi_0\ket{s}\|}\right\|\leq \epsilon' + \Or(\|\ket{\varepsilon}\|\|\Pi_0\ket{s}\|^{-1})=\epsilon'+\Or((\epsilon + \varsigma)\|\Pi_0\ket{s}\|^{-1}).
\end{equation}

From \cref{cor:spectral_properties} (ii), because the 0-eigenspace of $H$ is spanned by the orthogonal vectors $\ket{\eta^{\mathrm{even}}}$ and $\ket{\eta^{\mathrm{odd}}}$, we have
\[
\Pi_0 = \ket{\eta^{\mathrm{even}}}\bra{\eta^{\mathrm{even}}} + \ket{\eta^{\mathrm{odd}}}\bra{\eta^{\mathrm{odd}}}.
\]
Therefore
\[
\Pi_0\ket{s} = \ket{\eta^{\mathrm{even}}}\braket{\eta^{\mathrm{even}}|s} + \ket{\eta^{\mathrm{odd}}}\braket{\eta^{\mathrm{odd}}|s} = \ket{\eta^{\mathrm{odd}}}\braket{\eta^{\mathrm{odd}}|s},
\]
where we have used the fact that $\ket{s}=\ket{S_{1,1}}$ and therefore have no overlap with $\ket{\eta^{\mathrm{even}}}$ according to \cref{cor:spectral_properties} (iii).
Consequently
\begin{equation}
\label{eq:state_prepared_after_amplitude_amplification}
    \frac{\Pi_0\ket{s}}{\|\Pi_0\ket{s}\|} = \ket{\eta^{\mathrm{odd}}}.
\end{equation}
This enables us to compute the quantities we need for the amplitude amplification overhead in \eqref{eq:overhead_amplitude_amplification}:
\begin{equation}
\label{eq:overlap_for_state_prep}
    \|\Pi_0\ket{s}\| =| \braket{\eta^{\mathrm{odd}}|S_{1,1}}|=\Omega(1/\sqrt{mn}),
\end{equation}
where we have used \cref{cor:spectral_properties} (iii) for $i=1$.

From the above we can see that we have prepared a state $\ket{\psi_{\mathrm{target}}}$ that is close to the 0-eigenstate $\ket{\eta^{\mathrm{odd}}}$. We summarize the above into a lemma:
\begin{lemma}
    \label{lem:generate_eta_odd}
    Let $\tilde{\SG}$ be the enlarged sunflower graph as defined in \cref{defn:enlarged_sunflower_graph} with $s=S_{1,1}$ known,  $m$ odd, and $n$ an integer multiple of 4.
    There exists a unitary circuit using $\Or(m^{1.5}n^{2.5}\polylog{mn/\epsilon'\epsilon})$ queries to the quantum oracles defined in \cref{defn:quantum_adjacency_list_oracle}, $\Or(m\polylog{mn/\epsilon_A}+\log(\mathsf{N}_{\mathrm{aux}}))$ qubits, and $\Or(m^{2.5}n^{2.5}\polylog{mn/\epsilon'\epsilon\epsilon_A}+\polylog{\mathsf{N}_{\mathrm{aux}}mn/\epsilon'\epsilon\epsilon_A})$ additional one- or two-qubit gates, such that after running the circuit with the all-0 state as the initial state and measuring some of the qubits, we obtain a quantum state $\ket{\psi_{\mathrm{target}}}$ such that
    \[
    \|\ket{\psi_{\mathrm{target}}}-\ket{\eta^{\mathrm{odd}}}\|\leq \epsilon'+\Or((\epsilon+\varsigma)\sqrt{mn}),
    \]
    where $\varsigma$ is defined in \cref{eq:robustness_err}, in which $\alpha=\Theta(1)$ and $\ell=\Or(mn^2\log(1/\epsilon))$.
\end{lemma}

\begin{proof}
    We use the eigenstate filtering technique as discussed above in \cref{thm:robust_subspace_eigenstate_filtering}. Implementing the circuit $\mathcal{V}_{\mathrm{circ}}$ requires $2\ell=\Or((\alpha/\Delta)\log(1/\epsilon))$ queries to the block encoding of the adjacency matrix $A$, which translates to $\Or((\alpha/\Delta)\log(1/\epsilon))$ to the adjacency list oracle because of \cref{lem:block_encoding_adjacency_matrix_new}. From  \cref{lem:block_encoding_adjacency_matrix_new} we also know that $\alpha=d^2=\Theta(1)$. By \cref{cor:spectral_properties} (i) we have $\Delta=\Omega(1/mn^2)$. We therefore have
    \[
    \ell = \Or(mn^2\log(1/\epsilon)),
    \]
    and this is also the expression for the query complexity of one application of $\mathcal{V}_{\mathrm{circ}}$. To prepare the quantum state $\ket{\psi_{\mathrm{target}}}$, we need to either post-select the measurement outcome of ancilla qubits after applying $\mathcal{V}_{\mathrm{circ}}$ or use fixed-point amplitude amplification, the latter of which being more efficient. By \cref{eq:overhead_amplitude_amplification} and \cref{eq:overlap_for_state_prep}, the amplitude amplification overhead is therefore $\Or(\sqrt{mn}\log(1/\epsilon'))$. This enables us to write down the total query complexity for preparing the state $\ket{\psi_{\mathrm{target}}}$, which is
    \[
    \underbrace{\Or(\sqrt{mn}\log(1/\epsilon'))}_{\text{amplitude amplification overhead}}\times \underbrace{\Or(mn^2\log(1/\epsilon))}_{\text{implementing }\mathcal{V}_{\mathrm{circ}}} = \Or(m^{1.5}n^{2.5}\polylog{1/\epsilon\epsilon'}).
    \]
    The deviation of $\ket{\psi_{\mathrm{target}}}$ from $\ket{\eta^{\mathrm{odd}}}$ is computed in \eqref{eq:eigenstate_filtering_state_prep_err}, and substituting \eqref{eq:overlap_for_state_prep} into it we obtain the error bound in the statement of the lemma. The number of ancilla qubits and gates are computed in a similar way.
\end{proof}

Next, we will measure the $\beta_2$ register to sample graph vertices. The goal is to obtain vertices of the form $S_{i,1}$. The probability of obtaining any $S_{i,1}$ is
\begin{equation}
\label{eq:prob_getting_Si1}
|\braket{S_{i,1}|\psi_{\mathrm{target}}}|^2\geq |\braket{S_{i,1}|\eta^{\mathrm{odd}}}|^2 - 2\epsilon'-\Or((\epsilon + \varsigma)\|\Pi_0\ket{s}\|^{-1}).
\end{equation}
From \eqref{eq:state_prepared_after_amplitude_amplification} and \cref{cor:spectral_properties} (iii) for odd $i$,
\[
 \braket{S_{i,1}|\eta^{\mathrm{odd}}}=\Omega(1/\sqrt{mn}),
\]
Consequently, if we choose $\epsilon,\epsilon',\epsilon_A$ to be small enough so that the $2\epsilon'+\Or((\epsilon + \varsigma)\|\Pi_0\ket{s}\|^{-1})$ term on the right-hand side of \eqref{eq:prob_getting_Si1} is at most $|\braket{S_{i,1}|\eta^{\mathrm{odd}}}|^2/2$, then the probability of obtaining $S_{i,1}$ upon measurement, as described in \eqref{eq:prob_getting_Si1}, is at least
\[
|\braket{S_{i,1}|\eta^{\mathrm{odd}}}|^2/2 =\Omega(1/mn)
\]
for odd $i$ and when. We can see that it suffices to have, using \cref{eq:overlap_for_state_prep},
\[
\epsilon'=\Theta(1/(mn)), \quad \epsilon,\varsigma=\Theta(1/(mn)^{1.5}).
\]
Through \eqref{eq:robustness_err} we can see that it suffices to choose
\[
\epsilon_A = \Theta(1/(mn)^3).
\]
With these choices of parameters $\epsilon,\epsilon',\epsilon_A$ and \cref{lem:generate_eta_odd}, we now have the cost of generating a bit-string sample which with $\Omega(1/mn)$ probability yields $S_{i,1}$ for each odd $i$. We state this result in the following lemma:
\begin{lemma}
    \label{lem:generate_single_sample}
    Under the same assumptions as in \cref{lem:generate_eta_odd},
    there exists a unitary circuit using $\Or(m^{1.5}n^{2.5}\polylog{mn})$ queries to the adjacency list oracles defined in \cref{defn:quantum_adjacency_list_oracle}, $\Or(m\polylog{mn}+\log(\mathsf{N}_{\mathrm{aux}}))$ qubits, and $\Or(m^{2.5}n^{2.5}\polylog{mn}+\polylog{\mathsf{N}_{\mathrm{aux}}mn})$ additional one- or two-qubit gates, such that after running the circuit with the all-0 state as the initial state and measuring some of the qubits, we obtain a bit-string representing a vertex of $\tilde{\SG}$ that is $S_{i,1}$ with probability at least $\Omega(1/(mn))$, for all odd $1\leq i\leq n$.
\end{lemma}

With the ability to obtain each $S_{i,1}$ for odd $i$ with $\Omega(1/(mn))$ probability, we are now ready to prove the time complexity of our main algorithm.

\begin{theorem} 
\label{thm:the_algorithm}
Let $\tilde{\SG}$ be the enlarged sunflower graph as defined in \cref{defn:enlarged_sunflower_graph} with $s=S_{1,1}$ known,  $m$ odd, and $n$ an integer multiple of 4. Then with probability at least $2/3$, Algorithm~\ref{alg:mildexpanderfinding} finds an \st path with 
\[
\Or(m^{2.5}n^{3.5}\polylog{mn})
\]
queries to the quantum adjacency list oracle in \cref{defn:quantum_adjacency_list_oracle} and
\[
\Or(mn\log(n))
\]
queries to the oracle for $t$ in \cref{defn:exit_oracle},
using in total $\poly(m,n,\log(\mathsf{N}_{\mathrm{aux}}))$ qubits, $\poly(m,n,\log(\mathsf{N}_{\mathrm{aux}}))$ other primitive gates, and $\Or(mn\log(n))$ runtime for classical post-processing.
 \end{theorem}

\begin{proof}
By Lemma~\ref{lem:generate_single_sample}, we can run the procedure described above $N_s$ times to generate $N_s$ independent bit-strings, each of which corresponds to a vertex of $\tilde{\SG}$, and is equal to $S_{i,1}$ with probability at least $\mathfrak{p}=\Omega(1/(mn))$ for odd $i$.
From these vertices, we first find all their neighbors in the graph $\tilde{\SG}$. There are at most $dN_s$ of them, and finding all the neighbors take at most $dN_s=\Or(mn\log(n))$ queries to the adjacency list oracle.
From these vertices and their neighbors, we generate a subgraph $\mathcal{G}_{\mathrm{samp}}$ of the original graph $\mathcal{G}$ with them as vertices. 

For any odd $i$, the probability of $S_{i,1}$ not appearing among the samples is $(1-\mathfrak{p})^{N_s}$. By the union bound, the probability of at least one of the even vertices not showing up among the samples is at most
$
    (n/2)(1-\mathfrak{p})^{N_s}.
$
Therefore, we can choose $N_s= \Theta(mn\log(n))$ to ensure that we obtain all $S_{i,1}$ for odd $1\leq i\leq n$ with probability at least $2/3$. Because all vertices $S_{i,1}$ with even $i$ are neighbors of those with odd $i$, the constructed subgraph $\mathcal{G}_{\mathrm{samp}}$  contains the target \st path, i.e., $S_{i,1}$ for all $1\leq i\leq n$, with probability at least $2/3$.

It now remains to find the \st path from the subgraph  $\mathcal{G}_{\mathrm{samp}}$, which contains at most $dN_s=\Or(mn\log(n))$ vertices. Therefore a Breadth First Search can find the \st path in time $\Or(mn\log(n))$.

In total we need to use the unitary circuit in Lemma~\ref{lem:generate_single_sample} $\Or(mn\log(n))$ times, and therefore all cost metrics need to be multiplied by this factor, except for the number of qubits which remain the same.

\end{proof}

\section{The classical lower bound}
\label{sec:the_classical_lower_bound}



In this section, we show that no efficient classical algorithm can find an \st path in the enlarged sunflower graph $\tilde{\SG}=(\tilde{\SV},\SE)$ defined in \cref{defn:enlarged_sunflower_graph} when $\mathsf{N}_{\mathrm{aux}}=\Omega(\mathsf{N}_{\SG}^2)=n^2(d-1)^{2m-2}$. In fact,  we show an even stronger result that no efficient classical algorithm can find the exit $t$ efficiently in this regular sunflower graph.
We follow the lower bound proof in \cite{childs2003ExpSpeedupQW}, which has also been used to prove lower bound in similar contexts \cite{li2023exponential,li2023multidimensional}.    
Below we briefly summarize the high level ideas of the proof in \cite{childs2003ExpSpeedupQW}. 

First, by considering an enlarged graph $\tilde{G}$ instead of the regular sunflower graph $\SG$ defined in \cref{defn:randomG_build}, one can essentially restrict classical algorithm to explore the graph contiguously. This is because if the classical algorithm queries the oracle with a vertex label that has not been previously returned by the oracle, this label will be exponentially unlikely (with $\mathsf{N}_{\SG}/(\mathsf{N}_{\mathrm{aux}}+\mathsf{N}_{\SG})=\Or(n^{-1}(d-1)^{-m+1})$) to correspond to a vertex in the graph $\SG$. This means that a classical algorithm can only explore the graph $\SG$ by generating a connected subgraph of it. Next, one can show that before encountering a cycle in $\SG$, the behavior of the classical algorithm can be described by a \emph{random embedding} $\pi$ of a rooted $(d-1)$-ary tree $\mathcal{T}$ into $\SG$ as defined in \cite{childs2003ExpSpeedupQW}, which we restate with slight modification here:
\begin{definition}[Random embedding]
    \label{defn:random_embedding}
    Let $\mathcal{T}$ be a $(d-1)$-ary tree. A mapping $\pi$ from $\mathcal{T}$ to $\SG$ is a random embedding of $\mathcal{T}$ into $\SG$ if it satisfies the following:
    \begin{enumerate}
        \item The root of $\mathcal{T}$ is mapped to $s$.
        \item Let $v_1,\ldots,v_{d-1}$ be the children of the root, and $a_1,\ldots,a_{d-1}$ be the neighbors of $s$. Then $(\pi(v_1),\ldots,\pi(v_{d-1}))$ is an unbiased random permutation of $(a_1,\ldots,a_{d-1})$.
        \item Let $v$ be an internal vertex in $\mathcal{T}$ other than the root, let $v_1,\ldots,v_{d-1}$ be its children and $u$ its parent, and let $a_1,\ldots,a_{d-1}$ be the neighbors of $\pi(v)$ except for $\pi(u)$. Then $(\pi(v_1),\ldots,\pi(v_{d-1}))$ is an unbiased random permutation of $(a_1,\ldots,a_{d-1})$. 
    \end{enumerate}
    We also say that the embedding $\pi$ is \emph{proper} if it is injective, and say it \emph{exits} if $t\in\pi(\mathcal{T})$.
\end{definition}
It is therefore sufficient to show that a random embedding has to be exponentially sized in order to find a cycle or the exit $t$. For this we have the following lemma:

\begin{lemma}
    Let $\mathcal{T}$ be a rooted $d-1$-ary tree with $q$ vertices, with $q=(d-1)^{cn}$ and $c<1/4$. Assuming $m=n+1$, then a random embedding $\pi$ of this tree into $\SG$ as defined in Definition~\ref{defn:random_embedding} is improper or exits with probability at most $\Or((d-1)^{-(1/2-2c)n})$.
\end{lemma}

\begin{proof}
    First, we note that in order for $\pi$ to exit, there must exist a path in $\mathcal{T}$ leading from root to leaves that, when mapped to $\SG$ through $\pi$, moves right in $\SG$ along the bottom path for $n-1$ consecutive steps, or moves downward from the top part of the graph to the bottom for $m-1$ consecutive steps. Both involve probability at most $\Or((d-1)^{-n})$. Because there are $\binom{q}{2}=\Or(q^2)$ paths in $\mathcal{T}$ the probability of it exiting is therefore $\Or((d-1)^{-(1-2c)n})$.

    Next we consider the probability of the embedding being improper, i.e., there existing vertices $a\neq b$ of $\mathcal{T}$ such that $\pi(a)=\pi(b)$, which also means the presence of a cycle in the subgraph $\pi(\mathcal{T})$ of $\SG$. 
    Note that with the exception of the cycle at the bottom of the graph consisting of $S_{1,1},S_{2,1},\ldots,S_{n,1}$, a cycle must involve topmost level, the only part of the graph aside from the bottom cycle that does not consist of trees.
    We can exclude the bottom cycle from our consideration since finding it already involves finding the exit $t$, which we have established is unlikely in a random embedding. Now we will focus on a cycle that passes through the topmost level.
    We will show that in such a cycle, it is very unlikely for lower levels $1,2\ldots,n/2$ of $\SG$ to be involved in the cycle in $\pi(\mathcal{T})$, because that would require a path in $\mathcal{T}$ that, if we take the direction away from the root, when mapped through $\pi,$ travels downwards  consecutively for distance at least $n/2+1$. Given the tree structure in the levels from $n/2+1$ to $n$, for each path this happens with probability most $(d-1)^{-n/2}$. Because there are $\binom{q}{2}$ paths in $\mathcal{T}$, the probability of having a cycle in $\pi(\mathcal{T})$ that involves a vertex on levels $1,2\ldots,n/2$ is therefore at most
    \begin{equation}
    \label{eq:long_path_downward_probability_upper_bound}
        \binom{q}{2}(d-1)^{-n/2} = \Or((d-1)^{-(1/2-2c)n}).
    \end{equation}
    Hereafter we will only consider the scenario where the cycle in $\pi(\mathcal{T})$ does not involve a vertex on levels $1,2\ldots,n/2$.

    We consider each of the $\binom{q}{2}=\Or(q^2)$ pairs of vertices $a$ and $b$ in $\mathcal{T}$ separately. We will show that it is exponentially unlikely for $\pi(a)=\pi(b)$. We will introduce some notations: let $P$ be the path linking $a$ and $b$, and let $c$ be the vertex on this path that is closest to the root. Let $P_1$ be the path leading from $c$ to $a$, and $P_2$ the path leading from $c$ to $b$.

    We then divide all vertices in the supervertices $S_{i,n/2+1},S_{i,n/2+2}, \ldots, S_{i,n+1}$, i.e., all vertices in $\mathcal{T}_i$ that are at least $n/2+1$-distance from the root of $\mathcal{T}_i$ (defined in Definition~\ref{defn:randomG_build} and illustrated in Figure~\ref{fig:sunflowergraphd=3}), into $(d-2)(d-1)^{n/2}$ complete $d$-ary subtrees of height $n/2$, for $i=1,2,\ldots,n$. 


    Since if $\pi(a)=\pi(b)$ then $\pi(a)$ will be involved in a cycle in $\mathcal{T}$, this tells us that we only need to consider the situation where $\pi(a)$ is among these complete $d$-ary subtrees. Otherwise we will have the situation of a cycle going to lower levels, which we have established is exponentially unlikely. Then we consider two possible scenarios: $c$ being an ancestor of $a$ and $b$, or $c=a$ or $c=b$. We know that one of these two possibilities must be true because $c$ is at least as close to the root as $a$ and $b$. 
    
    In the first scenario where $c$ is an ancestor of $a$ and $b$, the path $\pi(P_1)$ must visit a sequence of these subtrees, which we denote by $S_{k_1},S_{k_2},\ldots S_{k_u}$, and similarly $\pi(P_2)$ visits $S_{l_1},S_{l_2},\ldots S_{l_v}$.  For any choice of $S_{k_1},S_{k_2},\ldots S_{k_{u-1}}$, suppose that $S_{k_{u-1}}$ is in the tree $\mathcal{T}_i$ (as defined in \cref{defn:randomG_build}), the next subtree $S_{k_u}$ will be uniformly randomly chosen among the subtrees that are within the adjacent trees $\mathcal{T}_{i-1\pmod{n}}$ and $\mathcal{T}_{i+1\pmod{n}}$ due to the random connectivity between the trees $\{\mathcal{T}_{i'}\}$ and the random embedding. There are $2(d-2)(d-1)^{n/2}$ such subtrees, and therefore the probability of $S_{k_u}$ being any one of them is $1/(2(d-2)(d-1)^{n/2})$. The same is also true for $S_{l_v}$. Because of the Markovian nature of the random embedding, $S_{k_u}$ and $S_{l_v}$ are independent of each other, and therefore the probability of $S_{k_u}=S_{l_v}$ is at most
    \begin{equation}
    \label{eq:subtrees_meet_probability_upper_bound}
        2(d-2)(d-1)^{n/2}\times \frac{1}{(2(d-2)(d-1)^{n/2})^2}=\frac{1}{2(d-2)(d-1)^{n/2}}.
    \end{equation}

    In the second scenario, without loss of generality we assume $c=a$. Then we denote the subtrees visited on the path from $c$ to $b$ as $S_{k_1},S_{k_2},\ldots S_{k_u}$. Conditional on $S_{k_1},S_{k_2},\ldots S_{k_{u-1}}$, again because the random embedding is Markovian, $S_{k_u}$ can be any one of the subtrees with probability at most $1/(2(d-2)(d-1)^{n/2})$. $S_{k_1}$ is one of these trees, and therefore \eqref{eq:subtrees_meet_probability_upper_bound} is also an upper bound for the probability of $S_{k_u}=S_{k_1}$.

    From the probability upper bound \eqref{eq:subtrees_meet_probability_upper_bound}, we can see that for any pair of $a$ and $b$ the probability of $\pi(a)=\pi(b)$ is at most $\Or((d-1)^{-n/2})$, unless there is a cycle going through levels $1,2,\ldots, n/2$. There are $\binom{q}{2}$ pairs of $a$ and $b$, and therefore the probability of $\pi$ being an improper embedding is at most
    \begin{equation}
        \binom{q}{2}\times \Or((d-1)^{-n/2}) = \Or((d-1)^{-(1/2-2c)n}),
    \end{equation}
    excluding the scenario of a cycle going through lower levels. The probability of such a cycle existing is upper bound by \eqref{eq:long_path_downward_probability_upper_bound}, and therefore the total probability of $\pi$ being an improper embedding is at most $\Or((d-1)^{-(1/2-2c)n})$ accounting for all situations.
\end{proof}

The above reasoning, which follows the same general idea as the classical lower bound proof in \cite{childs2003ExpSpeedupQW}, yields the following classical lower bound as stated in the following theorem:
\begin{theorem}
\label{thm:classical_lower_bound}
    For the enlarged sunflower graph $\SG$ defined in Definition~\ref{defn:enlarged_sunflower_graph}, with $m=n+1$, $\mathsf{N}_{\mathrm{aux}}=\Omega(n^2(d-1)^{2m-2})$, using $q=(d-1)^{cn}$ queries of the adjacency list oracle $O_{\SG}$ where $c<1/4$, the probability of a classical randomized algorithm finding a cycle or finding the vertex $t$ is at most $\Or((d-1)^{-(1/2-2c)n})$. 
\end{theorem}

\bibliographystyle{alpha}
\bibliography{cqw}

\appendix

\section{Simulating the sparse access oracle}
\label{sec:block_encoding_construction_from_adjacency_list}



In this section we will describe how to simulate the sparse access oracle using the adjacency list oracle, and thereby prove \cref{lem:adj-simu-sparse}.

\begin{proof}[Proof of \cref{lem:adj-simu-sparse}]

Recall from \cref{defn:sparse_access_oracle} that the oracles we want to simulate are
\begin{equation}
    \label{eq:sparse_access_oracle}
    O^Q_{\mathrm{loc}}:\ket{v,k}\mapsto \ket{v,O_{G,1}(v,k)},\quad O^Q_{\mathrm{val}}:\ket{v,v',0}\mapsto\ket{v,v',A_{v,v'}},
\end{equation}
where have used the fact that $r_{vk}$ in \cref{defn:sparse_access_oracle} is nothing other than $O_{G,1}(v,k)$ from \cref{defn:classical_adjacency_list_oracle}.
Note that \cite[Lemma~48]{gilyen2018QSingValTransfArXiv} requires two separate oracles for accessing the matrix elements by row and by column, but here we only need one $O^Q_{\mathrm{loc}}$ because $A$ is a real symmetric matrix. $O^Q_{\mathrm{val}}$ is used to access the value of the matrix entry.

Note that $O^{Q}_{G,2}$ can be directly used as $O^Q_{\mathrm{val}}$.
We will then construct $O^Q_{\mathrm{loc}}$ from the oracle $O^{Q}_{G,1}$. This essentially means uncomputing the index $k$. We will need to use the following fact: for every $v\in V$, the function $O_{G,1}(v,\cdot):k\mapsto O_{G,1}(v,k)$ is injective. This can be easily verified by checking all combinations of $k,k'$ to see that $O_{G,1}(v,k)=O_{G,1}(v,k')$ implies $k=k'$.

As a first step, we construct a circuit $U^{[1]}$ such that 
\[
U^{[1]}\ket{v,O_{G,1}(v,k),c} = \ket{v,O_{G,1}(v,k),c\oplus k}.
\]
In this construction we will require a circuit that performs
\[
U^{[2]}\ket{x_0,x_1,\cdots,x_{d-1},c} = \ket{x_0,x_1,\cdots,x_{d-1},c\oplus k},
\]
where $k=\min\{j\in\{0,1,\cdots,d-1\}:x_j=0\}$, i.e., $k$ is the index of the first $x_j$ to be $0$. This circuit implements a classical invertible function that can be efficiently computed, and therefore can be implemented using $\Or(\poly(d))$ elementary gates. With $U^{[2]}$, we can now construct $U^{[1]}$ as follows:
\[
\begin{aligned}
    &\ket{v}_{\gamma}\ket{O_{G,1}(v,k)}_{\delta}\bigotimes_{j=0}^{d-1}(\ket{0}_{\alpha_j}\ket{0}_{\beta{j}})\ket{c}_{\epsilon} \\
    &\mapsto \ket{v}_{\gamma}\ket{O_{G,1}(v,k)}_{\delta}\bigotimes_{j=0}^{d-1}(\ket{j}_{\alpha_j}\ket{0}_{\beta{j}}) \ket{c}_{\epsilon} \\
    &\xmapsto{O^{Q}_{G,1} \text{ on } \gamma,\alpha_j,\beta_j} \ket{v}_{\gamma}\ket{O_{G,1}(v,k)}_{\delta}\bigotimes_{j=0}^{d-1}(\ket{j}_{\alpha_j}\ket{O_{G,1}(v,j)}_{\beta{j}})\ket{c}_{\epsilon}  \\
    &\xmapsto{\mathrm{CNOT} \text{ on } \delta,\beta_j} \ket{v}_{\gamma}\ket{O_{G,1}(v,k)}_{\delta}\bigotimes_{j=0}^{d-1}(\ket{j}_{\alpha_j}\ket{O_{G,1}(v,k)\oplus O_{G,1}(v,j)}_{\beta{j}})\ket{c}_{\epsilon}  \\
    &\xmapsto{U^{[2]} \text{ on all } \beta_j, \epsilon} \ket{v}_{\gamma}\ket{O_{G,1}(v,k)}_{\delta}\bigotimes_{j=0}^{d-1}(\ket{j}_{\alpha_j}\ket{O_{G,1}(v,k)\oplus O_{G,1}(v,j)}_{\beta{j}})\ket{c\oplus k}_{\epsilon} \\
    &\xmapsto{\mathrm{CNOT} \text{ on } \delta,\beta_j} \ket{v}_{\gamma}\ket{O_{G,1}(v,k)}_{\delta}\bigotimes_{j=0}^{d-1}(\ket{j}_{\alpha_j}\ket{O_{G,1}(v,j)})\ket{c\oplus k}_{\epsilon} \\
    &\xmapsto{(O^{Q}_G)^\dag \text{ on } \gamma,\alpha_j,\beta_j} \ket{v}_{\gamma}\ket{O_{G,1}(v,k)}_{\delta}\bigotimes_{j=0}^{d-1}(\ket{j}_{\alpha_j}\ket{0}_{\beta{j}})\ket{c\oplus k}_{\epsilon}  \\
    &\mapsto \ket{v}_{\gamma}\ket{O_{G,1}(v,k)}_{\delta}\bigotimes_{j=0}^{d-1}(\ket{0}_{\alpha_j}\ket{0}_{\beta{j}})\ket{c\oplus k}_{\epsilon}
\end{aligned}
\]
In the above, $U^{[2]}$ can find the value $k$ because $O_{G,1}(v,k)\oplus O_{G,1}(v,j)=0$ only when $j=k$ due to the injectivity of $O_{G,1}(v,\cdot)$ that we have mentioned. Discarding the ancilla qubits in the $\alpha_j$ and $\beta_j$ registers, which have all been returned to the 0 state, we then have the desired unitary $U^{[1]}$.

Next, we will use $U^{[1]}$ to construct $O^Q_{\mathrm{loc}}$ as defined in \eqref{eq:sparse_access_oracle}.
\[
\begin{aligned}
    \ket{v,k,0} \xmapsto{O^{Q}_{G,1}} \ket{v,k,O_{G,1}(v,k)} \xmapsto{\mathrm{SWAP}} \ket{v,O_{G,1}(v,k),k} \xmapsto{U^{[1]}} \ket{v,O_{G,1}(v,k),0}.
\end{aligned}
\]
Therefore we have the desired $O^Q_{\mathrm{loc}}$ after discarding the ancilla qubits. In the construction we have used $d+1$ queries to $O^{Q}_{G,1}$ and $d$ queries to its inverse.
\end{proof}

\section{Expansion properties of a random bipartite graph} \label{sec:apppendmildbipar}

We will first present a technical lemma that we need to use.
\begin{lemma}[ ]\label{lem:logbino} Let $1 \leq s\leq \frac{2}{3}N$ and $\delta=1/(2\log N)$ 
    Let $p_s= \frac{(N-s)!\left((1+\delta)s)!\right)^{2}}{s!(N-\delta s)! \left((\delta s)!\right)^{3} }$, then 
    \[
   \log p_s \lesssim \log (3/2)(\delta-1)s.
    \]
\end{lemma}
We omit the proof of this lemma because it follows almost exactly the analysis of the proof in \cite[Theorem 4.1.1]{kowalski2019introduction}. 

Next, we show $D$-regular random bipartite graph is a mild expander graph as stated in Lemma~\ref{lem:biparmildexp}.  We restate the lemma here:

\begin{lemma-non}[Lemma~\ref{lem:biparmildexp} in Section~\ref{sec:randomgraph}]
    Let $L$ and $R$ be two sets of vertices with $|L|=|R|=N$. Link $L$ and $R$ through $D\geq 3$ random perfect matchings. Denote the resulting graph by $G_B = (V_B, E_B)$, and let $\chi=2/3$, $\delta =1/(2\log N)$. Then $G_B$ has the following expansion properties with probability $1-\Theta(1/N^{2\log (3/2)(1-\delta)})$:
    \begin{itemize}
        \item[(i)] For any subset $L'\subseteq L$ and $|L'|\leq \chi N$, we have $|\N(L')\backslash L'| \geq (1+\delta)|L'|$, where $\N(L')$ denotes the neighborhood of $L'$ as defined in Definition~\ref{defn:neighborhood}.
        \item[(ii)] For any subset $T\subseteq L\cup R$ and $|T| \leq  N$, we have  $|\N(T) \backslash T| \geq \frac{\delta}{2} |T|$. 
    \end{itemize}
    In other words, with probability $1-\Theta(1/N^{2\log (3/2)(1-\delta)})$, this bipartite regular graph is a mild expander graph.
\end{lemma-non}

 \begin{proof}
     To prove (i), it suffices to show that the probability of there existing $L'\subset L$ such that $|L'|\leq \chi N$ and $|\N(L')\backslash L'| < (1+\delta)|L'|$ is small. For simplicity and slight abuse of the notation, in the rest of the proof, we use $\delta|L|$ to represent $\lfloor \delta |L|\rfloor $. 

     In one of the random perfect matchings, let the neighbors of $L'$ in $R$ be $I_1$, then $|I_1| = |L'|$. Therefore $|\N(L')\backslash L'| \geq |L'|$.
     If $|\N(L')\backslash L'| \leq (1+\delta)|L'|$ with the $D$ random perfect matchings, it is necessary that at most $\delta |L'|$ of the values of the other $D-1$ random perfect matchings are outside the set $I_1$. Each random perfect match is chosen independently with probability $1/N!$. For a specific set $I_{\delta}$ of size $\delta|L'|$ , the probability that $ \N(L')\backslash L' \subseteq I_1\cup I_{\delta}$ is bounded as follows:
    \[
    \begin{aligned}
        Pr[ \N(L')\backslash L' \subseteq I_1\cup I_{\delta}] &\leq \left(\frac{((1+\delta)|L'|)((1+\delta)|L'|-1)\cdots ((1+\delta)|L'|-|L'|+1)!}{N!}\right)^{D-1} \\
        &=\left(\frac{((1+\delta)|L'|)!}{(\delta|L'|)!N!}\right)^{D-1}.
    \end{aligned}
     \]

Note that this probability decreases with $D$.
For the choice of the additional set of $I_{\delta}$ in $R$ of size $\delta |L'|$, there are in total $\binom{N}{ \delta |L'| }$ options. Therefore,

     \begin{align*}
    \Pr[\text{There exists $L' \subseteq L$ with $|\N(L')\backslash L'| < (1+\delta)|L'|$ }]  &\leq \sum_{L'\subset L, |L'|\leq \chi N} \binom{N}{\delta|L'|} \Pr[\N(L') \subseteq I_1\cup I_{\delta}]  \\
    &= \sum_{s=1}^{\chi N}\binom{N}{s}\binom{N}{\delta s} \left(\frac{(N- s)!((1+\delta)s)!}{(\delta s)! N! }\right)^{D-1}\\
    &= \sum_{s=1}^{\chi N} \left( \frac{(N-s)!}{N!}\right)^{D-3} \frac{(N-s)!\left((1+\delta)s)!\right)^{D-1}}{s!(N-\delta s)! \left((\delta s)!\right)^{D} }\\
    & \leq \sum_{s=1}^{\chi N} \frac{(N-s)!\left((1+\delta)s)!\right)^{2}}{s!(N-\delta s)! \left((\delta s)!\right)^{3} }.
    \end{align*}

 To upper bound the probability, we will bound the value of each term $p_s= \frac{(N-s)!\left((1+\delta)s)!\right)^{2}}{s!(N-\delta s)! \left((\delta s)!\right)^{3} }$. Let $\delta = 1/ (2\log N)$, we consider the following two cases of $s$, 
\begin{enumerate}
    \item  $1 \leq  s\leq 2 \log N$, we have $p_s = \frac{(N-s)!s!}{N!} = \frac{s (s-1)\cdots 2\cdot 1}{N (N-1)\cdots (N-s+1)} $ 
    and 
    \[\sum_{s=1}^{2\log N} p_s= \Theta(1/N).\]
    \item $2 \log N \leq  s\leq \chi N$, using the Stirling Formula $\log (k!)=k\log k - k+ \frac{1}{2}\log (2\pi k)+ O(1/k)$ and following Lemma \ref{lem:logbino}, we know that $\log p_s \lesssim \log (3/2)(\delta-1)s$. That is $p_s \approx 2^{-\log (3/2)(1-\delta)s} = a ^s$, where $ a=2^{-\log (3/2)(1-\delta)} <1 $ is a constant. 
    Thus  \[ \sum_{s=2\log N}^{ \chi N} a^s= a^{2\log N}\left(\frac{1-a^{\chi N- 2\log N+1}}{1-a}\right)=\Theta(N^{2\log a})= \Theta(1/N^{2\log (3/2)(1-\delta)}) .\]
\end{enumerate}
Combine the two cases, we have 
\[
\Pr[\text{There exists $L' \subseteq L$ with $|\N(L')\backslash L'| < (1+\delta)|L'|$ }] =\Theta(1/N^{2\log (3/2)(1-\delta)}).
\]
In other words, with probability at least $1-\Theta(1/N^{2\log (3/2)(1-\delta)})$, for any subset $L'\subset L$ and $|L'|\leq \chi N$, we have $|\N(L')\backslash L'| \geq (1+\delta)|L'|$.


Next we prove (ii), that is, for any subset $T\subseteq L\cup R$ with $|T|\leq N$, with high probability, we have $\N(T)\geq \delta |T|/2$. 
This provides us with the expansion property of an arbitrary subgraph of $G_B$ of size at most $N$.
Let $T=L' \cup R'$ with $L'=T\cap L$ and $R'=T\cap R$. 
Without loss of generality, assume that $|L'|\geq |R'|$, then $|L'|\geq  |T|/2$. 
\begin{itemize}
    \item If $\chi N \leq |L'|\leq  N$, by the injection property of a perfect matching and $|T|\leq N$, we know that $|\N(L')\backslash L'|\geq |L'|\geq \chi N$ and $|R'|\leq (1-\chi)N$. Therefore, there are at least $|L'|-|R'|\geq (2\chi-1)N$ neighbors of $L'$ out of $T$, that is $|\N(T)\backslash T|\geq (2\chi-1)N$.
    \item If $|L'|\leq \chi N$, we know that $|\N(L')\backslash L'|\geq (1+\delta)|L'|$.
    In addition to the vertices in the set $L'$, there are at most $|R'|$ neighbors of $L'$ that are in $T$. Therefore $L'$ has at least $|L'|-|R'|\geq \delta|L'|\geq \frac{\delta}{2}|T|$ neighbors that are not contained  in $T$. As a result, for any subset $T\subseteq L \cup R$ and $|T|\leq  N$, we have$|\N(T) \backslash T|\geq \frac{\delta}{2} |T|$. In other words, for any subset $T\subseteq L \cup R$, we have 
    \[
    |\N(T) \backslash T|\geq \frac{\delta}{2}  |T|.
    \]
\end{itemize}
Thus, for any subset $T\subseteq L\cup R$ with $|T|\leq N$, with probability $1-1/\Theta(N^{2\log (3/2)(1-\delta)})$, we have $|\N(T)\backslash T|\geq \delta |T|/2$.
\end{proof}

\section{Robust subspace eigenstate filtering}
\label{sec:robust_subspace_eigenstate_filtering}

We will use the quantum singular value transformation (QSVT) technique \cite{gilyen2018QSingValTransf,low2016HamSimQSignProc}, and the minimax filtering polynomial proposed in \cite{lin2019OptimalQEigenstateFiltering} to implement the projection operator to the $0$-eigenspace of the effective Hamiltonian. 
The filtering polynomial takes the form 
\begin{equation}
    \label{eq:minimax_filtering_polynomial_appendix}
    R_{\ell}(x;\delta)=\frac{T_{\ell}\left(-1+2\frac{x^{2}-\delta^{2}}{1-\delta^{2}}\right)}{T_{\ell}\left(-1+2\frac{-\delta^{2}}{1-\delta^{2}}\right)},
\end{equation}
where $T_\ell(x)$ is the $\ell$th Chebysehv polynomial of the first kind, and the degree of the polynomial is therefore $\ell$. This polynomial peaks at $x=0$ and is close to $0$ for $\delta\leq |x|\leq 1$. It is minimax in the sense that the deviation from $0$ for $\delta\leq |x|\leq 1$ is minimal among all polynomials of degree $\ell$ that take value $1$ at $x=0$.

The properties of $R_\ell(x;\delta)$ are stated in \cite[Lemma~2]{lin2019OptimalQEigenstateFiltering}, which we restate here:
\begin{lemma}[Lemma~2 of \cite{lin2019OptimalQEigenstateFiltering}]
\label{lem:minimax_poly}
Let $R_{\ell}(x;\delta)$ be as defined in \eqref{eq:minimax_filtering_polynomial_appendix}. It satisfies the following:
\begin{itemize}
    \item[(i)] $R_{\ell}(x;\delta)$ solves the minimax problem
\[
\underset{p(x)\in\mathbb{P}_{2\ell}[x],p(0)=1}{\mathrm{minimize}}  \max_{x\in\mathcal{D}_{\delta}}|p(x)|.
\]
\item[(ii)] $|R_{\ell}(x;\delta)|\leq2e^{-\sqrt{2}\ell\delta}$ for all $x\in\mathcal{D}_{\delta}$
and $0<\delta\leq1/\sqrt{12}$. Also $R_{\ell}(0;\delta)=1$.
\item[(iii)] $|R_{\ell}(x;\delta)|\leq1$ for all $|x|\leq1$.
\end{itemize}
\end{lemma}

A fact that we will use later is that, by Markov brothers' inequality,
\begin{equation}
    \label{eq:derivative_bound}
    \max_{-1\leq x\leq 1} \left|\frac{\mathrm{d}}{\mathrm{d} x}R_{\ell}(x;\delta)\right|\leq 4\ell^2\max_{-1\leq x\leq 1}|R_{\ell}(x;\delta)|\leq 4\ell^2,
\end{equation}
where we have used the fact that the degree of $R_{\ell}(x;\delta)$ is $2\ell$.

We will use the above polynomial to construct the projection operator. 

\begin{lemma}
    \label{lem:subspace_eigenstate_filtering_exact}
    Let $A$ be a Hermitian matrix acting on the Hilbert space $\mathcal{H}$, with $\|A\|\leq \alpha$. Let $\mathcal{S}$ be an invariant subspace of $A$. Let $H$ be the restriction of $A$ to $\mathcal{S}$. We assume that $0$ is an eigenvalue of $H$ (can be degenerate), and is separated from the rest of the spectrum of $H$ by a gap at least $\Delta$. Then
    \[
    \|R_{\ell}(A/\alpha;\delta)\Pi_{\mathcal{S}}-\Pi_{0}\Pi_{\mathcal{S}}\|\leq 2e^{-\sqrt{2}\ell\delta},
    \]
    where $R_{\ell}(x;\delta)$ is defined in \eqref{eq:minimax_filtering_polynomial_appendix}, $\delta=\min\{\Delta/\alpha,1\sqrt{12}\}$, $\Pi_{0}:\mathcal{S}\to\mathcal{S}$ is the projection operator into the $0$-eigenspace of $H$, and $\Pi_{\mathcal{S}}:\mathcal{H}\to\mathcal{S}$ is the projection operator into the invariant subspace $\mathcal{S}$.
\end{lemma}

\begin{proof}
    We first note that because $\mathcal{S}$ is an invariant subspace of $A$, and therefore is an invariant subspace of $A^k$. As a result, we have
    \[
    A^k\Pi_{\mathcal{S}} = H^k \Pi_{\mathcal{S}},
    \]
    where we have used the fact that the image of $\Pi_\mathcal{S}$ is $\mathcal{S}$ and that $H$ is the restriction of $A$ to $\mathcal{S}$.
    By linearity, this extends to all polynomials of $A$: for any polynomial $p(x)$, 
    \[
    p(A)\Pi_{\mathcal{S}} = p(H).
    \]
    Because $R_\ell(A/\alpha;\delta)$ is a polynomial of $A$, we have
    \[
    R_\ell(A/\alpha;\delta)\Pi_{\mathcal{S}} =   R_\ell(H/\alpha;\delta).
    \]
    Therefore 
    \[
    R_{\ell}(A/\alpha;\delta)\Pi_{\mathcal{S}}-\Pi_{0}\Pi_{\mathcal{S}} = (R_\ell(H/\alpha;\delta) - \Pi_0)\Pi_{\mathcal{S}}.
    \]
    because $\|\Pi_{\mathcal{S}}\|=1$ it suffices to prove that 
    \begin{equation}
        \label{eq:proof_goal_filtering_subspace}
        \|R_\ell(H/\alpha;\delta)-\Pi_0\|\leq 2e^{-\sqrt{2}\ell\delta},
    \end{equation}
    which can be done by examining each eigenstate of $H$, given that both $R_\ell(H/\alpha;\delta)$ and $\Pi_0$ are diagonal in the eigenbasis of $H$. For an eigenstate $\ket{\Psi}$ of $H$ and its corresponding eigenvalue $\lambda$, if $\lambda=0$, then
    \[
    (R_\ell(H/\alpha;\delta)-\Pi_0)\ket{\Psi} = R_\ell(0;\delta)\ket{\Psi} - \ket{\Psi} = 0.
    \]
    If $\lambda\neq 0$, then by the assumption of the spectral gap, we have $|\lambda|\geq \Delta/\alpha\geq \delta$. We therefore have
    \[
    \|(R_\ell(H/\alpha;\delta)-\Pi_0)\ket{\Psi}\|=\|(R_\ell(\lambda/\alpha;\delta)\ket{\Psi}\| = |R_\ell(\lambda/\alpha;\delta)|\leq 2e^{-\sqrt{2}\ell\delta},
    \]
    where we have used \eqref{lem:minimax_poly} (ii) and the fact that $\delta\leq 1/\sqrt{12}$. Therefore we have proved \eqref{eq:proof_goal_filtering_subspace}. 
\end{proof}

We will then combine the above lemmas with QSVT to provide a robust implementation of the $0$-eigenspace projection operator of $\tilde{H}$, as stated in Theorem~\ref{thm:robust_subspace_eigenstate_filtering}, which we restate here. 

\begin{theorem}[Robust subspace eigenstate filtering]
    Let $A$ be a Hermitian matrix acting on the Hilbert space $\mathcal{H}$, with $\|A\|\leq \alpha$. Let $\mathcal{S}$ be an invariant subspace of $A$. Let $H$ be the restriction of $A$ to $\mathcal{S}$. We assume that $0$ is an eigenvalue of $H$ (can be degenerate), and is separated from the rest of the spectrum of $H$ by a gap at least $\Delta$.

    We also assume that $A$ can be accessed through its $(\alpha,\mathfrak{m},\epsilon_A)$-block encoding $U_A$ acting on $\beta_1,\beta_2$ as defined in \cref{defn:block_encoding}.
    Then there exists a unitary circuit $\mathcal{V}_{\mathrm{circ}}$ on registers $\alpha,\beta_1,\beta_2$ such that
    \begin{equation}
    \label{eq:robust_subspace_eigenstate_filtering}
        \|(\bra{0}_{\alpha\beta_1}\otimes I_{\beta_2})\mathcal{V}_{\mathrm{circ}}(\ket{0}_{\alpha\beta_1}\otimes \Pi_{\mathcal{S}})-\Pi_{0}\Pi_{\mathcal{S}}\|\leq \epsilon + \varsigma,
    \end{equation}
    where
    \begin{equation}
        \label{eq:robustness_err_appendix}
        \varsigma = \frac{16\ell^2\epsilon_A}{\alpha}\left[\log\left(\frac{2\alpha}{\epsilon_A}+1\right)+1\right]^2,
    \end{equation}
    and it uses $2\ell=\Or((\alpha/\Delta)\log(1/\epsilon))$ queries to (control-) $U_A$ and its inverse, as well as $\Or(\ell \mathfrak{m})$ other single- or two-qubit gates. In the above $\Pi_{0}:\mathcal{S}\to\mathcal{S}$ is the projection operator into the $0$-eigenspace of $H$, and $\Pi_{\mathcal{S}}:\mathcal{H}\to\mathcal{S}$ is the projection operator into the invariant subspace $\mathcal{S}$.
\end{theorem}

\begin{proof}
    Similar to Lemma~\ref{lem:subspace_eigenstate_filtering_exact}, we let $\delta=\min\{\Delta/\alpha,1\sqrt{12}\}$.
    First, we denote $\tilde{A} = \alpha(\bra{0}_{\beta_1}\otimes I_{\beta_2})U_A(\ket{0}_{\beta_1}\otimes I_{\beta_2})$, then according to \cref{defn:block_encoding} we have
    \[
    \|\tilde{A}-A\|\leq \epsilon_A.
    \]
    Moreover, $U_A$ is now a $(\alpha,\mathfrak{m},0)$-block encoding of $\tilde{A}$. 
    Therefore we can implement a singular value transformation $R^{\mathrm{SV}}_{\ell}(\tilde{A}/\alpha;\delta)$ using quantum singular value transformation \cite[Theorem~2]{gilyen2018QSingValTransf}.\footnote{For an even function $f$, the singular value transformation of $A$ with singular value decomposition $A=U\Sigma V^\dag$ is defined to be $f^{\mathrm{SV}}(A)=V f(\Sigma)V^\dag$ \cite{gilyen2018QSingValTransfArXiv}.} More precisely, we can construct a circuit $\mathcal{V}_{\mathrm{circ}}$ such that
    \begin{equation}
        \label{eq:exact_block_encoding_filter_poly}
        (\bra{0}_{\alpha\beta_1}\otimes I_{\beta_2})\mathcal{V}_{\mathrm{circ}}(\ket{0}_{\alpha\beta_1}\otimes I_{\beta_2}) = R^{\mathrm{SV}}_{\ell}(\tilde{A}/\alpha;\delta).
    \end{equation}
    The same theorem also tells us that this circuit uses the block encoding $U_A$ $d=2\ell$ times, which is the degree of the polynomial $R_{\ell}(\cdot;\delta)$.\footnote{The readers may find \cite[Theorem~2]{gilyen2018QSingValTransf} ambiguous in the parameters used in the block encoding. We remark that in \cite{gilyen2018QSingValTransf}, they use ``block encoding of $A$'' to mean an $(1,m,0)$-block encoding in the context of our work. See \cite[Eq.~(1)]{gilyen2018QSingValTransf}.} 

    Note that instead of the singular value transformation of $\tilde{A}$, we actually want to implement an eigenvalue transformation of $A$. We observe that because $A$ is Hermitian, its singular value transformation coincides with its eigenvalue transformation, i.e., $R^{\mathrm{SV}}_{\ell}(A/\alpha;\delta)=R_{\ell}(A/\alpha;\delta)$. Moreover, $R^{\mathrm{SV}}_{\ell}(A/\alpha;\delta)$ is not far from $R^{\mathrm{SV}}_{\ell}(\tilde{A}/\alpha;\delta)$, which we will show next.

    By \eqref{eq:derivative_bound} and the intermediate value theorem, we have
    \[
    |R_{\ell}(x;\delta)-R_{\ell}(y;\delta)|\leq 4\ell^2|x-y|
    \]
    for any $x,y\in[-1,1]$. Since $R_{\ell}(\cdot;\delta)$ is an even polynomial, from \cite[Corollary 21]{gilyen2018QSingValTransfArXiv} we know that
    \[
    \|R_{\ell}(A/\alpha;\delta)-R^{\mathrm{SV}}_{\ell}(\tilde{A}/\alpha;\delta)\|=\|R^{\mathrm{SV}}_{\ell}(A/\alpha;\delta)-R^{\mathrm{SV}}_{\ell}(\tilde{A}/\alpha;\delta)\|\leq \varsigma,
    \]
    for $\varsigma$ given in \eqref{eq:robustness_err}.
    Combined with \eqref{eq:exact_block_encoding_filter_poly}, we have
    \begin{equation}
    \label{eq:approx_block_encoding_filter_poly}
        \|(\bra{0}_{\alpha\beta_1}\otimes I_{\beta_2}){\mathcal{V}}_{\mathrm{circ}}(\ket{0}_{\alpha\beta_1}\otimes I_{\beta_2})-R_{\ell}(A/\alpha;\delta)\|\leq  \varsigma.
    \end{equation}
    
    By \cref{lem:subspace_eigenstate_filtering_exact}, we have
    \[
    \|R_{\ell}(A/\alpha;\delta)\Pi_{\mathcal{S}} - \Pi_{0}\Pi_{\mathcal{S}}\|\leq 2e^{-\sqrt{2}\ell \delta}.
    \]
    Combining the above inequality with \eqref{eq:approx_block_encoding_filter_poly} via the triangle inequality, we then have
    \begin{equation}
        \|(\bra{0}_{\alpha\beta_1}\otimes I_{\beta_2}){\mathcal{V}}_{\mathrm{circ}}(\ket{0}_{\alpha\beta_1}\otimes I_{\beta_2})\Pi_{\mathcal{S}}-\Pi_{0}\Pi_{\mathcal{S}}\|\leq  \varsigma+2e^{-\sqrt{2}\ell \delta}.
    \end{equation}
    In order to make $2e^{-\sqrt{2}\ell \delta}\leq \epsilon$, it suffices to choose
    \[
    \ell = \Or\left(\frac{1}{\delta}\log(1/\epsilon)\right) = \Or\left(\frac{\alpha}{\Delta}\log(1/\epsilon)\right).
    \]
    We therefore have \eqref{eq:robust_subspace_eigenstate_filtering}.
\end{proof}

\section{Spectrum estimates}
\label{sec:spectrum_estimates}

We will use the following lemma to upper bound the matrix spectral norm.

\begin{lemma}
\label{lem:bound_spectral_norm_using_inf_norm}
    Let $A = (A_{ij})_{n\times n}$ be a Hermitian matrix. Then
    \[
    \|A\|\leq \max_{\ket{x}:\|\ket{x}\|_\infty\leq 1} \|A\ket{x}\|_{\infty} = \max_{1\leq i\leq n} \sum_{j=1}^n |A_{ij}|,
    \]
    where $\|\cdot\|_{\infty}$ denotes the infinity norm. 
\end{lemma}

This result in fact holds for any general matrix. To prove that we only need to dilate the matrix $A$ to be $\begin{pmatrix}
    0 & A \\
    A^\dag & 0
\end{pmatrix}$. We will not discuss the detail because we are only going to use the Hermitian version.

\begin{proof}
    Because $A$ is Hermitian, either $\|A\|$ or $-\|A\|$ must be an eigenvalue of $A$. Therefore there exists $\ket{\phi}$ such that $A\ket{\phi}=\pm \|A\|\ket{\phi}$ and $\|\ket{\phi}\|_{\infty}=1$. 
    Consequently $\|A\ket{\phi}\|_\infty=\|A\|$, and this proves the  inequality. The equality can be easily checked and we omit the proof.
\end{proof}

We will then prove \cref{lem:spectral_gap_D1} in the main text, which we restate below:


\begin{lemma-non}
    Let \begin{equation}
\label{eq:defn_D1}
    D_1 = 
    \begin{pmatrix}
         0     & t_1  &        &            &     \\
        t_1  & 0     & t_2   &            &     \\
              & t_2  & 0      & \ddots     &     \\
              &       & \ddots & \ddots     & t_{m-1}     \\
              &       &        & t_{m-1}  & 0
    \end{pmatrix}_{ m\times m}
\end{equation} as defined in \eqref{eq:defn_D1},
    in which $t_1=\sqrt{d-2}$, $t_2=\cdots =t_{m-1}=\sqrt{d-1}$.
    Then $D_1$ has a non-degenerate 0-eigenstate $\ket{\Psi}=(\Psi_1,\Psi_2,\cdots,\Psi_m)^\top$ where
    \[
    \Psi_j = \prod_{k=1}^{(j-1)/2}\left(-\frac{t_{2k-1}}{t_{2k}}\right)\Psi_1 = (-1)^{(j-1)/2}\sqrt{\frac{d-2}{d-1}}\Psi_1,\quad \text{for all odd }j\geq 2,
    \]
    and $\Psi_j=0$ for even $j$.
    Moreover, 0 is separated from the rest of the spectrum of $D_1$ by a gap of at least $2\sqrt{d-2}/(m-1)$.
\end{lemma-non}

\begin{proof}
    First, we prove that 0 is an eigenvalue of $D_1$. The can be done by verifying $D_1\ket{\Psi}=0$ for the $\ket{\Psi}$ given above. 

    We then prove the non-degeneracy of the eigenvalue 0 and the spectral gap. We do this through the eigenvalue interlacing theorem, which tells us that the eigenvalues of $D_1$ must interlace those of $D_1'$, where $D_1'$ is the $(m-1)\times (m-1)$ sub-matrix of $D_1$ on the upper-left corner, i.e.,
    \begin{equation}
        D_1' = 
        \begin{pmatrix}
             0     & t_1  &        &            &     \\
            t_1  & 0     & t_2   &            &     \\
                  & t_2  & 0      & \ddots     &     \\
                  &       & \ddots & \ddots     & t_{m-2}     \\
                  &       &        & t_{m-2}  & 0
        \end{pmatrix}_{ (m-1)\times (m-1)}.
    \end{equation}
    We will therefore first study the spectrum of $D_1'$.

    We observe that $D_1'$ is in fact an invertible matrix. Consider the linear system $D_1'\ket{x}=\ket{y}$, where $\ket{x}=(x_1,\cdots,x_{m-1})^\top$ and $\ket{y}=(y_1,\cdots,y_{m-1})^\top$, then we can compute the solution $\ket{x}$ using the following recursion:
    \begin{equation}
        \begin{aligned}
            x_{2k+2} &= \frac{y_{2k+1}}{t_{2k+1}} - \frac{t_{2k}}{t_{2k+1}}x_{2k},\quad \forall k\geq 1, \\
            x_2 &= \frac{y_1}{t_1}, \\
            x_{m-2k} &= \frac{y_{m-2k+1}}{t_{m-2k}} - \frac{t_{m-2k+1}}{t_{m-2k}}x_{m-2k+2}, \quad \forall k\geq 1 \\
            x_{m-2} &= \frac{y_{m-1}}{t_{m-2}}.
        \end{aligned}
    \end{equation}
    Given the values of $t_j$, we observe that 
    \[
    \frac{t_{2k}}{t_{2k+1}}=1,\quad \frac{t_{m-2k+1}}{t_{m-2k}} = \begin{cases}
        1,&\text{ if }k<(m-1)/2,\\
        \sqrt{\frac{d-1}{d-2}},&\text{ if }k=(m-1)/2.\\
    \end{cases}
    \]
    Therefore, for even entries we have
    \[
    |x_{2k+2}|\leq \left|\frac{y_{2k+1}}{t_{2k+1}}\right| + |x_{2k}|\leq \frac{(k+1)\|\ket{y}\|_{\infty}}{\sqrt{d-2}},
    \]
    For odd entries we have
    \[
    |x_{m-2k}|\leq \left|\frac{y_{m-2k+1}}{t_{m-2k}}\right|+|x_{m-2k+2}|\leq \frac{k\|\ket{y}\|_{\infty}}{\sqrt{d-1}},
    \]
    for $k<(m-1)/2$. For $k=(m-1)/2$, we have
    \[
    |x_1|\leq \left|\frac{y_{2}}{t_{1}}\right|+\sqrt{\frac{d-1}{d-2}}|x_{3}|\leq \frac{\|\ket{y}\|_{\infty}}{\sqrt{d-2}}+\frac{\|\ket{y}\|_{\infty}(m-3)/2}{\sqrt{d-2}} = \frac{\|\ket{y}\|_{\infty}(m-1)/2}{\sqrt{d-2}}.
    \]
    Therefore we have
    \[
    |x_j|\leq \frac{(m-1)\|\ket{y}\|_{\infty}}{2\sqrt{d-2}}
    \]
    for all $j$.
    From this we can see that if $\|\ket{y}\|_\infty\leq 1$, then 
    \[
    \|(D_1')^{-1}\ket{y}\|_\infty = \max_j |x_j|\leq \frac{(m-1)}{2\sqrt{d-2}}.
    \]
    Therefore by \cref{lem:bound_spectral_norm_using_inf_norm} we have $\|(D_1')^{-1}\|\leq (m-1)/(2\sqrt{d-2})$. As a result the eigenvalues of $D_1'$ must be bounded away from 0 by at least $2\sqrt{d-2}/(m-1)$.

    Because by the eigenvalue interlacing theorem the eigenvalues of $D_1$ interlace those of $D_1'$, if $0$ is a degenerate eigenvalue of $D_1$, then there must exists an eigenvalue of $D_1'$ between two 0's, i.e., this eigenvalue of $D_1'$ must also be 0. This is impossible because we have just shown that $D_1'$ is invertible. This proves the non-degeneracy of 0 as an eigenvalue of $D_1$. If there is an eigenvalue $\lambda$ of $D_1$ such that $|\lambda|<2\sqrt{d-2}/(m-1)$, then there must exist an eigenvalue $\lambda'$ of $D_1'$ between 0 and $\lambda$, and therefore $|\lambda'|<2\sqrt{d-2}/(m-1)$. This is again impossible because we have just shown that all eigenvalues of $D_1'$ must be bounded away from 0 by at least $2\sqrt{d-2}/(m-1)$. Therefore all non-zero eigenvalues of $D_1$ must be bounded away from 0 by at least $2\sqrt{d-2}/(m-1)$.
\end{proof}

\begin{lemma}[Inverse of a nonsingular tridiagonal matrix {\cite[Theorem 2.1]{da2001explicit}}] \label{lem:inverseH} 
Let $H_1(a,b)$ be as defined in \eqref{eq:H1ab}.
Let $\sigma$, $\delta$ be the two $m$ dimensional vectors defined as follows:
\begin{enumerate}
    \item $\sigma_m=b$, $\sigma_i=-t_i^2/\sigma_{i+1}$ for $i=m-1, \ldots, 2$, $\sigma_1=a-t_1^2/\sigma_2$.
    \item $\delta_1=a$, $\delta_i=-t_{i-1}^2/\delta_{i-1}$ for $i=2, \ldots, m-1$, $\delta_m=b-t_{m-1}^2/\sigma_{m-1}$.
\end{enumerate}

Then the matrix element of the inverse of $H_1(a,b))$ is 

\begin{equation} \label{eq:inverseHentry}
    (H_1^{-1}(a,b))_{i,j} = \left\{\begin{array}{ll}
  (-1)^{i+j} t_i\cdots t_{j-1} \frac{\sigma_{j+1}\cdots \sigma_m}{\delta_i\cdots\delta_{m}} & \mbox{if }  i\leq j \\
 (-1)^{i+j} t_j\cdots t_{i-1} \frac{\sigma_{i+1}\cdots \sigma_m}{\delta_j\cdots\delta_{m}}  & \mbox{if }  i>j \\
\end{array}\right.,  \end{equation} 

with the convention that the empty product equals 1.\\


\end{lemma}
\begin{corollary}\label{cora:maxinverselement}
Let $t_1=\sqrt{d-2}$, $t_2=t_3=\dots=t_{m-1}=\sqrt{d-1}$, $\gamma=\frac{d-1}{2}$, for integer $d\geq 3$, and let $m$ be an odd integer. Let $\mu_l = 2\cos(2\pi l/n)$ with $l=0,1\ldots, n-1$ satisfying $\mu_l\neq 0$. Let $a=\mu_l$ and $b=\gamma\mu_l$. Then we have 
\begin{equation}
  \max_{i,j}  (H_1^{-1}(a,b))_{i,j}= O(1/|a|^2)= O(n^2).    \end{equation}
\end{corollary}
\begin{proof} 

Since $t_1=\sqrt{d-2}$, $t_2=t_3=\ldots=t_{m-1}=\sqrt{d-1}$, $\gamma=\frac{d-1}{2}$ and $m$ is an odd integer, we have $t_1t_{2}=\sqrt{(d-2)(d-1)}$,  $t_it_{i+1}=(d-1)$ for $ 2\leq i \leq m-1$, $\sigma_1=a\left(1+\frac{\gamma(d-1)}{d-2}\right)$, $|\sigma_i\sigma_{i+1}|=(d-1)$ for $i=2,\ldots, m-1$ and $\delta_1=a$, $|\delta_i\delta_{i+1}|=(d-1)$ for $2\leq i\leq m-2$, $\delta_m=a(1+\frac{\gamma(d-1)}{d-2})$. We can then compute the two $m$ dimensional vectors $\sigma,\delta$ in \cref{lem:inverseH} to be:
\begin{enumerate} 
    \item $\sigma_{i}=-\frac{(d-1)}{b}, \sigma_{i+1}=b,$ for $i$ to be an even integer from $m-1$ to $2$, $\sigma_1=a+ \frac{b(d-1)}{(d-2)}$.
    \item $\delta_1=a$, $\delta_{i}=-(d-2)/a$,  $\delta_{i+1}=\frac{a(d-1)}{(d-2)}$ for $i$ to be an even integer from $2$ to $m-1$, $\delta_m=b+ \frac{a(d-1)}{(d-2)}$.
\end{enumerate}

Using the results above we will compute upper bounds for the matrix entries $(H_1^{-1}(a,b))_{i,j}$. Note that because $H_1(a,b)$ is Hermitian, we only need to consider the case of $i\leq j$. By \cref{lem:inverseH} we have
\begin{equation}
    \label{eq:product_into_three_parts}
    |(H_1^{-1}(a,b))_{i,j}| = \frac{t_i t_{i+1}\cdots t_{j-1}}{\delta_i \delta_{i+1}\cdots \delta_{j-1}}\cdot \frac{1}{\delta_j}\cdot \frac{\sigma_{j+1} \sigma_{j+2}\cdots \sigma_{m}}{\delta_{j+1} \delta_{j+2}\cdots \delta_{m}}.
\end{equation}
We will next deal with the three parts on the right-hand side separately.

For the first part $\frac{t_i t_{i+1}\cdots t_{j-1}}{\delta_i \delta_{i+1}\cdots \delta_{j-1}}$, we consider two different cases. If $j-i$ is even, then
\begin{equation}
\begin{aligned}
    \left|\frac{t_i t_{i+1}\cdots t_{j-1}}{\delta_i \delta_{i+1}\cdots \delta_{j-1}}\right| &= \left|\frac{(t_it_{i+1})(t_{i+2}t_{i+3})\cdots(t_{j-2}t_{j-1})}{(\delta_i\delta_{i+1})(\delta_{i+2}\delta_{i+3})\cdots(\delta_{j-2}\delta_{j-1})}\right| \\
    &=\left|\frac{t_it_{i+1}}{t_i^2}\frac{t_{i+2}t_{i+3}}{t_{i+2}^2}\cdots\frac{t_{j-2}t_{j-1}}{t_{j-2}^2}\right| \\
    &=\left|\frac{t_{i+1}}{t_i}\frac{t_{i+3}}{t_{i+2}}\cdots\frac{t_{j-1}}{t_{j-2}}\right| \\
    &= \begin{cases}
        \sqrt{\frac{d-1}{d-2}} & \text{ if } i=1, \\
        1 & \text{ if } i\geq 2.
    \end{cases}
\end{aligned}
\end{equation}
Therefore 
\begin{equation}
\label{eq:part_one_i_minus_j_even}
     \left|\frac{t_i t_{i+1}\cdots t_{j-1}}{\delta_i \delta_{i+1}\cdots \delta_{j-1}}\right|\leq \sqrt{\frac{d-1}{d-2}},
\end{equation}
when $j-i$ is even.

When $j-i$ is odd, then 
\[
 \left|\frac{t_i t_{i+1}\cdots t_{j-1}}{\delta_i \delta_{i+1}\cdots \delta_{j-1}}\right|=  \left|\frac{t_i t_{i+1}\cdots t_{j-2}}{\delta_i \delta_{i+1}\cdots \delta_{j-2}}\right|\left|\frac{t_{j-1}}{\delta_{j-1}}\right|\leq \sqrt{\frac{d-1}{d-2}}\left|\frac{t_{j-1}}{\delta_{j-1}}\right|,
\]
where we have used \eqref{eq:part_one_i_minus_j_even}. Note that
\begin{equation}
    \left|\frac{t_{j-1}}{\delta_{j-1}}\right|= 
    \begin{cases}
        \frac{a\sqrt{d-1}}{d-2},&\text{ if } j\text{ is odd,} \\ 
        \frac{\sqrt{d-2}}{a},&\text{ if } j=2, \\
        \frac{d-2}{a\sqrt{d-1}},&\text{ otherwise.}
    \end{cases}
\end{equation}
Therefore
\begin{equation}
    \left|\frac{t_i t_{i+1}\cdots t_{j-1}}{\delta_i \delta_{i+1}\cdots \delta_{j-1}}\right|\leq \max\left\{\frac{a (d-2)}{(d-2)^{3/2}},\frac{\sqrt{d-1}}{a}\right\},
\end{equation}
when $j-i$ is odd. Combining the above inequality with \eqref{eq:part_one_i_minus_j_even}, we have
\begin{equation}
\label{eq:part_one_final}
    \left|\frac{t_i t_{i+1}\cdots t_{j-1}}{\delta_i \delta_{i+1}\cdots \delta_{j-1}}\right| = \Or(1/a),
\end{equation}
for all pairs of $i\leq j$.

For the third part $\frac{\sigma_{j+1} \sigma_{j+2}\cdots \sigma_{m}}{\delta_{j+1} \delta_{j+2}\cdots \delta_{m}}$ in \eqref{eq:product_into_three_parts}, when $m-j$ is even, which implies that $j$ is odd, we use the fact that $\sigma_l\sigma_{l+1}=\delta_l\delta_{l+1}$ for $l\geq 2$ to show that 
\[
\frac{\sigma_{j+1} \sigma_{j+2}\cdots \sigma_{m}}{\delta_{j+1} \delta_{j+2}\cdots \delta_{m}} = 1.
\]
When $m-j$ is odd, which implies that $j$ is even, 
\[
\frac{\sigma_{j+1} \sigma_{j+2}\cdots \sigma_{m}}{\delta_{j+1} \delta_{j+2}\cdots \delta_{m}} = \frac{\sigma_{j+1}}{\delta_{j+1}}\frac{\sigma_{j+2} \sigma_{j+2}\cdots \sigma_{m}}{\delta_{j+2} \delta_{j+2}\cdots \delta_{m}} = \frac{\sigma_{j+1}}{\delta_{j+1}}.
\]
Note that
\[
\left|\frac{\sigma_{j+1}}{\delta_{j+1}}\right|=
\begin{cases}
    \frac{b(d-2)}{a(d-1)}=\frac{\gamma(d-2)}{d-1},&\text{ if }j<m-1,\\
    \frac{b(d-2)}{b(d-2)+a(d-a)}\leq 1,&\text{ if }j=m-1.
\end{cases}
\]
Therefore $\left|\frac{\sigma_{j+1}}{\delta_{j+1}}\right|=\Or(1)$. Consequently 
\begin{equation}
    \label{eq:part_three_final}
    \left|\frac{\sigma_{j+1} \sigma_{j+2}\cdots \sigma_{m}}{\delta_{j+1} \delta_{j+2}\cdots \delta_{m}}\right| = \Or(1).
\end{equation}

For the second part in \eqref{eq:product_into_three_parts}, we readily have $1/|\delta_j|=\Or(1/a)$. Combining this with \eqref{eq:part_one_final} and \eqref{eq:part_three_final}, we have through \eqref{eq:product_into_three_parts},
\begin{equation}
    |(H_1^{-1}(a,b))_{i,j}| = \Or(1/|a|^2).
\end{equation}
Note that because $a=2\cos(2l\pi/n)$, the smallest possible $|a|$ is 
\[
|a| = 2\sin(2\pi/n)=\Omega(1/n).
\]
Therefore we have the result as stated in this corollary.
\end{proof}

We will then prove the \cref{lem:inverseH_spectral_norm} in the main text, which we restate here:

\begin{lemma-non}
    Let $H_1(a,b)$ be as defined in \eqref{eq:H1ab}, with odd $m$, $t_1=\sqrt{d-2}$, $t_2=t_3=\dots=t_{m-1}=\sqrt{d-1}$, $\gamma=\frac{d-1}{2}$, for integer $d\geq 3$. Let $\mu_l = 2\cos(2\pi l/n)$ with $l=0,1\ldots, n-1$ satisfying $\mu_l\neq 0$. Let $a=\mu_l$ and $b=\gamma\mu_l$. Then $\|H_1^{-1}(a,b)\|=\Or(mn^2)$.
\end{lemma-non}

\begin{proof}
    This lemma is directly proved by combining \cref{cora:maxinverselement} with \cref{lem:bound_spectral_norm_using_inf_norm}.
\end{proof}

\end{document}

%% file: premable.tex
\pdfoutput=1

 \usepackage{tikz}
\usetikzlibrary{trees}
\usepackage[english]{babel}
\usepackage[utf8]{inputenc}
\usepackage[T1]{fontenc}
\usepackage[pdftex, pdftitle={Article}, pdfauthor={Author}]{hyperref} 
\usepackage[margin=1in]{geometry} 

\usetikzlibrary {positioning}
\definecolor {processblue}{cmyk}{0.96,0,0,0}

\usepackage{comment}

\usepackage{fancyhdr}
\usepackage{braket}
\usepackage{amssymb,bbm}
\usepackage{amsmath}
\usepackage{mathtools}

\usepackage{latexsym}
\usepackage{amsthm}
\usepackage[capitalise]{cleveref}

\usepackage{comment}
\usepackage{mathdots}
\usepackage{url}
\usepackage{algorithm}
\usepackage{algorithmic}
\usepackage{appendix}

\newtheorem{theorem}{Theorem}[section]
\newtheorem{lemma}[theorem]{Lemma}
\newtheorem*{lemma-non}{Lemma}
\newtheorem*{theorem-non}{Theorem}

\newtheorem{corollary}[theorem]{Corollary}
\newtheorem{definition}[theorem]{Definition}

\newtheorem{problem}[theorem]{Problem}

\newtheorem{example}[theorem]{Example}

\theoremstyle{definition}
{

}

\mathchardef\mhyphen="2D

\newcommand{\R}{\mathbb{R}}

\def\({\left(}
\def\){\right)}

\usepackage{graphicx}
\definecolor{greenn}{rgb}{0,0.8,0.2}
\definecolor{bluue}{rgb}{0.3,0,0.7}

\usepackage{ifdraft}
\ifdraft{\newcommand{\authnote}[3]{{\color{#3}{\text{ \bf #1:}} #2}}}{\newcommand{\authnote}[3]{}}

\newcommand{\Or}{\mathcal{O}}

\newcommand\polylog[1]{\mathrm{polylog}\left( #1 \right)}
\newcommand{\poly}{\operatorname{poly}}
\newcommand{\N}{\Gamma}

\usepackage{authblk}